\documentclass{article}
\usepackage{graphicx} 
\usepackage{fullpage}
\usepackage{amsmath}
\usepackage{amsfonts}
\usepackage{amssymb}
\usepackage{physics}
\usepackage{xcolor}
\usepackage{amsthm}
\usepackage{hyperref}
\hypersetup{hypertexnames=false,hidelinks}
\usepackage{comment}
\usepackage{algorithm}
\usepackage{algpseudocode}
\usepackage{bbm}
\usepackage[shortlabels]{enumitem}
\usepackage{authblk}

\usepackage{tikz}
\usetikzlibrary{arrows.meta,positioning,shapes.geometric,fit,calc}
\usetikzlibrary{positioning,fit,arrows.meta,backgrounds}

\newcommand{\cL}{\mathcal{L}}
\newcommand{\cS}{\mathcal{S}}

\newcommand{\cP}{\mathcal{P}}
\newcommand{\cO}{\mathcal{O}}
\newcommand{\cH}{\mathcal{H}}
\newcommand{\cF}{\mathcal{F}}
\newcommand{\cA}{\mathcal{A}}
\newcommand{\cMn}{\mathcal{M}_{2^n}}
\newcommand{\C}{\mathbb{C}}

\newcommand{\poly}[1]{\mathrm{poly}(#1)}
\newcommand{\polylog}[1]{\mathrm{polylog}(#1)}
\newcommand{\supp}{\mathrm{supp}}

\newtheorem{thm}{Theorem}[section]
\newtheorem{prop}{Proposition}[section]

\newtheorem{lemma}{Lemma}[section]
\newtheorem{definition}{Definition}[section]
\newtheorem{corollary}{Corollary}[section]
\newtheorem{problem}{Problem}[section]
\newtheorem{remark}{Remark}[section]

\title{Learning and certification of local \\time-dependent quantum dynamics and noise}

\date{}

\author[1]{Daniel Stilck Fran\c{c}a\thanks{Contact author: \href{mailto:dsfranca@math.ku.dk}{dsfranca@math.ku.dk}}}
\author[4]{Cambyse Rouzé}
\author[1]{Albert H. Werner}
\author[2,3]{Tim M\"{o}bus}

\affil[1]{\small Department of Mathematical Sciences, University of Copenhagen, Denmark}
\affil[2]{\small Department of Applied Mathematics and Theoretical Physics, University of Cambridge, United Kingdom}
\affil[3]{\small Department of Mathematics, University of Tübingen, Germany}
\affil[4]{\small Inria, Télécom Paris - LTCI, Institut Polytechnique de Paris, France}

\begin{document}

\maketitle
\begin{abstract}
Hamiltonian learning protocols are becoming valuable tools for benchmarking and verifying quantum computers and simulators, yet few rigorous protocols address explicitly time-dependent dynamics. We give a scalable protocol for learning the time-dependent coefficients of a locally interacting $n$-qubit system on a graph $\mathsf{G}$ of effective dimension $D$. The unknown generator contains geometrically local Hamiltonian terms and weak onsite Pauli-dephasing terms with time-dependent, nonnegative rates. Its coefficient functions and their derivatives admit a certified simultaneous $C^1$ approximation in a known finite-dimensional function class that supports stable interpolation. The only experimental primitives are preparation of product Pauli eigenstates, evolution for chosen times, and product-Pauli measurements. Under an explicit weak-dissipation condition, the Lieb--Robinson estimates proved here, and a certified effective approximation/interpolation degree $m_{\rm eff}$, the protocol reconstructs every coefficient uniformly on a constant interval with error $\epsilon$ and success probability $1-\delta$ using
\[
\widetilde{\cO}\!\left(\epsilon^{-2}\poly{m_{\rm eff}}\log(n\delta^{-1})\right)
\]
samples and polynomial classical processing. The method estimates derivatives of finite-time channel overlaps by stable interpolation, uses process shadows to learn the required local transfer matrices, and constructs sparse diagonally dominant systems through localized semidefinite programs.
\end{abstract}
\newpage

\section{Introduction}

Quantum simulators and computers are rapidly growing in scale and quality, and with that growth comes a demand for methods that can benchmark, verify, and extract useful dynamical information in a way that remains practical as the system size increases. What one would like, ideally, are protocols whose sample cost grows at most polylogarithmically with the number of qubits, whose classical postprocessing is polynomial, and with minimal requirements on the operations we need to perform on the quantum computer or simulator. Over the last few years, Hamiltonian learning and shadow-based techniques have crystallized into an appealing answer to this demand for the case of \emph{time-independent dynamics}~\cite{tang_structure,StilckFranca.2024,PhysRevLett.130.200403,zubida2021optimalshorttimemeasurementshamiltonian,PhysRevLett.122.020504,10.1145/3670418,Flynn2022,Gu2024,Holzpfel2015,li2023heisenberglimitedhamiltonianlearninginteracting,Moebus.2023,Moebus.2025,PhysRevLett.107.210404,PhysRevLett.112.190501}: we now have a flurry of different protocols for this task assuming various degrees of control over the quantum device or assumptions on the structure of the evolution~\cite{Dutkiewicz2024}, leading to essentially optimal protocols~\cite{tang_structure} for purely Hamiltonian dynamics.

The dynamical tasks we care to certify or benchmark on present-day simulators, however, are often time-dependent. Adiabatic state preparation~\cite{RevModPhys.90.015002}, quenches and ramps~\cite{RevModPhys.83.863,Dziarmaga2010}, Floquet or periodic driven engineering~\cite{Goldman2014}, dynamical decoupling~\cite{Viola1999} and pulse sequences~\cite{Krantz_2019} all generate explicitly time-dependent dynamics; so do many realistic noise sources~\cite{Krantz_2019}, such as drifts in local fields~\cite{Burnett2019} and periodic control modulations~\cite{PhysRevApplied.22.054065}. Despite this ubiquity, essentially all scalable and rigorous learning guarantees to date treat the time-independent or weakly dependent case~\cite{PRXQuantum.3.020357}. We introduce and analyze a protocol that learns local time-dependent generators on a bounded-degree graph under an experimentally minimal access model. Our theorem covers geometrically local Hamiltonians together with weak onsite Pauli dephasing. The only primitives we require are preparation of random single-qubit Pauli eigenstates, forward evolution under the unknown dynamics for chosen times, and measurements in product Pauli bases.

Our starting point is a quantitative smoothness model for the time dependence: each coefficient and its derivative admit a certified simultaneous $C^1$ approximation in a finite-dimensional function class $\cF_{m_\epsilon}$ (for instance, degree-$m_\epsilon$ polynomials) that admits numerically stable interpolation. Under this model, for graphs of polynomial growth, $k$-local Hamiltonian terms, weak onsite Pauli dephasing, and constant time $T=\cO(1)$, we reconstruct all coefficients uniformly on $[0,T]$ up to error $\epsilon$ with success probability $1-\delta$, using
\[
S=\widetilde{\cO}\big(\epsilon^{-2}\,\poly{m_{\rm eff}}\,\log(n\delta^{-1})\,\big)
\]
samples and
$\widetilde{\cO}(\epsilon^{-2}\poly{n,m_{\rm eff},\log(\delta^{-1})})$
classical computation. Thus, the logarithmic dependence on system size familiar from time-independent learning~\cite{tang_structure,StilckFranca.2024} persists, and the extra price of time dependence appears only as a polynomial in the certified effective degree. For exponentially convergent $C^1$ approximations, $m_{\rm eff}=\polylog{\epsilon^{-1}}$ up to fixed regularity parameters.

At a high level, our protocol is inspired by that of~\cite{StilckFranca.2024}, which combined robust polynomial interpolation~\cite{Kane2017-dn} with Lieb--Robinson bounds to obtain stable linear systems for time-independent parameters from first derivatives at the origin. In the time-dependent setting we instead need derivatives at finite times. Estimating many higher derivatives at the origin would incur rapidly growing Markov constants and generally produce nonlinear, ill-conditioned systems. Our main technical step is to evaluate the coefficient functions at finite, constant times with only a polynomial overhead in the effective degree $m_{\rm eff}$.

To achieve that, the protocol exploits locality twice. First,
local-propagation estimates of Lieb--Robinson type~\cite{lieb1972finite,BarthelKliesch2012,Cubitt2015,chen2021operator,chen2023speed}
localize the inverse of the unitary Hamiltonian component on local probe
operators; a perturbative comparison then transfers these probes to the
weakly dephasing forward channel. Lieb-Robinson bounds show that the inverse Hamiltonian evolution of a local Pauli can be approximated by a local operator, thereby establishing the existence of local inverse probes. Under sufficiently weak onsite dephasing, we prove that the same probes remain close to their target Paulis under the full forward channel.
This lets us combine process shadows~\cite{StilckFranca.2024} with semidefinite programs to find local probes that the channel maps close to prescribed Paulis.

Second, at each sampled time, the unknown values of the coefficient functions form the solution vector of a sparse, diagonally dominant linear system. This gives $\ell_\infty$ stability against statistical error in the local transfer matrices and interpolation error in the derivative estimates. The locality and comparison estimates used below follow from the
Heisenberg-picture Lieb--Robinson theorem proved in
Appendix~\ref{app:LR}; their specialization to the bistochastic learning
evolution is Proposition~\ref{prop:locality_bounds}.

Beyond reconstruction, the learned model can be used to compare local expectation values generated by two schedules. We formulate this consequence directly for the physical Heisenberg channel using the same proved locality theorem. The protocol can also track weak time-dependent Pauli-dephasing rates during gates, pulses, and annealing schedules.

Because the time-dependent learning problem has been largely unexplored, several natural extensions present themselves. General non-unital onsite dissipation and multiqubit dissipative interactions require new stable-probe constructions and are open problems. Other directions include algebraically decaying interactions, structure learning, and Heisenberg-limited scaling in $\epsilon$ using entanglement or adaptivity.

In short, we provide an experimentally friendly, provably efficient protocol for learning local time-dependent generators on quantum devices. It matches the best known scaling in system size for the time-independent case, adds only a polynomial overhead in the expressive power required to represent the time dependence, and directly supports benchmarking tasks that practitioners care about. We believe this closes an important gap in the literature and opens a promising direction for rigorous, scalable certification of dynamical protocols on near-term quantum simulators and processors.

\section{Main result}

\subsection{Setup and basic notions}

\paragraph{Forward time-dependent Lindbladians in the Schrödinger picture:}

We start by introducing the unknown processes we aim at learning.
Given $k\in\mathbb{N}$, let $\mathcal{M}_{k}$ denote the set of linear operators on $\mathbb{C}^k$, write $I=I_k$ the identity operator and $\mathcal{I}$ the identity superoperator. Given a matrix $O\in\mathcal{M}_k$, we denote by $O^\dagger$ its adjoint, and given a map $\Phi:\mathcal{M}_k\to\mathcal{M}_k$, we denote by $\Phi^*$ its adjoint with respect to the Hilbert Schmidt inner product $\langle O,O'\rangle:=\tr(O^\dagger O')$.
A \emph{time-dependent Schrödinger-picture Lindbladian} over an $n$-qubit system is a family of linear maps
\[
  \mathcal{S}(t) : \mathcal{M}_{2^n} \to \mathcal{M}_{2^n}, \quad t \in [0,T],
\]
of the Gorini–Kossakowski–Sudarshan–Lindblad (GKSL)~\cite{Breuer2007} form
\begin{align}
  \label{eq:schrodinger-gksl}
  \mathcal{S}(t)(X)
  &= -i[H(t), X]
     + \sum_{\mu} \left( L_\mu(t) X L_\mu(t)^\dagger
     - \frac12 \{ L_\mu(t)^\dagger L_\mu(t), X \} \right),
\end{align}
where $H(t)$ is Hermitian and $\{L_\mu(t)\}$ are the jump operators.  
The forward Schrödinger propagator from $s$ to $t$ is the linear map $T(s,t)$ defined for $s\le t$ by
\begin{align}
  \partial_t T(s,t) = \mathcal{S}(t) \circ T(s,t), \quad T(s,s) = \mathcal{I}.
\end{align}
Its adjoint is the physical Heisenberg evolution and obeys
\begin{align}
 \partial_t T(s,t)^*=T(s,t)^*\circ\mathcal S(t)^*.
\end{align}
It is applied physically to density operators and extended linearly to all
matrices.  In the theorem below, each local generator term is
trace-annihilating and vanishes on the identity, and every truncated sum used
in the proof generates a completely positive, trace-preserving, unital
evolution.  Thus $T(s,t)$ is a bistochastic channel and is contractive in
operator norm. General Schrödinger Lindbladians need not be unital and are not
covered by our learning theorem.
The solution of this differential equation can be expressed via the \emph{Dyson expansion}: 
\begin{align}
  \label{eq:forward-dyson}
  T(s,t) 
  &= \sum_{j=0}^\infty  
     \int_s^t ds_j \int_s^{s_j} ds_{j-1} \cdots \int_s^{s_{2}} ds_1 \,
        \mathcal{S}(s_j) \cdots \mathcal{S}(s_1),
\end{align}
For simplicity, we use $\mathcal{S}^j(s_j,\dots,s_1)$ for the product
$\mathcal{S}(s_j)\cdots\mathcal{S}(s_1)$ and write
$T(s,t)=T_{\mathcal S}(s,t)$ when the generator needs to be explicit.

\paragraph{Pauli strings:} Our unknown generators are expressed through a standard Pauli decomposition: for any $(x,z) \in \{0,1\}^{2n}$, we define the \emph{Pauli string} $P_{(x,z)}\in\mathcal{M}_{2^n}$ as follows:
\begin{align}
  P_{(x,z)} := \bigotimes_{j=1}^n i^{x_j z_j} X_j^{x_j} Z_j^{z_j},
\end{align}
where $X_j$ and $Z_j$ are the usual Pauli matrices acting on site $j$. We will also at times use the notation $P_\alpha$ for $\alpha\in\{0,1\}^{2n}$. The set $\{P_{(x,z)}\}$ forms an orthogonal basis of $\mathcal{M}_{2^n}$ with respect to the Hilbert--Schmidt inner product (and an orthonormal basis after normalization by $2^{-n}$):
\[
  \mathrm{tr}\!\left( P_{(x,z)} P_{(x',z')} \right) = 2^n \delta_{x,x'} \delta_{z,z'}.
\]
They form a projective representation of $\{0,1\}^{2n}$ with
\begin{align}
  P_{(x,z)} P_{(x',z')} 
  &= i^{x'\cdot z-x\cdot z'} P_{(x \oplus x',\, z \oplus z')}\,.
\end{align}
Hence, the commutator is
\begin{align}
  [P_{(x,z)}, P_{(x',z')}] 
  &= 2i^{x'z-xz'} \, \delta_{\mathrm{ac}} \, P_{(x \oplus x',\, z \oplus z')},
  \label{eq:pauli_comm}
\end{align}
where $\delta_{\mathrm{ac}} = 1$ if the strings anticommute and $0$ if they commute.  
Equivalently, $\delta_{\mathrm{ac}} = (x \cdot z' - z \cdot x') \bmod 2$.
The \emph{weight} of a Pauli string $P_{(x,z)}$ is $w(P_{(x,z)})=|\{j : x_j = 1 \text{ or } z_j = 1\}|$, i.e., the number of qubits on which it acts nontrivially.

\paragraph{Parametrizing the generators:} We learn a restricted class of bistochastic processes defined over graphs. Given a graph $\mathsf{G}=(V,E)$ of size $|V|=n$, we consider geometrically $k$-local Hamiltonian generator terms $-ih_{(x,z)}\cP_{(x,z)}$, where $(x,z)\ne(0,0)$ and $\cP_{(x,z)}:\cMn\to\cMn$ is
\begin{align}
\cP_{(x,z)}(O)=\frac{1}{2}[P_{(x,z)},O],\qquad\text{ with unknown coefficients }\qquad  h_{(x,z)}\in\mathbb{R},\quad |h_{(x,z)}|\le 1.
\end{align}
We further impose the physically relevant condition that these terms are geometrically $k$-local for some fixed integer $k\ge1$, with $k=\mathcal{O}(1)$ independent of the lattice size $n$. This means that for any $(x,z)\in\{0,1\}^{2n}$, 
$$\operatorname{diam}(\{j\in V| x_j\not=0  \operatorname{ or } z_j\ne 0\})>k\Longrightarrow h_{(x,z)}=0,$$
where the diameter is taken with respect to the underlying graph distance on $\mathsf{G}$.
The identity coefficient is gauge-fixed to
$h_{(0,0)}(t)\equiv0$: it multiplies the zero commutator and is therefore
unidentifiable.  We take $\mathcal A$ to be a known set of distinct
nonidentity Pauli strings; coefficients outside $\mathcal A$ are assumed to be zero.
Given the known graph and locality range, one may choose $\mathcal A$ to
contain all nonidentity Pauli strings supported on permitted local regions.
This gives $|\mathcal A|\lesssim n\,4^{\mathcal O(k^D)}$, and hence only a
linear dependence on $n$ for fixed $k$ and $D$.  Additional structural
knowledge may instead be used to choose a smaller candidate set.

In addition, for any $j\in V$ and $P\in\{X,Y,Z\}$ we include an onsite Pauli-dephasing term $\ell_{j,P}\cL_{j,P}$, where the unknown rate satisfies $0\le\ell_{j,P}\le\tau$ and
\begin{align}
\cL_{j,P}(X)=\frac{1}{2}(P_jXP_j-X).
\end{align}
Both $-i\cP_{(x,z)}$ and $\cL_{j,P}$ are trace-annihilating and
vanish on the identity; their physical nonnegative-rate sums generate
trace-preserving unital channels. General non-unital onsite dissipation and
dissipative terms supported on more than one qubit are not included.
We write $\mathcal K^H_{(x,z)}:=-i\cP_{(x,z)}$ for the
Hamiltonian working-basis generator.  This known factor makes all transfer
matrix entries used below real on Hermitian Pauli probes.

Furthermore, the graph is assumed to have dimension $D$, meaning that balls grow at most polynomially with their radius: for $v\in V$ and integers $r\ge0$, let
\begin{align*}
B_r(v)=\{v'\in V| \operatorname{dist}(v,v')\le r\},
\end{align*}
there exist positive coefficients $C_1,C_2$ independent of $n$ such that
\begin{align}\label{eq:dimensiongraph}
&|B_r(v)|\le C_1(1+r)^D,\\
&|B_r(v)|-|B_{r-1}(v)|\le C_2(1+r)^{D-1}\quad(r\ge1).
\end{align}
Note that the exponent $D$ matches the standard notion of dimension when the graph $\mathsf{G}$ coincides with a regular lattice.  Set
$b_k:=C_1(1+k)^D$, an upper bound on the number of vertices in a
diameter-$k$ interaction neighbourhood. We denote the set of nonzero
interactions $h_{\alpha}=h_{(x,z)}\ne 0$ by $\mathcal{A}$. A simple estimate is
\begin{align}\label{eq:numberinteractions}
|\mathcal{A}|\le n4^{b_k}\,.
\end{align}

\begin{prop}[Locality and comparison estimates]\label{prop:locality_bounds}
For the generator class in Eq.~\eqref{eq:thegenerator}, the following
finite-time locality estimates hold uniformly in the system size.  If $X$ is
supported on a region $A$, $A(r)$ is its graph $r$-neighbourhood, and
$\mathcal S_{A(r)}$ is the sum of physical generator terms whose supports are
contained in $A(r)$, then there are known constants
$C_{\rm loc},C_{\rm gen},C_{\rm cmp},\mu,v,\mu'>0$, depending only on
$k,D,C_1$, the coefficient bounds, and the fixed local Pauli basis, such that,
with $r_0=k$, for $0\le s\le t\le T$ and integral $r\ge r_0$,
\begin{align}
 \|T_{\mathcal S}(s,t)(X)-T_{\mathcal S_{A(r)}}(s,t)(X)\|_
 \infty
 &\le C_{\rm loc}\|X\|_\infty(1+r_A)^D
 e^{-\mu(r-r_0)}(e^{v(t-s)}-1),                                \label{eq:locality-assumption}\\
 \|(\mathcal S(t)-\mathcal S_{A(r)}(t))T_{\mathcal S}(s,t)(X)\|_
 \infty
 &\le C_{\rm gen}\|X\|_\infty(1+r_A)^{2D-1}
 e^{-\mu(r-r_0)}(e^{v(t-s)}-1).                                \label{eq:generator-tail-assumption}
\end{align}
Here $r_A$ is the radius of a ball containing $A$ and $D\ge1$.  Moreover, if
$\mathcal S_1$ and $\mathcal S_2$ use the fixed labelled Pauli decomposition
and every pair of corresponding local terms differs by at most $\eta$ in
completely bounded induced operator norm, then
\begin{align}
 \|T_{\mathcal S_1}(s,t)(X)-T_{\mathcal S_2}(s,t)(X)\|_\infty
 \le C_{\rm cmp}\eta\|X\|_\infty(1+r_A)^{2D-1}
 (e^{\mu'(t-s)}-1).                                               \label{eq:comparison-assumption}
\end{align}
The physical Heisenberg adjoints $T_{\mathcal S_i}(s,t)^*$ obey the analogous
truncation and comparison estimates with the same constants.
\end{prop}
\begin{proof}
This is Corollary~\ref{cor:route2-LR} of the right-generated
Heisenberg-picture Lieb--Robinson theorem proved in
Appendix~\ref{app:LR}.  The corollary applies that theorem to the present
left-generated channel only through the reversed-clock representation in
Eq.~\eqref{eq:route2-reversed-clock}.  
\end{proof}

\paragraph{Norms:}

For $p,q \in [1,\infty]$, the induced $p \to q$ norm of a linear map $\Phi : \mathcal{M}_{2^n} \to \mathcal{M}_{2^n}$ is
\begin{align}
  \|\Phi\|_{p \to q} := \sup_{O \neq 0} \frac{\|\Phi(O)\|_q}{\|O\|_p}.
\end{align}
Here, $\|O\|_p$ denotes the Schatten $p$-norm:
\begin{align}
  \|O\|_p = \left( \sum_{j} s_j(O)^p \right)^{1/p}, 
  \quad \|O\|_\infty = \max_j s_j(O),
\end{align}
where $s_j(O)$ are the singular values. If we omit the subscript in a Schatten norm, we refer to the $\|\cdot\|_{\infty}$ norm of the matrix.
We will also use matrix variations of these norms. For $A\in\mathcal{M}_{D}$, we define:
\begin{align}
\|A\|_{\ell_p\to\ell_q}=\sup_{x\in \C^D} \frac{\|Ax\|_{\ell_q}}{\|x\|_{\ell_p}},
\end{align}
where $\|\cdot\|_{\ell_p}$ is the usual $\ell_p$ norm of a vector. Often, we will also denote by $\|\cdot \|$ the sup norm of an operator-valued function. 
Similarly, for $p\in[1,\infty)$ and the interval $I\subset \mathbb{R}$, we write
\[
\|f\|_{p,I}:=\Big(\frac{1}{|I|}\int_I |f(t)|^p\,dt\Big)^{1/p}\qquad \text{ and }\qquad 
\|f\|_{\infty,I}:=\sup_{t\in I}|f(t)|.
\]

\paragraph{Parametrizing and interpolating time-dependences:} The main innovation of this paper is the learning of generators with time-dependent parameters. We further need to assume that the latter can be approximated within nicely behaved function spaces. More precisely, each time–dependent coefficient is modeled by a real-valued function on an interval
$I\subset\mathbb{R}$ (typically $I=[0,T]$ with $T=\mathcal{O}(1)$). 
We will work with finite–dimensional function classes $\cF_m\subset C^1(I,\mathbb{R})$ indexed by a degree/complexity
parameter $m$ (e.g., degree–$m$ polynomials). We will resort to various classical inequalities involving polynomials or trigonometric polynomials and refer to~\cite{Borwein1995} for a review and proof of the various statements here.
The following properties of $\cF_m$ are crucial for our protocol, which we will refer to as Markov-stable function systems:

\begin{definition}[Markov-stable function system]\label{def:Markov_stable}
We call a  subspace $\mathcal{F}\subset C^1(I,\mathbb{R})$ with corresponding subspaces $\cF_m\subset \cF$ for $m\in\mathbb{N}$ a Markov-stable function system (MSFS) if:

\begin{enumerate}
    \item The hierarchy is nested, contains constants, and has polynomial
    dimension: $1\in\cF_0$, $\cF_m\subseteq\cF_{m+1}$, and
    $\dim(\cF_m)=\mathrm{poly}(m)$.
\item For all $K\in\mathbb{N}$ and any $f_1,\ldots, f_K\in\cF_{m}$:
\begin{align}
t\mapsto\int_0^t ds_K\int_0^{s_K}ds_{K-1}\cdots
\int_0^{s_{2}}ds_1\, f_K(s_K)\ldots f_1(s_1)
\in \cF_{G(m,K)},
\end{align}
for some function $G(m,K)=\poly{m,K}$. We take $G$ nondecreasing in $K$. Since each
$\cF_m$ is a vector space, finite sums of Dyson orders up to $K$ then lie in
$\cF_{G(m,K)}$.
\item For any compact interval $I\subset\mathbb{R}$ of positive length, the functions $f\in\cF_m$ satisfy a Markov-type inequality of the form:
    \begin{align}\label{eq:markov-type}
    \|f'\|_{\infty,I}\leq C_{\operatorname{der}}(I,m)\|f\|_{\infty,I}
    \end{align}
    with $C_{\operatorname{der}}(I,m)=\mathrm{poly}(m)$. 
    \item We can perform numerically stable interpolation over $\cF_m$: i.e.~there exist nodes $\Xi_m=\{t_i\}_{i=1}^{\xi}\subset I$
with $\xi=\mathrm{poly}(m)$ and a $\mathrm{poly}(m)$–time algorithm that, given estimates $y_i$ with $|y_i-f(t_i)|\le \epsilon$ for some $f\in\cF_m$, returns
$\widehat f\in\cF_m$ satisfying 
\begin{equation}
\label{eq:stable-interp}
\|f-\widehat f\|_{\infty,I}\le C_{\mathrm{int}}(I,m)\,\epsilon,
\end{equation} 
with high probability and $C_{\mathrm{int}}(I,m)=\mathrm{poly}(m)$. The
polynomial upper bounds for $C_{\rm int}$ and $C_{\rm der}$ are known to the
algorithm, and the returned representation and its derivative can be
evaluated in $\poly{m}$ time. For
$I=[0,T]$, the nodes can be selected so that
$\min_i t_i=\Omega(T/\poly{m})$ deterministically.
After enlarging $C_{\rm int}$ if necessary, we also require that these nodes
form a norming set for $\mathcal F_m$:
\begin{align}
 \|q\|_{\infty,I}
 \le C_{\rm int}(I,m)\max_{t_i\in\Xi_m}|q(t_i)|,
 \qquad q\in\mathcal F_m.                                    \label{eq:stable-node-norming}
\end{align}
\end{enumerate}
\end{definition}

\begin{definition}[$C^1$ approximation by an MSFS]
\label{def:C1-MSFS-approximation}
Let $I=[0,T]$ and let $\Theta\subset C^1(I,\mathbb R)$ be a family of
coefficient schedules.  We say that $\Theta$ has the known uniform
$C^1$-approximation modulus $(\rho_m)_{m\ge0}$ in the MSFS
$(\mathcal F_m)_{m\ge0}$ if $\rho_m\downarrow0$ and, for every
$\theta\in\Theta$ and every $m$, there is a $p_{\theta,m}\in\mathcal F_m$
such that
\begin{align}
 \|\theta-p_{\theta,m}\|_{\infty,I}
 +T\|\theta'-p_{\theta,m}'\|_{\infty,I}\le \rho_m .            \label{eq:C1-MSFS-modulus}
\end{align}
The approximants need not be known to the learner; the modulus and the MSFS
are known.  Exact membership in $\mathcal F_{m_0}$ is the special case
$\rho_m=0$ for $m\ge m_0$.  We call the approximation exponentially stable
in $C^1$ if, for known constants $C_{\rm app},c_{\rm app},\nu>0$ and a
regularity parameter $m_0$,
\begin{align}
 \rho_m\le C_{\rm app}
 \exp\!\left[-c_{\rm app}(m/m_0)^{1/\nu}\right].              \label{eq:exp-C1-modulus}
\end{align}
\end{definition}

 \paragraph{Polynomials.} If $\mathcal{F}_m\equiv \mathcal{P}_m$, the set of degree $\le m$ polynomials over $I=[0,T]$, $\dim(\cP_m)=m+1$. Moreover, for all $K\in\mathbb{N}$, it is easy to see that $G(m,K)=Km+K$. Markov brothers' inequality~\cite{Markoff1916} (see also \cite[Sec.~5.2]{Borwein1995}) gives
$\|f'\|_{\infty,I}\le \frac{2m^2}{T}\,\|f\|_{\infty,I}$, so $C_{\mathrm{der}}(I,m)=\mathcal{O}(m^2/T)$. A concrete deterministic choice is the set of shifted Chebyshev roots
\begin{align}
t_j=\frac{T}{2}\left(1-\cos\frac{(2j-1)\pi}{2(m+1)}\right),
\qquad j=1,\ldots,m+1.                                           \label{eq:det-cheb-nodes}
\end{align}
Its smallest node is $\Theta(T/m^2)$, and interpolation from uniformly
accurate values has Lebesgue constant $\cO(\log m)$.  Oversampling and robust
regression may instead be used to tolerate stochastic noise or outliers; see
\cite[Corollary 1.5]{Kane2017-dn}.

\bigskip

\noindent Combining \eqref{eq:stable-interp} and
\eqref{eq:markov-type} shows that if $\widehat f\in\cF_m$ interpolates $f\in\cF_m$ from noisy samples on $\Xi_m$
with uniform error $\|f-\widehat f\|_{\infty,I}\le C_{\mathrm{int}}(I,m)\epsilon$, then
\[
\|f'-\widehat f'\|_{\infty,I}\le C_{\mathrm{der}}(I,m)\,C_{\mathrm{int}}(I,m)\,\epsilon
=\mathrm{poly}(m)\,\epsilon.
\]
Thus the derivative estimation error inflation is only polynomial in the model degree $m$. In particular, once we can
evaluate $f(t)$ on nodes of a $\mathrm{poly}(m)$-size grid with accuracy $\epsilon$, we can stably recover $f'(t)$
at prescribed points with accuracy $\mathrm{poly}(m)\epsilon$.

\begin{remark}
We use the simultaneous smoothness promise
Eq.~\eqref{eq:C1-MSFS-modulus}, rather than exact membership.  Appendix
\ref{app:MSFS} proves that the reconstruction error is the statistical error
plus an explicit multiple of $\rho_m$.  The learner therefore increases the
working degree until this deterministic bias fits inside the requested error
budget.  Under Eq.~\eqref{eq:exp-C1-modulus}, this requires only
$m=\poly{m_0,\log(\epsilon^{-1})}$ and recovers the scaling of the exact-MSFS
result.  Analytic schedules on a common complex neighbourhood of $[0,T]$ are
a standard sufficient example; finite smoothness is also allowed whenever
its certified (generally algebraic) $C^1$ approximation rate passes the
explicit bias test in Eq.~\eqref{eq:approx-working-degree}.
\end{remark}

\subsection{Problem  and main result}
Having stated our assumptions regarding the evolution, we are now ready to present the problem addressed in this paper:
\begin{problem}[Time-dependent Hamiltonian and Lindbladian learning]\label{prob:time_dependent}
Let $\cS(t)$ be a family of time-dependent Lindbladians on $n$ qubits s.t.~the following assumptions are met:

\paragraph{Structural assumptions:} The generator is of the form
\begin{align}\label{eq:thegenerator}
\cS(t)\equiv \cH(t)+\cL(t):=\sum_{(x,z)\in\mathcal{A}} -ih_{(x,z)}(t)\cP_{(x,z)}+\sum_{j\in V}\sum_{P\in \{X,Y,Z\}}\ell_{j,P}(t)\,\cL_{j,P},
\end{align}
where all $h_{(x,z)},\ell_{j,P}\in C^1([0,T],\mathbb R)$ share a known
$C^1$-approximation modulus $\rho_m$ in a known MSFS, as in
Def.~\ref{def:C1-MSFS-approximation}.  We assume
$\|h_{(x,z)}\|_{\infty,[0,T]}\leq1$ and
$\|\ell_{j,P}\|_{\infty,[0,T]}\le \tau$ for some $\tau>0$, and the
Hamiltonian is geometrically $k$-local for all $t\in[0,T]$. Here
$\mathcal A$ contains distinct nonidentity strings and
$h_{(0,0)}\equiv0$. Here, $\mathcal A$ is the known set of candidate nonidentity Pauli strings, and we fix the identity gauge freedom by setting $h_{(0,0)}(t)\equiv0$. We term these generators $C^1$-approximable MSFS
generators.  The underlying site graph $\mathsf G$ satisfies the
polynomial-growth condition in Eq.~\eqref{eq:dimensiongraph}; this is needed
for the advertised learning resources, although the general Lieb--Robinson theorem in
Appendix~\ref{app:LR} does not require it.
The rates $\ell_{j,P}(t)$ are real and nonnegative pointwise.  The locality
and comparison estimates needed below are consequences of
Proposition~\ref{prop:locality_bounds}, not additional input promises.

\begin{remark}
The dissipative model is deliberately restricted to weak onsite Pauli
dephasing. General non-unital onsite dissipation and multiqubit dissipative
terms require new stable-probe and locality arguments and are left as open
problems in Sec.~\ref{sec:outlook}.
\end{remark}

\paragraph{Access model:} Given a chosen time $t_0$, we can prepare an
arbitrary product of single-qubit Pauli eigenstates, evolve it by the forward
channel $T(0,t_0)$, and measure in an arbitrary product Pauli basis. Each such
experiment produces one sample. Many existing platforms support these
operations~\cite{Krantz_2019,Foss_Feig_2025}. Process shadows reconstruct
$2^{-n}\tr[Q T(0,t_0)(P)]$ and therefore extend to bounded local Hermitian
probe operators by linearity.

\paragraph{Goal:} Given $0<\delta,\epsilon<1$ and an interval $[0,T]$, we say that we solved the time-dependent Hamiltonian and Lindbladian learning problem using $S$ samples if with probability of success at least $1-\delta$ we output functions $\{\widehat{h}_{(x,z)}\}_{(x,z)\in \mathcal A},\{\widehat{\ell}_{j,P}\}_{j\in V,P\in \{X,Y,Z\}}$ s.t.:
\begin{align}
&\forall (x,z)\in \mathcal A:\|\widehat{h}_{(x,z)}-h_{(x,z)}\|_{\infty,[0,T]}\leq \epsilon, \\
&\forall j\in V,\forall P\in \{X,Y,Z\},\|\widehat{\ell}_{j,P}-\ell_{j,P}\|_{\infty,[0,T]}\leq \epsilon.
\end{align}
The algorithm outputs zero for Hamiltonian strings outside $\mathcal A$,
including the gauge-fixed identity string.
\end{problem}
Note that if we consider the class of functions to be just given by the constant function, then the problem above boils down to learning the coefficients of the evolution in the $\ell_\infty$ norm.
Other relevant parameters are the total evolution time
$t_{\operatorname{tot}}$ and the shortest queried time $t_{\min}$.

We are now ready to state our main result more formally:

\begin{thm}[Efficient learning with weak onsite Pauli dephasing]\label{thm:main}
In the setting and access model of Problem~\ref{prob:time_dependent}, let
$\{\mathcal{P}_{(x,z)}\}_{(x,z)\in\mathcal{A}}$ and
$\{\cL_{j,P}\}_{j\in V,P\in\{X,Y,Z\}}$ be the known basis generators on
$n$ qubits. Choose $(m_\epsilon,q_\epsilon)$ and $m_{\rm eff}$ by
Eqs.~\eqref{eq:approx-working-degree}--
\eqref{eq:effective-complexity-degree}, and assume that a feasible pair
exists. Suppose that the rate bound $\tau$ and horizon $T=\cO(1)$ satisfy
the weak-dissipation condition in
Eq.~\eqref{equ:tau_T_interplay}. Then we can solve the time-dependent
Hamiltonian and Lindbladian learning problem given a total of
\begin{align}
S=\widetilde{\cO}\big(\epsilon^{-2}\,\poly{m_{\rm eff}}\,\log(n\delta^{-1})\big)\qquad \text{samples.}
\end{align}
 Furthermore, the algorithm requires 
\[  \widetilde{\cO}\big(\poly{m_{\rm eff},n,\log(\delta^{-1})}\epsilon^{-2}\big)\qquad \text{classical preprocessing and postprocessing}\,,\]
\[ t_{\operatorname{tot}}=\widetilde{\cO}\big(\epsilon^{-2}\,\poly{m_{\rm eff}}\,\log(n\delta^{-1})\big)\qquad \text{total evolution time }.\]
For a general MSFS, the shortest queried time is the deterministic node
lower bound promised in Def.~\ref{def:Markov_stable}. In the polynomial
specialization with the nodes in Eq.~\eqref{eq:det-cheb-nodes}, it is
$t_{\min}=\Theta(T/N^2)$, where
$N=\poly{m_{\rm eff},\log(\epsilon^{-1})}$ is the largest interpolation-grid size.

If the coefficients belong exactly to $\mathcal F_m$, one may take the
working coefficient degree to be $m$ (and the minimal certified degree may
be smaller). If they obey the exponentially stable $C^1$ modulus
Eq.~\eqref{eq:exp-C1-modulus}, then
$m_{\rm eff}=\poly{m_0,\log(\epsilon^{-1})}$, so the additional overhead is
only polylogarithmic in the target precision.

\end{thm}

Here $m_\epsilon$ is the working degree used to approximate and reconstruct
the coefficient schedules, while $d_\epsilon$ is the auxiliary degree used to
interpolate the evolved overlap functions and may be larger. The precision
index $q_\epsilon$ and these two degrees are summarized by
$m_{\rm eff}:=\max\{m_\epsilon,d_\epsilon,q_\epsilon\}$, which controls the
complexity. For exact polynomial schedules of degree $m$, one may take
$m_\epsilon=m$.

\begin{remark}[A linear dependence on the polynomial degree is necessary]
The dependence on the schedule degree cannot be removed, even at constant
accuracy.  Indeed, let
\[
 \mathbb B_m:=\{p\in\mathcal P_m([0,T]):\|p\|_{\infty,[0,T]}\le1\}.
\]

This is the unit ball of a real normed space of dimension $d=m+1$.
By the equivalence of covering and packing numbers and the standard
volumetric lower bound~\cite[Lemma~4.2.8 and
Proposition~4.2.10]{Vershynin.2026}, whose lower-bound proof applies
verbatim to balls of arbitrary finite-dimensional norms,
$\mathbb B_m$ contains a $1/2$-separated set
$\mathcal Q_m$ satisfying
\[
 |\mathcal Q_m|
 \ge
 \frac{\operatorname{vol}_d(\mathbb B_m)}
      {\operatorname{vol}_d(\tfrac12\mathbb B_m)}
 =2^d=2^{m+1}.
\]
After discarding the zero polynomial if present, we retain
$M:=|\mathcal Q_m|\ge2^{m+1}-1$.

Enumerate $\mathcal Q_m=\{p_1,\ldots,p_M\}$.  For every $j$, consider the one-qubit instance
$H_j(t)=p_j(t)X/2$, with all dephasing rates set to zero.  It obeys all
coefficient promises and belongs exactly to the degree-at-most-$m$
polynomial class.  Any protocol in the stated product-Pauli access model that
returns an estimate $\widehat p$ satisfying
$\|\widehat p-p_j\|_{\infty,[0,T]}\le1/8$ with probability at least $2/3$
also identifies $j$ by nearest-neighbour decoding.

To count how many hypotheses $S$ experiments can identify, choose $j$
uniformly and average over the learner's internal randomness.  Some fixed
random seed still succeeds on average with probability at least $2/3$.
With that seed fixed, an adaptive one-qubit protocol is a binary decision
tree of depth at most $S$, since every Pauli measurement has two outcomes.
Query times, preparations, and measurement bases merely label its nodes; its
leaves number at most $2^S$.  Because a leaf's output can lie within $1/8$ of
at most one member of the $1/2$-separated packing, its average success
probability is at most $2^S/M$.  Therefore
\[
 \frac23\le\frac{2^S}{M},
 \qquad
 S\ge\log_2M+\log_2(2/3)\ge m+\log_2(2/3).
\]
Thus $S=\Omega(m)$ samples are necessary already for a single Hamiltonian
coefficient at fixed constant accuracy and success probability.  In
particular, at least linear dependence on $m$ is unavoidable, even for
adaptive protocols.
\end{remark}

\section{Informal overview of the protocol}

\begin{figure}[ht!]
  \centering
  \begin{minipage}[t]{0.32\textwidth}
    \includegraphics[width=\linewidth]{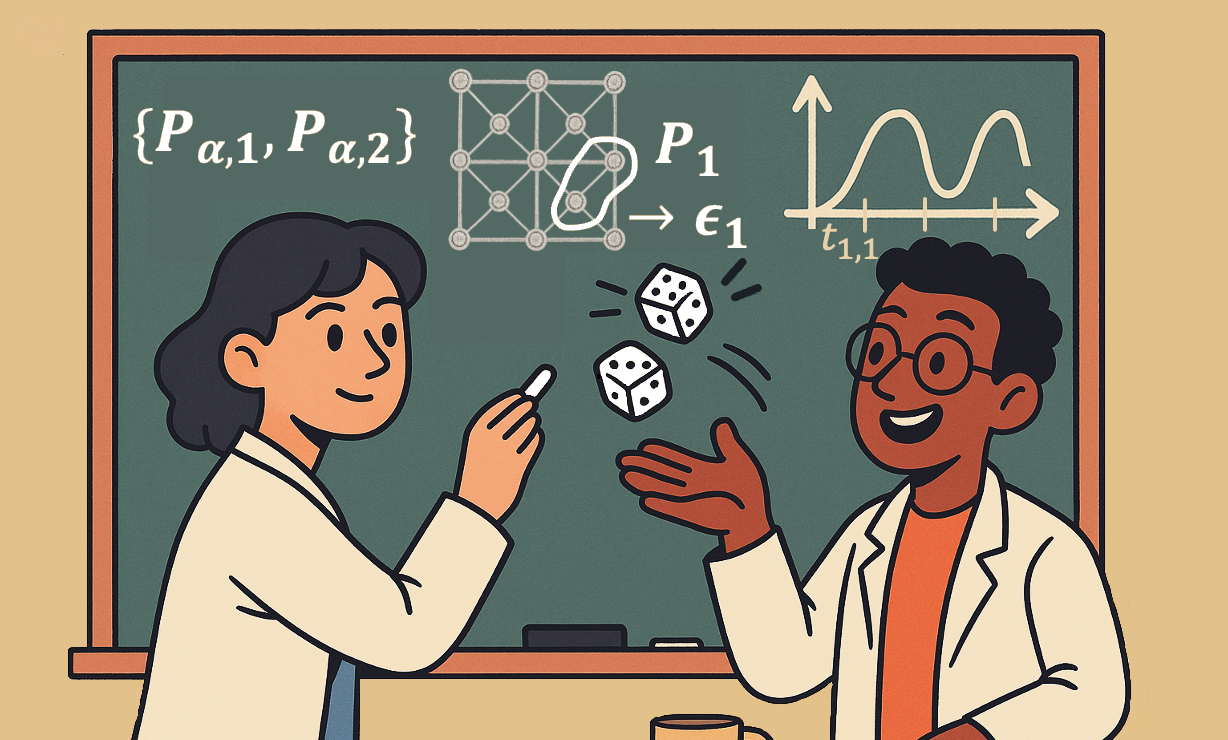}
    \par\smallskip
    \small \center{(P1)\textit{ Classical preprocessing}}
  \end{minipage}\hfill
  \begin{minipage}[t]{0.32\textwidth}
    \includegraphics[width=\linewidth]{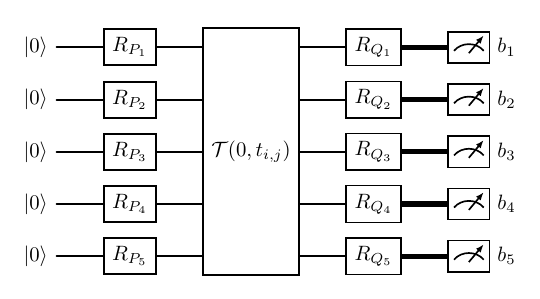}
    \par\smallskip
    \small \center{(P2)\textit{ Quantum data acquisition}}
  \end{minipage}\hfill
  \begin{minipage}[t]{0.32\textwidth}
    \includegraphics[width=\linewidth]{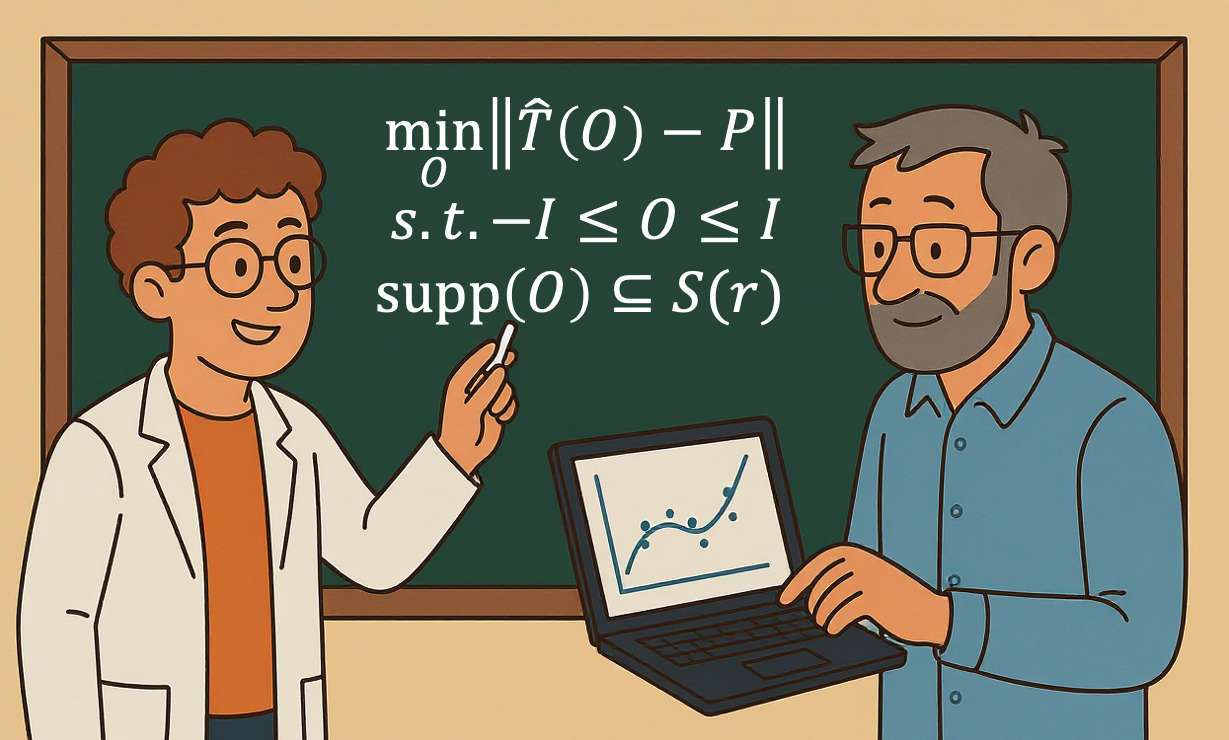}
    \par\smallskip
    \small  \center{(P3)\textit{ Classical postprocessing}}
  \end{minipage}

  \caption{Overview of the three phases of our protocol. (P1)
  we compute the necessary degree of the functions and select deterministic
  stable interpolation nodes. For each time, we compute the necessary number
  of samples and draw random Paulis used for quantum data acquisition. (P2)
  process shadow tomography: we run the quantum circuits above for the
  prescribed times and random initial Pauli rotations and bases, recording the
  measurement outcomes. (P3) we form SDPs to get stable linear systems;
  stably interpolate derivatives, solve for values at chosen times, then
  re-interpolate to recover the generator's time-dependent functions.
  This figure was generated using OpenAI's ChatGPT (GPT-5.5 and GPT-5.6)
  under the authors' direction. The authors checked the final figure against
  the protocol for scientific accuracy.
  }
  \label{fig:three-side-by-side}
\end{figure}

Our protocol learns the time-dependent parameters $\{h_{\alpha}(t)\}_{\alpha\in\{0,1\}^{2n}}$ and $\{\ell_{j,P}(t), j\in V,P\in \{X,Y,Z\}\}$ of a Lindbladian $\mathcal{S}(t)$, assuming:
\begin{enumerate}
    \item the underlying interactions are local with respect to a known graph of effective dimension $D$, and
    \item the coefficient functions and their derivatives obey the certified
    $C^1$ MSFS approximation promise in
    Def.~\ref{def:C1-MSFS-approximation}, at a selected degree $m_\epsilon$.
\end{enumerate}
For clarity, we first present the Hamiltonian mechanism.  The weak onsite
Pauli-dephasing extension uses the transformed working basis and explicit
error budget of Theorem~\ref{thm:systemAnoisy}.  The workflow is summarized
in Fig.~\ref{fig:workflow}.

\subsection{From parametrized functions to interpolation}

If we could approximately evaluate each $h_{\alpha}(t)$ at chosen times, we
could stably interpolate its MSFS approximant and recover the true function
up to the certified bias $\rho_{m_\epsilon}$ (see
Sec.~\ref{sec:appproximation}). The key challenge is thus \emph{approximately
extracting $h_{\alpha}(t)$ from experimental data for prescribed times
$t_i$}.

\medskip
\noindent\textbf{Takeaway:} Our goal is to find a method to evaluate $h_{\alpha}$ at prescribed times $t_i$.

\subsection{Relating derivatives of transfer-matrix entries to \texorpdfstring{$h_{\alpha}(t_i)$}{h-alpha(t-i)}}

For a Hermitian probe $O$ and a Pauli observable $P_{(x,z)}$, define the
Schrödinger-picture transfer-matrix entry
\begin{align}\label{eq:thefunctionf}
f_O(t): = 2^{-n} \tr{T(0,t)(O) P_{(x,z)}}.
\end{align}
Differentiating and inserting the generator $\mathcal{S}(t)$ yields a \emph{linear equation} for the values of $\{h_\alpha(t)\}_{\alpha\in\mathcal{A}}$: since
\begin{align}\label{equ::derivative_form}
f_O'(t)=2^{-n}\tr{\cS(t)(T(0,t)(O))P_{(x,z)}},
\end{align}
if we denote by $c^{O,t}_{(x',z')}=2^{-n}\tr{T(0,t)(O)P_{(x',z')}}$, then we can expand Eq.~\eqref{equ::derivative_form} as
\begin{align}\label{equ:linear_system}
\underbrace{f_O'(t)}_{\text{estimate}}=\sum_{\alpha\in\mathcal{A}}\sum\limits_{(x',z')}\,\underbrace{h_\alpha(t)}_{\text{unknown}}\,\cdot\,{c^{O,t}_{(x',z')}}\,\cdot \,\underbrace{-i2^{-n}\tr{\cP_{\alpha}(P_{(x',z')})P_{(x,z)}}}_{\text{known and real}}.
\end{align}
This gives an equation at each prescribed time $t_i$ whose unknowns are the
parameters $h_{\alpha}(t_i)$. Coefficients in this system involve overlaps
$c^{O,t_i}_{(x',z')}$ that depend on the unknown evolution itself. Hence, both
$f_O'(t_i)$ and $c^{O,t_i}_{(x',z')}$ must be estimated from data.  The probe
$O$ need not be a density matrix: process shadows estimate the Pauli transfer
matrix, and the required value for a bounded local $O$ follows by linearity.

\medskip
\noindent\textbf{Takeaway:} The derivative $f'(t_i)$ contains the parameters we want, but only through a \emph{noisy linear system} whose coefficients are partially unknown. 

\subsection{Making the linear system stable}
A first observation is that, as shown in~\cite{StilckFranca.2024} and
App.~\ref{app:stable_pauli}, for $k$-body terms one can find Pauli strings
$P_{\alpha,1},P_{\alpha,2}$ for which
$2^{-n}\tr{\mathcal K^H_{\alpha}(P_{\alpha,1})P_{\alpha,2}}=1$ after a known
sign or row normalization, while the other designated columns vanish. Thus, an
ideal choice of $O_{i,\alpha}$ would satisfy
$T(0,t_i)(O_{i,\alpha})=P_{\alpha,1}$.  Inverting a full $n$-qubit channel is
too costly, so instead we construct a sparse, diagonally dominant system that
can be inverted stably (Prop.~\ref{prop:stable_systems}).
This requires finding bounded probe operators $O_{i,\alpha}$ such that $O_{i,\alpha}(t_i):=T(0,t_i)(O_{i,\alpha})$ is close to the target Pauli and has small weight on other terms.

\medskip
\noindent\textbf{Takeaway:} Stable recovery is possible if we construct \emph{approximate isolating probes} that yield sparse, diagonally dominant systems.

\subsection{Constructing good probe operators via LR bounds and SDPs}\label{subsec:good-observables}

For the Hamiltonian part, the ideal probe is
$T_{\mathcal H}(0,t_i)^{-1}(P_{\alpha,1})$.  A Hamiltonian LR estimate
localizes it.  Proposition~\ref{prop:locality_bounds} and the perturbative
comparison bound then show that the same local probe remains an approximate
isolator for weak Pauli dephasing.  We need only constant isolator accuracy,
so the localization region remains independent of the final statistical
precision. We find such a probe by solving a constant-size semidefinite
program using estimates of local Pauli-transfer entries.
These estimates are obtained efficiently in parallel via \emph{process shadow tomography} (see Prop.~\ref{prop:good_obs}). In the following, we provide a brief description of the techniques used in this step, namely process shadow tomography, LR bounds and SDPs. We leave technical details to the following sections, starting from 
Sec.~\ref{sec:appproximation}.

\paragraph{Process shadow tomography:}

Process shadows enable parallel estimation of Pauli overlaps of the form $$2^{-n}\mathrm{tr}[T(0,t)(P_{(x,z)})P_{(x',z')}]$$ for the evolution map $T(0,t)$. The protocol consists of:
\begin{enumerate}
    \item Preparing a random product Pauli eigenstate \\(tensor product of eigenstates of $X$, $Y$, or $Z$ on each qubit)
    \item Evolving under $T(0,t)$
    \item Measuring in a random product Pauli basis
\end{enumerate}
For collections of Pauli strings $\{P_{(x_i,z_i)}\}_{i=1}^{K_1}$ and $\{P_{(x'_j,z'_j)}\}_{j=1}^{K_2}$ with maximum combined weight 
$${\max_{i,j}}\,w(P_{(x_i,z_i)}) + w(P_{(x'_j,z'_j)}) \leq w\,,$$
all $K_1 K_2$ overlaps can be estimated to precision $\epsilon$ with probability $\geq 1-\delta$ using 
\begin{align}\label{eq:processshadows}
    S = \mathcal{O}\left(3^{w} \log(K_1 K_2 \delta^{-1}) \epsilon^{-2}\right)
\end{align}
samples. The postprocessing involves a median-of-means estimator from
appropriately weighted outcomes; see~\cite{StilckFranca.2024}. By estimating
the Pauli transfer matrix and using linearity, this gives all required channel
overlaps $2^{-n}\mathrm{tr}[T(0,t)(O)P_{(x,z)}]$ for local bounded probes
$O$, with exponential dependence only on the constant local weights.

\paragraph{Lieb-Robinson (LR) bounds:}

The following finite-speed propagation bound is the constant-support
specialization of Proposition~\ref{prop:locality_bounds} used in the
localization steps.

\begin{prop}\label{prop:truncation_bound}
Let $\mathcal{S}$ be a time-dependent geometrically $k$-local Lindbladian defined over a graph $\mathsf{G}=(V,E)$ of the form of Eq.~\eqref{eq:thegenerator}, and let $O$ be an operator supported on a region $S\subset V$ of radius $r_S=\mathcal{O}(1)$. Then there exist constants $v,\mu,C_3$ depending on $k,D,C_1$, the coefficient bound $\tau$, and the fixed local Pauli basis such that for all integral $r\ge r_0=k$ and enlargement $S(r):=\{v\in V\mid\operatorname{dist}(v,S)\le r\}$
\begin{align}\label{eq:LRused}
\left\| T_{\mathcal{S}}(s,t)(O)-T_{\mathcal{S}_{S(r)}}(s,t)(O)\right\|\le\,C_3 e^{-\mu(r-r_0)}\big(e^{v(t-s)}-1\big)\,\|O\|.
\end{align}
\end{prop}
\begin{proof}
Apply Eq.~\eqref{eq:locality-assumption} with $A=S$ and absorb the fixed
polynomial factor in $C_3$.
\end{proof}

\paragraph{Semidefinite programs (SDPs):}

Semidefinite programs are a class of convex optimization problems involving matrices; see e.g.~\cite[Chapter 7]{watrous} for a review. Under mild regularity conditions (e.g., Slater's condition), SDPs can be solved to $\epsilon$-accuracy in time polynomial in the problem dimension and $\log(1/\epsilon)$ using interior-point methods.
In our context, we encounter optimization problems of the form:
\begin{align}
\min_{O} \quad &\|\widehat{T}(O)-P\| \label{eq:our_sdp}\\
\text{s.t.} \quad & -I \preceq O \preceq I \nonumber\\
& \mathrm{supp}(O)\subseteq S(r) \nonumber
\end{align}
where $\widehat{T}$ is a real, Hermiticity-preserving empirical map between
the two local regions, $P$ is a target Hermitian operator, and $S(r)$ denotes
the permitted input support. Complete positivity of the empirical map is not
required for the linear matrix inequalities.
This problem admits an efficient SDP formulation. To see this explicitly, we introduce an auxiliary variable $t \geq 0$ and reformulate \eqref{eq:our_sdp} as:
\begin{align}
\min_{O, y} \quad & y \\
\text{s.t.} \quad & -yI \preceq \widehat{T}(O) - P \preceq yI \nonumber\\
& I - O \succeq 0 \nonumber\\
& I + O \succeq 0 \nonumber\\
& \mathrm{supp}(O) \subseteq S(r)
\end{align}
The equivalence follows from the fact that the operator norm constraint $\|\widehat{T}(O) - P\| \leq y$ is equivalent to the linear matrix inequality $-yI \preceq \widehat{T}(O) - P \preceq yI$. Therefore, this problem can be solved in polynomial time in the dimension of $X$, which in our case will be restricted to a  region of constant-size due to the locality constraints.

\medskip
\noindent\textbf{Takeaway:} LR bounds localize the inversion problem, allowing us to find good $O$ by solving constant-size SDPs with shadow tomography data.

\subsection{Estimating derivatives stably}

We still need $f'(t)$ in Eq.~\eqref{equ:linear_system}.  
Following the approach of~\cite{StilckFranca.2024}, we show $f(t)$ can be approximated by an MSFS of slightly higher degree (Sec.~\ref{sec:appproximation}).  
Markov-type inequalities and the localized generator comparison guarantee a
controlled approximation of both $f(t)$ and $f'(t)$, with only a polynomial
precision overhead in $m_{\rm eff}$ and an explicit deterministic model-bias
term.

Crucially, LR bounds again ensure that both the degree blow-up and the required sample precision remain polylogarithmic in $n$.  
The required expectation values can all be estimated via process shadows on constant-weight Pauli strings.

\medskip
\noindent\textbf{Takeaway:} We estimate $f'(t)$ by interpolating $f(t)$, keeping sample complexity quadratic in precision and logarithmic in system size.

\subsection{Putting it all together}

\begin{enumerate}
    \item Use process shadows + SDPs to build approximate isolating probes $O_{i,\alpha}$ for each $\alpha\in\mathcal{A}$ and $t_i\in[0,T]$ (Prop.~\ref{prop:good_obs}).
    \item For each $t_i$, estimate derivatives $f'(t_i)$ for these probes (Sec.~\ref{sec:appproximation}).
    \item Solve the sparse, diagonally dominant working-basis system to recover
    $h_\alpha(t_i)$ and $\widetilde\ell_{j,P}(t_i)$, then convert
    $\ell_j=\Gamma_j\widetilde\ell_j$
    (Thm.~\ref{thm:systemAnoisy}).
    \item Interpolate over $t_i$ to reconstruct the coefficient functions on
    $[0,T]$, including the certified $C^1$ approximation bias
    (Def.~\ref{def:C1-MSFS-approximation}).
\end{enumerate}

The overall protocol achieves
\[
S = \widetilde{\cO}\!\left(\epsilon^{-2}\,\poly{m_{\rm eff}}\log(n\delta^{-1})\right)
\]
samples and
$\widetilde{\cO}(\epsilon^{-2}\poly{n,m_{\rm eff},\log(\delta^{-1})})$
classical processing, with minimal experimental requirements (Pauli product state preparation and measurement). Thus, we see that our protocol is highly efficient both in terms of sample and computational complexity and works under minimal experimental setups.

 \subsection{Extrapolation to larger times}

Our learning guarantee is for a constant training interval $[0,T]$.
For degree-$m$ polynomial schedules, extrapolation to $T_f\ge T$ amplifies a
coefficient error by the exact factor
$\chi_m(T_f/T)=|\mathrm T_m(2T_f/T-1)|$.  Under additional uniform
interaction and rate bounds on $[0,T_f]$, local physical expectation values
can therefore be verified, for final target error $\eta$, with sample cost
\begin{align}
 \widetilde\cO\!\left(
 \eta^{-2}M^2T_f^2
 \bigl(\sqrt{T_f/T}+\sqrt{T_f/T-1}\bigr)^{4m}\poly{m}\right),   \label{eq:overview-explicit-extrapolation-cost}
\end{align}
or with the corresponding locality-sensitive alternative in
Eq.~\eqref{eq:correct-extrapolated-verification}.  A simpler upper bound is
\begin{align}
 \widetilde\cO\bigl(\eta^{-2}M^2T_f^2
 (4T_f/T)^{2m}\poly{m}\bigr).                                  \label{eq:overview-coarse-extrapolation-cost}
\end{align}
This is
practical only for small fixed degree; Sec.~\ref{sec:extrapolation} states the
physical Heisenberg comparison argument and long-interval hypotheses.

\section{Comparison to previous results}
Learning Hamiltonians and open-system generators from \emph{real-time} dynamics has been explored under many access models and assumptions~\cite{tang_structure,StilckFranca.2024,PhysRevLett.130.200403,zubida2021optimalshorttimemeasurementshamiltonian,PhysRevLett.122.020504,10.1145/3670418,Flynn2022,Gu2024,Holzpfel2015,li2023heisenberglimitedhamiltonianlearninginteracting,Moebus.2023,Moebus.2025,PhysRevLett.107.210404,PhysRevLett.112.190501}. However, to the best of our knowledge, all rigorous and scalable guarantees focus on time-independent models. As such, the present work is the first to provide \emph{provable} sample- and time-efficient learning guarantees for \emph{local, time-dependent} generators (Hamiltonians and certain Lindbladians) under an experimentally minimal access model, with complexity $S=\cO\big(2^{ckd}\poly{m_{\rm eff}}\polylog{n\epsilon^{-1}\delta^{-1}}\epsilon^{-2}\big)$ and $\poly{n,m_{\rm eff}}$ classical postprocessing.

\paragraph{Time-independent Hamiltonians:}
For purely Hamiltonian, time-independent evolutions, an increasingly complete algorithmic picture has recently emerged. Several families of protocols now achieve polylogarithmic dependence on system size under locality assumptions, sometimes augmented by structure-learning guarantees~\cite{tang_structure,StilckFranca.2024,PhysRevLett.130.200403,zubida2021optimalshorttimemeasurementshamiltonian,PhysRevLett.122.020504,10.1145/3670418,Flynn2022,Gu2024,Holzpfel2015,li2023heisenberglimitedhamiltonianlearninginteracting,Moebus.2023,Moebus.2025,PhysRevLett.107.210404,PhysRevLett.112.190501}. Reviewing the literature for the time-independent case is beyond the scope of this work, and we focus on what is arguably the protocol with the best guarantees available.
In the recent work of Tang et al.~\cite{tang_structure}, the authors attain Heisenberg scaling in $\epsilon$ and handle very general evolutions while learning both parameters and structure. However, that result relies on an access model in which one can \emph{interleave} the unknown generator with \emph{known} Hamiltonians and more broadly illustrates that access assumptions vary substantially across the literature (e.g., availability of entangled probes, ancilla-assisted schemes, coherent control between unknown evolutions). \cite{Dutkiewicz2024} discusses to what extent this is necessary. These differences make head-to-head comparisons delicate. Our setting is intentionally close to a practical minimum: prepare product Pauli eigenstates, evolve under the \emph{unknown} time-dependent process for chosen times, and measure in product Pauli bases --- no interleaving with known controls, no auxiliary systems, and no entangled input states are required. Within this model, we match the best known logarithmic dependence on $n$ from the time-independent case and pay only a \emph{polynomial} overhead in the certified effective degree $m_{\rm eff}$ to handle time dependence.

\paragraph{Open systems (Lindbladians):}

When dissipation is present, the prior framework most directly relevant to ours
is due to Stilck França \emph{et al.}~\cite{StilckFranca.2024}, who also
leverage stable interpolation of local observables to produce well-conditioned
\emph{linear} systems. A subsequent preprint by Möbus \emph{et
al.}~\cite{MobusBergamaschi2026} treats robust structure learning of
time-independent, general $k$-local Pauli--GKSL generators under short-time
access, without assuming that the interaction support is known. That work is
complementary to the present setting, which assumes a known local ansatz and
addresses explicitly time-dependent coefficients in an interpolable function
class on finite intervals, together with the time-dependent dissipative
Lieb--Robinson bounds required for their analysis.

The framework of Ref.~\cite{StilckFranca.2024}, however, fundamentally
exploits \emph{first-order} derivatives at $t=0$. This creates a barrier in the
time-dependent setting: First, knowledge of only $f'(0)$ is insufficient to
identify nontrivial time dependences (e.g., polynomials) of the parameters;
second, attempting to move to higher-order derivatives at $t=0$ leads to two
compounding issues: (i) numerical stability deteriorates \emph{exponentially}
with the derivative order (Markov-type constants blow up), inflating sample
complexity, and (ii) higher derivatives introduce \emph{nonlinear} dependence
on the semigroup parameters, which complicates both analysis and algorithms.
In contrast, our protocol works at \emph{finite} times and ensures that
derivatives of suitable expectation values still yield \emph{linear} systems
whose conditioning we control via Lieb–Robinson localization and diagonally
dominant constructions. As a result, the dependence on the certified effective
degree $m_{\rm eff}$ is only \emph{polynomial}, and all classical recovery
steps reduce to solving linear systems and small SDPs.

\paragraph{Bosonic systems:} Another interesting line of research involves bosonic systems, which, for example, model quantum harmonic oscillators, photons/ ultracold atoms in optical cavities, or superconducting systems. Learning the seminal Bose-Hubbard model was approached by the work \cite{li2023heisenberglimitedhamiltonianlearninginteracting}
achieving the Heisenberg limit assuming discrete control on the evolution. At the same time, part of the authors investigated an extension of \cite{StilckFranca.2024} in \cite{Moebus.2023}, which covers Hamiltonians polynomial in annihilation and creation operators --- for example, the Bose-Hubbard model. The main major challenge is that the Hamiltonian under consideration is no longer a bounded operator. This introduces a second difficulty: Lieb-Robinson bounds are rarely proven in such settings. To address this issue, dissipation serves as a key ingredient, regularizing the dynamics and enabling Lieb-Robinson-type information propagation bounds. Lately, in \cite{Moebus.2025} stronger control is assumed as in \cite{li2023heisenberglimitedhamiltonianlearninginteracting} to use engineered dissipation to not only drive the system dynamics into a finite-dimensional subspace, but also to decouple the dynamic. This reduction is then used to decouple the bosonic dynamics and to learn the coefficients with Heisenberg-limited scaling. A major application of both approaches includes Gaussian dynamics, the Bose-Hubbard model as well as the bosonic cat code, which incorporates dissipation as well as a time-dependence to implement the logical $X$-gate (see \cite{Guillaud.2023}). This not only motivates the dissipation considered but also highlights the necessity of time-dependent learning schemes, as discussed in this work.

\clearpage

\begin{figure}[p]
\centering
\begin{tikzpicture}[
  font=\footnotesize,
  node distance=1.6mm and 0mm,
  >=Stealth,
  arrow/.style={-Stealth, line width=0.7pt},
  box/.style={draw, rounded corners=1.5pt, line width=0.1,
  inner sep=4pt, text width=8cm,
  align=left},
  title/.style={font=\bfseries, align=left, inner sep=2pt},
  phase/.style={draw}
  every node/.style={outer sep=0pt}
]

\node[title] (t1) {Phase P1: Classical preprocessing (Sec.~\ref{sec:finding_equations}, Sec.~\ref{sec:appproximation})};

\def\a{2mm}

\def\b{4mm}

\node[box, below=\a of t1] (in) {\textbf{Inputs:} certified $C^1$ modulus,
MSFS hierarchy $(\mathcal{F}_m)$, targets $\epsilon,\delta,T$.};

\node[box, below=\b of in] (stable) {Compute $s$-stable Pauli pairs $\{(Q_{\beta}, \overline{Q}_\beta)\}$ for sparse rows/columns. \emph{Refs:} Def.~\ref{def:stable_paulis}, Prop.~\ref{prop:size_s_stable}.};

\node[box, below=\b of stable] (times) {Choose interpolation nodes $t_1,\dots,t_\xi\subset[0,T]$ and effective degree for derivative estimation. \emph{Refs:} Def.~\ref{def:Markov_stable}, Sec.~\ref{sec:appproximation}.};

\node[box, below=\b of times] (budget) {Verify the certified model-bias
condition; set precision $\epsilon_1$ and sample budget $S$ from LR tails and
MSFS constants. \emph{Refs:} Cor.~\ref{cor:deg_approx}, App.~\ref{app:MSFS}.};

\draw[arrow] (in.south) -- (stable.north);
\draw[arrow] (stable.south) -- (times.north);
\draw[arrow] (times.south) -- (budget.north);

\begin{scope}[on background layer]
  \node[phase, fit=(t1) (budget)] (p1box) {};
\end{scope}

\node[title, below=3.5mm of budget] (t2) {Phase P2: Quantum data acquisition (Algorithm~\ref{alg:quantum_data})};

\node[box, below=\a of t2] (shadows) {For each $t_i$: run \emph{process shadows}\\
1) Prepare random product Pauli eigenstate;\\
2) Evolve by $T(0,t_i)$; \\
3) Measure in random Pauli basis. Collect estimates of $2^{-n}\operatorname{tr}[T(0,t_i)(P)Q]$ for low-weight $P,Q$. \emph{Ref:} Algorithm~\ref{alg:quantum_data}.};

\draw[arrow] (budget.south) -- (t2.north);

\begin{scope}[on background layer]
  \node[phase, fit=(t2) (shadows)] (p2box) {};
\end{scope}

\node[title, below=3.5mm of shadows] (t3) {Phase P3: Classical postprocessing (Algorithms \ref{alg:rev_from_shadows}-\ref{alg:solve_and_interpolate})};

\node[box, below=\a of t3] (sdp) {Localized inversion via SDP: $O_{i,\beta} \leftarrow \textsc{Rev}(Q_\beta,t_i,\eta)$ minimizing $\|\widehat T(0,t_i)(O_{i,\beta})-Q_\beta\|$ (localized by LR) using shadow data. \emph{Refs:} Prop.~\ref{prop:Oexists}, Prop.~\ref{prop:good_obs}.};

\node[box, below=\b of sdp] (coeffs) {Estimate overlaps for the linear system $2^{-n}\operatorname{tr}\!\big[\mathcal{K}_\beta(O_{i,\beta}(t_i))\,\overline{Q}_{\beta}\big]$ using constant-weight Pauli strings (LR locality). \emph{Ref:} Sec.~\ref{sec:finding_equations}.};

\node[box, below=\b of coeffs] (deriv) {For each $(\alpha,i)$, interpolate $f_{O_{i,\beta}}(t)=2^{-n}\operatorname{tr}[T(0,t)(O_{i,\beta})\overline{Q}_{\beta}]$ to obtain $f'_{O_{i,\beta}}(t_i)$ stably. \emph{Refs:} Sec.~\ref{sec:appproximation}, Lem.~\ref{lem:dyson_series_conv}, Cor.~\ref{cor:deg_approx}.};

\node[box, below=\b of deriv] (lin) {Assemble sparse, diagonally dominant system $\widehat{A}^i\,\widehat{\theta}^{i,1}=\widehat{b}^i$ with $\widehat{\theta}^{i,1}\approx \theta(t_i)$. \emph{Refs:} Eq.~\eqref{equ:linear_system}, Def.~\ref{def:stable_paulis}, App.~\ref{app:stab_infty}.};

\node[box, below=\b of lin] (solve) {Solve using stability bounds to get $\widehat{\theta}^{i,1}\approx \theta(t_i)$, then interpolate over $t_i$. \emph{Refs:} Prop.~\ref{prop:stable_systems}, Sec.~\ref{sec:appproximation}.};

\node[box, below=\b of solve, fill=gray!10] (out) {\textbf{Outputs} ---  $\widehat{\theta}_\beta$ over $[0,T]$ with error $\le \epsilon$ (w.p.\ $\ge 1-\delta$). Samples: $S = \widetilde{O}\big(\epsilon^{-2}\,\operatorname{poly}(m_{\rm eff})\,\log(n\delta^{-1})\big)$; Classical: $\mathrm{poly}(n,m_{\rm eff})$. \emph{Ref:} Main theorem. \ref{thm:main}};

\draw[arrow] (shadows.south) -- (t3.north);
\draw[arrow] (sdp.south) -- (coeffs.north);
\draw[arrow] (coeffs.south) -- (deriv.north);
\draw[arrow] (deriv.south) -- (lin.north);
\draw[arrow] (lin.south) -- (solve.north);
\draw[arrow] (solve.south) -- (out.north);

\begin{scope}[on background layer]
  \node[phase, fit=(t3) (out)] (p3box) {};
\end{scope}

\end{tikzpicture}

\caption{Workflow of the learning protocol. Preprocess to pick stable Pauli pairs and interpolation nodes; collect process-shadow data; then perform SDP-based inversion, coefficient estimation, derivative computation, and interpolation to recover parameter functions. The TikZ code for this figure was prepared with substantial assistance from OpenAI's ChatGPT (GPT-5.5 and GPT-5.6) and checked by the authors for accuracy.}
\label{fig:workflow}
\end{figure}

\clearpage

\section{Approximating forward evolution of local probe operators by MSFS}\label{sec:appproximation}
We are now ready to prove the technical statements required for our protocol.
The aim of this section is to show that functions of the form $t\mapsto 2^{-n}\tr{T(0,t)(P)Q}$ for locally supported probe operators $P,Q$ are well approximated by MSFS under evolution by $C^1$-approximable MSFS generators.
\begin{lemma}\label{lem:dyson_series_conv}
Let $\cS(t)$ be a bounded family of generators, not necessarily physical or
contractive, and for $T>0$ let
\begin{align}
M=\sup\limits_{t\in[0,T]}\|\cS(t)\|_{\infty\to\infty}.
\end{align}
Define the truncated Dyson series: 
\begin{align}
D_K(t)= \mathcal I+\sum_{j=1}^K
     \int_0^t ds_j \int_0^{s_j} ds_{j-1} \cdots \int_0^{s_{2}} ds_1 \,
        \mathcal{S}(s_j) \cdots \mathcal{S}(s_1)\,.
\end{align}
Then for all $t\in[0,T]$ we have
\begin{align}\label{equ:estimate_tail}
\|T(0,t)-D_K(t)\|_{\infty\to\infty}
\leq e^{MT}\frac{(M T)^{K+1}}{(K+1)!}\,.
\end{align}
\end{lemma}
\begin{proof}
The proof of this statement is quite standard. First, we expand the propagator with the integral remainder
\begin{equation*}
    T(0,t)=D_K(t)+\int_0^t ds_{K+1} \int_0^{s_{K+1}} ds_{K} \cdots \int_0^{s_{2}} ds_1\mathcal{S}(s_{K+1}) \cdots \mathcal{S}(s_1)T(0,s_1)\,.
\end{equation*}
The elementary bound
$\|T(0,s)\|_{\infty\to\infty}\le e^{Ms}$ and a triangle inequality give
\begin{align}
    \left\|\int_0^t ds_{K+1} \int_0^{s_{K+1}} ds_{K} \cdots \int_0^{s_{2}} ds_1\mathcal{S}(s_{K+1}) \cdots \mathcal{S}(s_1)T(0,s_1)\right\|_{\infty\to\infty}\leq e^{MT}\frac{t^{K+1}M^{K+1}}{({K+1})!}\,.
\end{align}
Note that this also proves $T(0,t)=D_\infty(t)$. 
\end{proof}
We have the following corollary:
\begin{corollary}\label{cor:deg_approx}
Let $\cS(t)$ obey Problem~\ref{prob:time_dependent} and let its coefficient
schedules have the common $C^1$ approximation modulus $\rho_m\le1$ at a
chosen working degree $m$.  Then for any $\eta>0$, Hermitian operator $O$
with norm $\|O\|\le 1$ and support of
diameter $r_O$, and arbitrary Hermitian operator $S$ with
$\|2^{-n}S\|_1=1$, the function
\begin{align}\label{equ:time_traces}
f:t\mapsto 2^{-n}\tr{T(0,t)(O)S}
\end{align}
has an approximant $f_m^{\operatorname{MSFS}}\in\mathcal F_{d_{m,\eta}}$,
where
\begin{align}
 A&:=\supp(O),\nonumber\\
 r&=r_0+vT+\mathcal O(\log(\eta^{-1})),\nonumber\\
 \mathfrak G_{\rm phys}&:=
 \{\mathcal K^H_\alpha\}_{\alpha\in\mathcal A}
 \cup\{\mathcal L_{j,P}\}_{j,P},\nonumber\\
 J_r&:=\sum_{\mathcal G_\beta\in\mathfrak G_{\rm phys}:\,
       \supp(\mathcal G_\beta)\subseteq A(r)}
       \|\mathcal G_\beta\|_{\infty\to\infty},\nonumber\\
 \overline M_r&:=(2+\tau)J_r,\qquad
 \mathfrak b_r:=J_r e^{\overline M_rT}
       \max\{T,1+\overline M_rT\},                              \label{eq:transfer-model-bias-factor}\\
 d_{m,\eta}&=G(m,K),\qquad
 K=\mathcal O\!\left(\overline M_rT+
      \log\frac{1+\overline M_r}{\eta}+1\right),              \label{eq:approx-transfer-degree}
\end{align}
such that
\begin{align}
 \|f-f_m^{\operatorname{MSFS}}\|_{\infty,[0,T]}&\le
 \eta+\mathfrak b_r\rho_m,\nonumber\\
 \|f'-(f_m^{\operatorname{MSFS}})'\|_{\infty,[0,T]}&\le
 \eta+\mathfrak b_r\rho_m.                                   \label{eq:C1-transfer-approximation}
\end{align}
Here $J_r=\mathcal O(4^{b_k}(1+r_O+r)^D)$, so both $K$ and
$d_{m,\eta}$ have the same polynomial/polylogarithmic scaling used in the
exact-MSFS case.  If the coefficients belong exactly to $\mathcal F_m$, then
$\rho_m=0$ and Eq.~\eqref{eq:C1-transfer-approximation} reduces to the exact
statement.
\end{corollary}

\begin{proof}
First choose the truncation radius so that the right-hand side of
Eq.~\eqref{eq:generator-tail-assumption} is at most $\eta/3$ and that of
Eq.~\eqref{eq:locality-assumption} is at most
$\eta/[3(1+\overline M_r)]$.
Write $\mathcal S_r$ for the physical truncated generator.  For every true
physical-basis coefficient
$\theta_\beta\in\{h_\alpha,\ell_{j,P}\}$, choose a witness
$p_{\beta,m}\in\mathcal F_m$ from
Eq.~\eqref{eq:C1-MSFS-modulus}, and let
$\widetilde{\mathcal S}_{r,m}$ be obtained from $\mathcal S_r$ by replacing
$\theta_\beta$ by $p_{\beta,m}$ in the physical directions
$\mathcal G_\beta\in\mathfrak G_{\rm phys}$. The auxiliary generator need not be a
physical Lindbladian.  Nevertheless,
\begin{align}
 \Delta_r:=\sup_t\|\mathcal S_r(t)-
 \widetilde{\mathcal S}_{r,m}(t)\|_{\infty\to\infty}
 \le J_r\rho_m,
 \qquad
 \sup_t\max\{\|\mathcal S_r(t)\|_{\infty\to\infty},
 \|\widetilde{\mathcal S}_{r,m}(t)\|_{\infty\to\infty}\}
 \le\overline M_r.                                             \label{eq:localized-generator-model-error}
\end{align}
The Duhamel estimates proved in Appendix~\ref{app:MSFS} therefore bound both
the transfer function and its derivative under these two localized
generators by $\mathfrak b_r\rho_m$.

By Lemma~\ref{lem:dyson_series_conv} and Stirling's inequality, choose $K$
as in Eq.~\eqref{eq:approx-transfer-degree}, so that the function tail and
the generator-multiplied derivative tail are each at most $\eta/3$; the
protocol's explicit dyadic choices are recorded in
Eq.~\eqref{eq:explicit-Dyson-choice}.  The time-ordered integrals
contain products of at most $K$ coefficient functions and therefore belong to
the degree $d_{m,\eta}$.  Define
\begin{align}
 f_m^{\operatorname{MSFS}}(t)
 :=2^{-n}\tr[D_{r,m,K}(t)(O)S].
\end{align}
Differentiating the left-ordered auxiliary Dyson
series gives
\begin{align}
 (f_m^{\operatorname{MSFS}})'(t)
 =2^{-n}\tr[\widetilde{\mathcal S}_{r,m}(t)
 D_{r,m,K-1}(t)(O)S].
\end{align}
The localization, auxiliary-generator comparison, and Dyson-tail budgets
bound the function and derivative errors respectively by
$\eta+\mathfrak b_r\rho_m$.  H\"older's inequality transfers the operator
bounds to the normalized overlaps, proving
Eq.~\eqref{eq:C1-transfer-approximation}.
\end{proof}

Because MSFS admit stable interpolation by definition, the result shows that
functions of the form in Eq.~\eqref{equ:time_traces} and their derivatives can
be reconstructed from polynomially many uniformly accurate channel-overlap
estimates.  The model-bias term $\mathfrak b_r\rho_m$ is deterministic and,
whenever the feasibility test in Eq.~\eqref{eq:approx-working-degree} passes,
is made smaller than the statistical budget by increasing the certified
working degree, as quantified in Appendix~\ref{app:MSFS}.

We are ultimately interested in the generator parameters, not the precise
form of the function in Eq.~\eqref{equ:time_traces}.  The MSFS approximant
satisfies the Markov inequality~\eqref{eq:markov-type}; together with the
simultaneous function-and-derivative comparison in
Eq.~\eqref{eq:C1-transfer-approximation}, this lets us estimate the physical
derivative
$t\mapsto \tr{\cS(t)(T(0,t)(O))S}$ with only polynomial precision overhead.

Our main technical insight is that, for certain choices of $O$ and $P$, it is possible to ensure that the trace $\tr{T(0,t)(O)S}$ has sufficient structure to isolate the parameters of $\cS(t)$, as we will show in the next section. That is, we will be able to approximately evaluate the functions $h_\alpha$ and $\ell_{j,P}$ at a chosen time $t$ from derivatives of the form $t\mapsto \tr{\cS(t)(T(0,t)(O))S}$.  Interpolating these pointwise values in the working MSFS recovers the true schedules with an additional $C^1$ approximation bias controlled by $\rho_m$.

The preceding estimates give a convenient sampling statement.
\begin{prop}\label{prop.explaintrickMSFS}
Let $P_1,\ldots,P_M$ be bounded Hermitian operators on fixed local regions,
let $S_1,\ldots,S_M$ be local Pauli strings, and set
\begin{align}
 f_i(t):=2^{-n}\tr[T(0,t)(P_i)S_i].
\end{align}
Choose an internal Dyson/localization accuracy $\eta_{\rm app}$ and a
working coefficient degree $m$. Let
$d_\epsilon$ be the MSFS degree supplied by
Corollary~\ref{cor:deg_approx} at that accuracy, and abbreviate(q.v.~Def.~\ref{def:Markov_stable})
\begin{align}
 C_*:=C_{\rm int}([0,T],d_\epsilon)\,
 C_{\rm der}([0,T],d_\epsilon).
\end{align}
Define the two common transfer-function approximation errors
\begin{align}
 a_{m,\eta}:=\eta_{\rm app}+\mathfrak b_r\rho_m,
 \qquad b_{m,\eta}:=\eta_{\rm app}+\mathfrak b_r\rho_m .        \label{eq:transfer-function-bias-budgets}
\end{align}
Choose $m$ and $\eta_{\rm app}$ so that
\begin{align}
 C_*a_{m,\eta}+b_{m,\eta}\le\epsilon/2.                        \label{eq:derivative-model-bias-condition}
\end{align}
If such a pair exists, the deterministic enumeration in
Appendix~\ref{app:MSFS} finds it from the known bounds and modulus.
Equation~\eqref{eq:exp-C1-modulus} guarantees existence, and in that case
$d_\epsilon=\poly{m_0,\log(\epsilon^{-1})}$. Let
$N=\poly{d_\epsilon}$ be the number of stable interpolation nodes and set
\begin{align}
 \epsilon_{\rm node}:=
 \frac{\epsilon}{2C_*}.
\end{align}
Expand each $P_i$ in its local Pauli basis and let $C_{\rm sh}$ denote the
largest resulting process-shadow norm factor (a constant for fixed probe
regions). Then all $M$ derivatives can be estimated
uniformly to error $\epsilon$, with probability at least $1-\delta$, using
\begin{align}
 \cO\!\left(
 NC_{\rm sh}\epsilon_{\rm node}^{-2}
 \log\frac{MN}{\delta}\right)                                  \label{eq:clean-derivative-samples}
\end{align}
uses of the forward channel.  The total evolution time is at most $T$ times
this quantity.
\end{prop}
\begin{proof}
Let $g_i\in\mathcal F_{d_\epsilon}$ be the approximation from
Corollary~\ref{cor:deg_approx}; its function and derivative errors are at
most $a_{m,\eta}$ and $b_{m,\eta}$, respectively. At each node, process shadows
and a union bound estimate $f_i$ to $\epsilon_{\rm node}$, so the data differ
from $g_i$ by at most $\epsilon_{\rm node}+a_{m,\eta}$. Stable
interpolation and the Markov inequality therefore give
\begin{align}
 \|\widehat g_i'-f_i'\|_\infty
 \le C_*(\epsilon_{\rm node}+a_{m,\eta})+b_{m,\eta}
 \le\epsilon.
\end{align}
The process-shadow bound
Eq.~\eqref{eq:processshadows} gives Eq.~\eqref{eq:clean-derivative-samples}.
\end{proof}

Our main example of an MSFS is the space of bounded-degree polynomials.
For polynomials we use the deterministic shifted Chebyshev roots in
Eq.~\eqref{eq:det-cheb-nodes}.  Uniformly accurate values give an interpolant
whose uniform error grows by only $\cO(\log N)$, and the Markov inequality then
controls its derivative.  The smallest queried node is
$T\sin^2(\pi/(4N))=\Theta(T/N^2)$.  Robust oversampled regression may be used
instead when outlier tolerance is desired, but no deterministic minimum-time
claim is made for random nodes.
\section{Obtaining a good linear system of equations}\label{sec:finding_equations}

The time-dependent setting poses the challenge that first-order derivatives at $0$ of channel overlaps do not necessarily provide enough information to determine all model parameters. Although higher-order derivatives at $0$ could also be estimated experimentally, this would lead to unfavourable precision scaling through higher-order Markov brothers' inequalities and could make the interpolation system ill-conditioned.

In light of this, we instead estimate derivatives of channel overlaps at nonzero times. This requires choosing pairs of probe operators $P,P'$ such that, at the times $t_i$ where we evaluate
\begin{align}
2^{-n}\tr{\cS(t_i)(P)P'}
\end{align}
leads to a well-conditioned system to estimate the various parameters of $\cS(t_i)$, allowing us to evaluate the time-dependent functions. Then, given access to the value of the various parameters at various $t_i$, we pick interpolation nodes that lead to stable recovery.

In~\cite{StilckFranca.2024}, the authors show that it is possible to find pairs of Pauli strings that lead to a diagonal system of equations for the parameters of the Lindbladian at $t=0$, ensuring the stability of the method. We will now show how it is possible to obtain \emph{diagonally dominant} linear equations for the various parameters of the generator at a given $t$, ensuring similar stability.

 In a nutshell, given a local Pauli $P_{\alpha}$, we seek an initial probe operator $O_{i,\alpha}$ such that, for a given time $t_i\in[0,T]$, $O_{i,\alpha}(t_i)\simeq \lambda P_\alpha$ for some $\lambda\simeq 1$, while the other relevant Pauli components remain small. Locality also implies that only constantly many entries per row and column of the matrix with $(\alpha,\alpha')$ coefficients $2^{-n}\tr{\mathcal{S}(t_i)(O_{i,\alpha}(t_i))P_{\alpha'}}$ are nonzero.
Process-shadow tomography estimates these linear coefficients efficiently. For notational convenience, in this section we suppress the time-node subscript $i$ from $O_{i,\alpha}$ and write $f_O\equiv f$ and $c^{O,t}_\alpha\equiv c^t_\alpha$.

\subsection{Finding good probe operators: the Hamiltonian case}
We will first show how to solve this problem for the case of purely Hamiltonian evolution and then use it to generalize it to Lindbladians. In particular, we prove the following result:
\begin{prop}\label{prop:Oexists}
Let $H(t)$ be a geometrically $k$-local, time-dependent Hamiltonian with
propagator $T_{\mathcal H}(s,t)$.  Let $C_H,\mu_H,v_H$ be the constants
furnished by Corollary~\ref{cor:route2-LR} for the time-reversed negative
Hamiltonian schedule, and define the known enlargement and support radii
\begin{align}
 r_H(P,\epsilon)&:=\min\left\{r\in\mathbb N:r\ge r_0,\ 
 C_H(1+r_P)^De^{-\mu_H(r-r_0)}(e^{v_HT}-1)\le\epsilon\right\},\\
 R_H(P,\epsilon)&:=r_P+r_H(P,\epsilon).                         \label{eq:explicit-H-radius}
\end{align}
For every fixed $t\in[0,T]$, local Pauli $P$ with support of radius $r_P$, and $\epsilon>0$, there is a Hermitian probe operator $O_t$ with $\|O_t\|=1$ supported on a ball of radius $R_H(P,\epsilon)=\mathcal{O}(r_P+r_0+\log(\epsilon^{-1})+v_HT)$ such that
\begin{align}\label{equ:boundnorm}
\|T_{\mathcal H}(0,t)(O_t)-P\|\leq \epsilon,
\end{align}
which implies
\begin{align}\label{equ:bound_overlap}
2^{-n}\tr{T_{\mathcal H}(0,t)(O_t)P}\geq 1-\epsilon.
\end{align}
\end{prop}
\begin{proof}
The inverse of the purely Hamiltonian map $T_{\mathcal H}(0,t)$ exists and
is generated by the time-reversed negative Hamiltonian schedule covered by
Corollary~\ref{cor:route2-LR}.  Thus
\begin{align}\label{equ:LRinverse}
\|T_{\mathcal H}(0,t)^{-1}(P)-
T_{\mathcal H_{S(r_H(P,\epsilon))}}(0,t)^{-1}(P)\|\leq \epsilon.
\end{align}
Define
$O_t:=T_{\mathcal H_{S(r_H(P,\epsilon))}}(0,t)^{-1}(P)$. Then
\begin{align}\label{equ:operator_norm}
\|P-T_{\mathcal H}(0,t)(O_t)\|
&=\|T_{\mathcal H}(0,t)(T_{\mathcal H}(0,t)^{-1}(P)-
T_{\mathcal H_{S(r_H(P,\epsilon))}}(0,t)^{-1}(P))\|\\
&\leq \|T_{\mathcal H}(0,t)^{-1}(P)-
T_{\mathcal H_{S(r_H(P,\epsilon))}}(0,t)^{-1}(P)\|\leq \epsilon,
\end{align}
as the unitary channel $T_{\mathcal H}(0,t)$ contracts the operator norm.
Note that clearly $\|O_t\|=1$, as the unitary evolution preserves the operator norm.
The claim in Eq.~\eqref{equ:bound_overlap} follows by combining H\"older's inequality with Eq.~\eqref{equ:operator_norm}, as:
\begin{align}
2^{-n}\tr{T_{\mathcal H}(0,t)(O_t)P}
=1+2^{-n}\tr{(T_{\mathcal H}(0,t)(O_t)-P)P}\geq 1-\epsilon
\end{align}

\end{proof}
The proposition above guarantees a local probe that the forward evolution maps close to a local Pauli. The next proposition shows that such a probe can be found efficiently from local transfer-matrix estimates when the target precision $\epsilon$ is constant.

\begin{prop}[Local inverse from process shadows]\label{prop:good_obs}
In the setting of Prop.~\ref{prop:Oexists}, fix a time $t$, target Pauli
$P$, an inner region $S_0$ supporting the candidate probe, and an enlarged
output region $\Lambda$.  Let $\Pi_\Lambda$ be the conditional expectation
onto the matrix algebra on $\Lambda$.  Suppose
\begin{align}
 \|T(0,t)(O)-\Pi_\Lambda T(0,t)(O)\|\le 2\eta_{\rm LR}
 \quad\text{for all Hermitian }O\text{ on }S_0\text{ with }\|O\|\le1.
                                                                    \label{eq:sdp-uniform-tail}
\end{align}
Estimate every input/output Pauli-transfer entry of
$\Pi_\Lambda T(0,t)|_{\mathcal A(S_0)}$ to accuracy
$\eta_{\rm tom}4^{-(|S_0|+|\Lambda|)}$, symmetrize the resulting real map
$\widehat M$, and solve the SDP in Eq.~\eqref{eq:our_sdp} to objective
accuracy $\eta_{\rm opt}$.  If a feasible $O_*$ obeys
$\|T(0,t)(O_*)-P\|\le\eta_{\rm inv}$, then, with probability at least
$1-\delta$, its returned solution $O$ satisfies
\begin{align}
 \|T(0,t)(O)-P\|
 &\le \eta_{\rm inv}+4\eta_{\rm LR}+2\eta_{\rm tom}+\eta_{\rm opt},
                                                                    \label{eq:correct-sdp-budget}\\
 2^{-n}\tr[T(0,t)(O)P]
 &\ge 1-\eta_{\rm inv}-4\eta_{\rm LR}-2\eta_{\rm tom}-\eta_{\rm opt}.
\end{align}
The required number of samples and the postprocessing cost are exponential
only in $|S_0|+|\Lambda|$ and polynomial in the inverse error parameters and
$\log(\delta^{-1})$.
\end{prop}
\begin{proof}
The uniform transfer-matrix event and the Pauli expansion give
$\|\widehat M(O)-\Pi_\Lambda T(0,t)(O)\|\le\eta_{\rm tom}$ for every
feasible $O$, including the data-dependent optimizer.  Hence $O_*$ has
empirical objective at most
$\eta_{\rm inv}+2\eta_{\rm LR}+\eta_{\rm tom}$.  Optimality up to
$\eta_{\rm opt}$ and a second application of the same two uniform estimates
give Eq.~\eqref{eq:correct-sdp-budget}.  The overlap bound follows from
H\"older's inequality.  Process-shadow concentration and a union bound over
the $4^{|S_0|+|\Lambda|}$ transfer entries give the stated resource scaling.
The use of two regions is essential: it explicitly accounts for the evolved
tail outside the SDP input support.
\end{proof}

Having constructed probes whose forward images have large overlap with a single Pauli, we now use them to obtain well-conditioned linear systems for the parameter values at the selected times.
One important result to this end will be the following proposition, proved in App.~\ref{app:stab_infty}:
\begin{prop}\label{prop:stable_systems}
Let $A\in\mathbb{R}^{m\times m}$ be a matrix that is row and column diagonally dominant, that is 
\begin{align}
\forall i\in[m]:|A_{i,i}|\geq 0.75,\quad |A_{i,i}|\geq 0.5+\sum_{j\not=i}|A_{i,j}|,\quad |A_{i,i}|\geq 0.5+\sum_{j\not=i}|A_{j,i}|
\end{align}
Let $\Omega\subseteq[m]\times[m]$ be a known structural mask with at most
$s=\cO(1)$ positions in each row and column, and suppose
$\supp(A),\supp(B)\subseteq\Omega$.  Assume every entry of
$B\in\mathbb R^{m\times m}$ has magnitude at most
$0<\zeta=\cO(1)$.
Furthermore, let $b'=b+y$, with $y$ a vector with $\|y\|_{\infty}\leq \zeta$ and $\|b\|_{\infty}=\cO(1)$.
Then $A$ and $A+B$ are both invertible and
\begin{align}
\|(A+B)^{-1}b'-A^{-1}b\|_{\infty}=\cO(s\zeta)
\end{align}
for $\zeta\leq \tfrac{1}{10s}$.
\end{prop}
Let us now outline our strategy for the rest of the section, which will also motivate the relevance of the above proposition to our setting. In the last section, we saw that we can efficiently estimate time-evolved traces of the form $t\mapsto 2^{-n}\tr{T(0,t)(O)P}$ and their derivatives for a given $t$ by resorting to interpolation. Now we can expand the derivatives as follows:
\begin{align}
\sum_{\alpha\in \mathcal A}h_\alpha(t)2^{-n}\tr{\mathcal K^H_{\alpha}(O(t))P}.
\end{align}
Thus, if we know $2^{-n}\tr{\mathcal K^H_{\alpha}(O(t))P}$ and the full
derivative at that point, we obtain a real linear equation for the values
$h_{\alpha}(t)$. Solving at the interpolation nodes recovers the full
coefficient functions.

We know both the derivative and the transfer entries only approximately. We
therefore use Prop.~\ref{prop:stable_systems}; its hypotheses require probe
pairs for which the matrix of
$2^{-n}\tr{\mathcal K^H_{\alpha}(O(t))P}$ is sparse and diagonally dominant.
To achieve this, we will resort to the following concept:

\begin{definition}[$s$-stable Pauli strings]\label{def:stable_paulis}
Given a basis of the Hamiltonian working terms $\{\mathcal K^H_{\alpha}\}_{\alpha\in\mathcal{A}}$ , we say that a collection of pairs of Pauli strings $\{(P_{\alpha,1},P_{\alpha,2})\}_{\alpha\in\mathcal{A}}$, each associated to a non-zero interaction term $h_\alpha$, is $s$-stable if for all $\alpha$ we have that for any other $\alpha'\not=\alpha$: 
\begin{align}\label{cond:sstability}
\forall \alpha'\neq \alpha,\,\tr{\mathcal K^H_{\alpha'}(P_{\alpha,1})P_{\alpha,2}}=0\qquad \text{ and }\qquad 2^{-n}\tr{\mathcal K^H_{\alpha}(P_{\alpha,1})P_{\alpha,2}}=1,
\end{align}
 In addition, we require:
\begin{align}\label{equ:row_columns}
\max_{\alpha\in\mathcal{A}} |\{\alpha':\mathcal K^H_{\alpha}(P_{\alpha',2})\ne 0\}|\leq s\\
\max_{\alpha\in\mathcal{A}} |\{\alpha':\mathcal K^H_{\alpha'}(P_{\alpha,2})\not=0\}|\leq s.
\end{align}
\end{definition}
\noindent The next result shows the existence of such $s$-stable Pauli sets with $s$ of constant order, whose proof we leave to Prop.~\ref{prop:size_stable} of App.~\ref{app:stable_pauli}: 
\begin{prop}\label{prop:size_s_stable}
Let $\cS(t)$ be the generator of a geometrically $k$-local unitary dynamics on a graph $\mathsf{G}$ of effective dimension $D$ with terms $\{h_{\alpha}\mathcal K^H_\alpha\}_{\alpha\in\mathcal{A}}$. Then there exists a set of $s$-local stable Pauli strings $\{(P_{\alpha,1},P_{\alpha,2})\}_{\alpha\in\mathcal{A}}$ with:
\begin{align}
s=\mathcal{O}(b_k4^{b_k}).
\end{align}
Furthermore, $P_{\alpha,1}$ can be chosen to be $1$-local and $P_{\alpha,2}$ $k$-local.
\end{prop}
Computing the resulting system of linear equations we obtain from this data can also be done in time polynomial in the number of parameters.

\begin{thm}\label{equ:good_equations}

Let $\{\mathcal K^H_\alpha\}_{\alpha\in\mathcal{A}}$ be a basis of geometrically $k$-local Hamiltonian terms on a graph $\mathsf{G}$ of effective dimension $D$ with corresponding pairs of $s$-stable Pauli strings $\{(P_{\alpha,1},P_{\alpha,2})\}_{\alpha\in\mathcal{A}}$ as in Prop.~\ref{prop:size_s_stable}. Then, for $t_i\le T=\mathcal{O}(1)$ we can find a system of linear equations $A\in\mathbb{R}^{|\mathcal{A}|\times |\mathcal{A}|}$ for $\{h_\alpha(t_i)\}_{\alpha\in\mathcal{A}}$ s.t.~$A$
is a diagonally dominated system of linear equations satisfying the condition of Prop.~\ref{prop:stable_systems}. In addition, we can estimate a matrix $A+B$ with $B$ as in Prop.~\ref{prop:stable_systems} with probability of success at least $1-\delta$ from 
\begin{align}
S=\cO\left(\zeta^{-2}e^{\widetilde{\mathcal{O}}((vT)^Dk^{D^2})}\log(n/\delta)\right)
\end{align}
samples, where $\widetilde{\mathcal{O}}$ hides polylogarithmic dependencies on $k$. 
\end{thm}
\begin{proof}

Let $V(\epsilon)=\max_{\alpha\in\mathcal{A}}|S_{\alpha}(\epsilon)|=\mathcal{O}((k+vT+\log\epsilon^{-1})^D)$, where $S_{\alpha}(\epsilon)$ denotes the set of vertices at distance at most $vT+\log(\epsilon^{-1})$ from the support of $P_{\alpha,1}$.
Our matrix $A$ will be constructed as follows.  Apply
Prop.~\ref{prop:good_obs} with its inverse, LR, tomography, and optimization
budgets chosen so that their total in
Eq.~\eqref{eq:correct-sdp-budget} is at most $1/(10s^2)$. Enlarge each
output region $\Lambda_\alpha$ to contain every nonzero Pauli
$\mathcal K_{\alpha'}^{H,*}(P_{\alpha,2})$ permitted by the known
structural mask. For each pair
$(P_{\alpha,1},P_{\alpha,2})$ this gives $O_\alpha$ supported on a
constant-size inner region such that
\begin{align}\label{equ:large_overlap2}
\|P_{\alpha,1}-O_{\alpha}(t_i)\|\leq 1/(10s^2).
\end{align}
Index the rows and columns of $A$ by $\alpha\in \cA$. Each row of $A$ will be given by the linear equation we obtain for the derivative of $O_\alpha$ at $t_i$. 
Let us first show that each row and column of $A$ contains at most $s$ nonzero entries. This follows from the fact that we have a $s$-stable set of Pauli strings, as it is not too difficult to see that Eq.~\eqref{equ:row_columns} implies that the number of nonzero entries per row and column are at most $s$. Indeed, note that entry $(\alpha,\alpha')$ of the matrix $A$ is given by:
\begin{align}
A_{\alpha,\alpha'}:=2^{-n}\tr{\mathcal K^H_{\alpha'}(O_{\alpha}(t_i))P_{\alpha,2}}=-2^{-n}\tr{O_{\alpha}(t_i)\mathcal K^H_{\alpha'}(P_{\alpha,2})}
\end{align}
From this expression we can easily read off from Eq.~\eqref{equ:row_columns} that the row and column sparsity of $A$ is at most $s$. Now let us show that $A$ is diagonally dominating both for rows and columns. Let
\begin{align}
O_{\alpha}(t_i)=\sum_{P}c^{\alpha}_{P}P,\qquad \text{ where }\qquad c^{\alpha}_{P}=2^{-n}\tr{O_{\alpha}(t_i)P},
\end{align}
be the expansion of $O_{\alpha}(t_i)$ into Pauli strings. It follows from Eq.~\ref{equ:large_overlap2} that for $P\not=P_{\alpha,1}$, 
\begin{align}\label{eq:overlapcap}
|c^{\alpha}_{P}|=\left|2^{-n}\tr{PO_{\alpha}(t_i)}\right|=\left|2^{-n}\tr{P(O_{\alpha}(t_i)-P_{\alpha,1})}\right|\leq \|O_{\alpha}(t_i)-P_{\alpha,1}\|\leq 1/(10s^2).
\end{align}
Thus, we have that for all $\alpha$:
\begin{align}
&|A_{\alpha,\alpha}|\geq1- 1/(10s^2)
\end{align}
and
\begin{align}
\sum\limits_{\alpha'\ne\alpha}|A_{\alpha,\alpha'}|&\le 2^{-n}\sum_{\alpha'\ne \alpha}\sum_P |c^\alpha_P|\,\left|\tr{P\mathcal K^H_{\alpha'}(P_{\alpha,2})}\right|\\
&=2^{-n}\sum_{\alpha'\ne \alpha}\sum_{P\ne P_{\alpha,1}} |c^\alpha_P|\,\left|\tr{P\mathcal K^H_{\alpha'}(P_{\alpha,2})}\right|\\
&\le \frac{s}{10s^2}+0\\
&=\frac{1}{10s}.
\end{align}
where the last inequality follows from \eqref{eq:overlapcap} and 
the condition of $s$-stability \eqref{cond:sstability}. Similarly, we get that 
\begin{align*}
\sum\limits_{\alpha'\ne\alpha}|A_{\alpha',\alpha}|\le \frac{1}{10s}.
\end{align*}
Thus, as $1/(10s)\leq 0.25$, we conclude that the conditions of Prop.~\ref{prop:stable_systems} are fulfilled for $A$.
Now let us show that the conditions for the matrix $B$ are fulfilled as well. As we know the stable Pauli strings and $\mathcal K^H_\alpha$, we can determine which entries of the matrix $A$ can be nonzero: these are the pairs for which $\mathcal K^H_{\alpha'}(P_{\alpha,2})\ne0$.  Its nonzero value is a known signed Pauli $P^{\alpha'}_{\alpha}$, and locality bounds its weight by $2k$.
To determine an entry, it suffices to estimate
$2^{-n}\tr[T(0,t_i)(O_\alpha)P^{\alpha'}_\alpha]$.  Expand only the
\emph{input} $O_\alpha$ in the Pauli basis of its inner support and estimate
the corresponding transfer entries to accuracy
$\zeta/4^{|S_{0,\alpha}|}$.  No claim that the evolved probe itself has
finite support is needed.  The uniform transfer-matrix event in
Prop.~\ref{prop:good_obs}, on the enlarged output region just specified,
simultaneously covers the data-dependent SDP output and all these matrix
entries. Its two constant-size regions, the
process-shadow bound Eq.~\eqref{eq:processshadows}, and
$|\mathcal A|\le n4^{b_k}$ give the stated
$\zeta^{-2}e^{\widetilde\cO((vT)^Dk^{D^2})}\log(n/\delta)$ scaling.
\end{proof}

Thus, the results of this section show that we can get a highly structured and stable system of linear equations to evaluate the time-dependent functions.

\subsection{Extension to weak onsite Pauli dephasing}

We now add the physical dissipator
\begin{align}
 \mathcal L(t)=\sum_{j\in V}\sum_{P\in\{X,Y,Z\}}
 \ell_{j,P}(t)\mathcal L_{j,P},
 \qquad 0\le \ell_{j,P}(t)\le\tau.                              \label{eq:physical-pauli-dissipator}
\end{align}
All locality and complete-positivity arguments use this physical basis and its
nonnegative rates.

\begin{prop}[Weakly dissipative local probes]\label{prop:OexistsLindblad}
Let $T_{\mathcal S}$ be generated by Eq.~\eqref{eq:thegenerator}.  For every
fixed $t\in[0,T]$, local Pauli $P$, and $\eta_{\rm inv}>0$, there is a
Hermitian $O_t$ with $\|O_t\|\le1$, supported in a ball of the known radius
$R=R_H(P,\eta_{\rm inv})$ from Eq.~\eqref{eq:explicit-H-radius},
such that
\begin{align}
 \|T_{\mathcal S}(0,t)(O_t)-P\|
 \le \eta_{\rm inv}+\Delta_{\rm diss}(R,t),\qquad
 \Delta_{\rm diss}(R,t):=C_{\rm cmp}\tau(1+R)^{2D-1}
 (e^{\mu't}-1).                                                  \label{equ:bound_error_dissipative}
\end{align}
\end{prop}
\begin{proof}
Apply Prop.~\ref{prop:Oexists} to the Hamiltonian propagator and use
Eq.~\eqref{eq:comparison-assumption} to compare it with the full forward
channel.  The constant absorbs the three onsite Pauli directions per site.
\end{proof}

Combining this candidate with Prop.~\ref{prop:good_obs} gives the following
computed-probe guarantee.  It retains the dissipative error, which cannot be
absorbed into an arbitrarily small target accuracy at fixed $\tau T$.
\begin{prop}\label{prop:good_obsdissip}
For a fixed $t$ and Pauli $P$, the two-region SDP of
Prop.~\ref{prop:good_obs} returns, on its uniform shadow event, a probe $O$
satisfying
\begin{align}
 \|T_{\mathcal S}(0,t)(O)-P\|
 \le \eta_{\rm inv}+\Delta_{\rm diss}(R,t)
 +4\eta_{\rm LR}+2\eta_{\rm tom}+\eta_{\rm opt}.                \label{eq:dissipative-sdp-budget}
\end{align}
Its Pauli overlap is at least one minus the same right-hand side.  The sample
and classical costs are exponential only in the two constant-size SDP
regions and polynomial in the inverse error budgets and
$\log(\delta^{-1})$.
\end{prop}

For the learning system only, define the algebraic projector basis
\begin{align}
 \widetilde{\mathcal L}_{j,P}
 :=\frac12\bigl(\mathcal L_{j,P}-\mathcal L_{j,P_1}
                         -\mathcal L_{j,P_2}\bigr),              \label{eq:tilde-L-definition}
\end{align}
where $P_1,P_2$ are the other two single-qubit Paulis.  In vector notation,
\begin{align}
 \widetilde L_j=\Gamma_j L_j,\qquad
 \ell_j=\Gamma_j\widetilde\ell_j,\qquad
 \Gamma_j:=\frac12
 \begin{pmatrix}1&-1&-1\\-1&1&-1\\-1&-1&1\end{pmatrix}.       \label{eq:Gamma-basis-change}
\end{align}
Thus
$\mathcal L(t)=\sum_{j,P}\widetilde\ell_{j,P}(t)
\widetilde{\mathcal L}_{j,P}$.  The transformed coefficients may be signed;
the transformed basis is not used to assert physicality, contractivity, or
locality.

Index the working basis by
\begin{align}
 \{\mathcal K_\beta\}_\beta
 =\{\mathcal K^H_\alpha=-i\mathcal P_\alpha\}_{\alpha\in\mathcal A}
 \cup\{\widetilde{\mathcal L}_{j,P}\}_{j,P}.                    \label{eq:working-generator-basis}
\end{align}
For a Hamiltonian row set
$(Q_\beta,\overline Q_\beta)=(P_{\alpha,1},P_{\alpha,2})$; for a
dissipative row $(j,P)$ set
$(Q_\beta,\overline Q_\beta)=(P_j,P_j)$.  Known signs in the Hamiltonian
pairs are normalized so that
$2^{-n}\tr[\mathcal K_\beta(Q_\beta)\overline Q_\beta]=1$.
For each row define the bounded structural output region
\begin{align}
 U_\beta:=\bigcup_{\beta':\,\mathcal K_{\beta'}^*(\overline Q_\beta)\ne0}
 \supp\!\left(\mathcal K_{\beta'}^*(\overline Q_\beta)\right).   \label{eq:structural-output-region}
\end{align}
Whenever Prop.~\ref{prop:good_obs} is used for this row, its output region
$\Lambda_\beta$ is enlarged to contain $U_\beta$.  Locality and the
structural mask imply that this enlargement has constant size.

\begin{prop}[$s$-stable working basis]\label{prop:size_stabledissipmain}
Let $s$ be the computable maximum number of structurally permitted entries
in any row or column for the pairs above. Then
$s\le C_s b_k4^{b_k}$ for a geometry-dependent constant $C_s$, and
their ideal transfer matrix
$2^{-n}\tr[\mathcal K_{\beta'}(Q_\beta)\overline Q_\beta]$ is the
identity.  In particular, Hamiltonian--dissipative cross entries vanish at
the ideal probes.
\end{prop}
The support-counting proof is given in
Prop.~\ref{prop:size_stabledissip}.

\begin{thm}[Stable system in the weak-dephasing regime]\label{thm:systemAnoisy}
Let $s$ be the computable integer in
Prop.~\ref{prop:size_stabledissipmain}, set
$\eta_{\rm inv}:=1/(60s^2)$, and define
\begin{align}
 R_s:=1+\max_\beta R_H(Q_\beta,\eta_{\rm inv}).                 \label{eq:exact-Rs}
\end{align}
Assume
\begin{align}\label{equ:tau_T_interplay}
 C_{\rm cmp}\tau R_s^{2D-1}(e^{\mu'T}-1)\le\frac1{40s}.
\end{align}
Choose the SDP budgets so that
\begin{align}
 4\eta_{\rm LR}+2\eta_{\rm tom}+\eta_{\rm opt}\le\frac1{24s}.   \label{eq:isolator-budgets}
\end{align}
For every interpolation node $t_i$, let $O_{i,\beta}$ be the probes from
Prop.~\ref{prop:good_obsdissip} and define
\begin{align}
 \widetilde A^i_{\beta,\beta'}
 &:=2^{-n}\tr\!\left[
 \mathcal K_{\beta'}(T(0,t_i)(O_{i,\beta}))\overline Q_\beta
 \right]                                                        \label{eq:correct-working-matrix}\\
 &=2^{-n}\tr\!\left[T(0,t_i)(O_{i,\beta})
 \mathcal K_{\beta'}^*(\overline Q_\beta)\right],\nonumber\\
 b^i_\beta
 &:=2^{-n}\tr\!\left[\mathcal S(t_i)(T(0,t_i)(O_{i,\beta}))
 \overline Q_\beta\right].                                    \label{eq:correct-working-rhs}
\end{align}
Then $b^i=\widetilde A^i\widetilde\theta(t_i)$ with
$\widetilde\theta=(h,\widetilde\ell)$, and $\widetilde A^i$ satisfies the
row/column diagonal-dominance and sparsity conditions of
Prop.~\ref{prop:stable_systems}.  All potentially nonzero entries of
$\widetilde A^i$ can be estimated, with
probability at least $1-\delta$, using
\begin{align}
 \cO\!\left(\zeta^{-2}
 e^{\widetilde{\cO}(R_s^Dk^{D^2})}\log(n/\delta)\right)
\end{align}
samples. The derivative right-hand side $b^i$ is estimated with the auxiliary
MSFS grid and polynomial overhead described in Sec.~\ref{sec:sampcomplexityanalysis}.
One solves this working-basis system directly and outputs the
physical rates by $\ell_j=\Gamma_j\widetilde\ell_j$; no conjugated learning
matrix is formed.
\end{thm}
\begin{proof}
Equations~\eqref{equ:tau_T_interplay}--\eqref{eq:isolator-budgets} and
Prop.~\ref{prop:good_obsdissip} give
\begin{align}
 \|T(0,t_i)(O_{i,\beta})-Q_\beta\|\le\frac1{12s}.
\end{align}
The stable-pair identities therefore imply
$|\widetilde A^i_{\beta,\beta}|\ge1-1/(12s)$ and row and column
off-diagonal sums at most $1/12$.  The matrix is real because every
$\mathcal K_\beta$ maps Hermitian matrices to Hermitian matrices.  Sparse
support counting gives at most $s$ entries per row and column.  Process
shadows estimate the uniformly bounded local transfer entries in parallel;
Prop.~\ref{prop:good_obs} accounts for the output tails and data-dependent
SDP solutions.  These facts give the sample bound and the hypotheses of
Prop.~\ref{prop:stable_systems}.  Finally,
Eq.~\eqref{eq:Gamma-basis-change} proves the output conversion.  Since
$\|\Gamma_j\|_{\ell_\infty\to\ell_\infty}=3/2$, estimating
$\widetilde\ell_j$ within $2\epsilon/3$ gives physical-rate error at most
$\epsilon$; clipping the final rates to $[0,\tau]$ cannot increase their
error.
\end{proof}

The smallness condition is a theorem hypothesis, not merely a proof
artifact. For one-qubit $Z$-dephasing,
$T_t(X)=e^{-\gamma t}X$, so no $\|O\|\le1$ can be mapped closer than
$1-e^{-\gamma t}$ to $X$.  Bistochasticity alone therefore does not provide
bounded inverse probes at strong dissipation.

\section{The learning protocol}

We now give the corrected protocol in the working basis
Eq.~\eqref{eq:working-generator-basis}.  Its inputs are the known local
Hamiltonian supports, the onsite Pauli-dephasing directions, the MSFS
$\mathcal F_{m_\epsilon}$, the certified modulus $\rho_m$, the bounds in
Problem~\ref{prob:time_dependent}, and target
accuracy and failure probability $(\epsilon,\delta)$.  The promised
weak-dissipation condition Eq.~\eqref{equ:tau_T_interplay} is checked as an
assumption of the theorem; empirical conditioning is useful diagnostically
but is not a substitute for this promise.

\subsection{Error and sample budgets}\label{sec:sampcomplexityanalysis}

Let $N_1=\poly{m_\epsilon}$ be the number of primary interpolation nodes for the
coefficient functions.  For the polynomial specialization we choose the
deterministic roots in Eq.~\eqref{eq:det-cheb-nodes}.  At each primary node
$t_i$, Theorem~\ref{thm:systemAnoisy} gives
\begin{align}
 b^i=\widetilde A^i\widetilde\theta(t_i),\qquad
 \widetilde\theta=(h,\widetilde\ell).                            \label{eq:protocol-working-system}
\end{align}
Write $C_{\rm p}:=C_{\rm int}([0,T],m_\epsilon)$, and let
$C_{\rm stab}$ be the known constant from
Prop.~\ref{prop:stable_systems} that converts the common entrywise error into
the $\ell_\infty$ error of the solved working coefficients.  We impose the
primary model-bias condition
\begin{align}
 (1+C_{\rm p})\rho_{m_\epsilon}\le\epsilon/2,                  \label{eq:primary-model-bias-condition}
\end{align}
and estimate every potentially nonzero entry of $\widetilde A^i$ and $b^i$
to the common precision
\begin{align}
 \epsilon_{\rm lin}:=
 \min\!\left\{\frac{\epsilon}{3C_{\rm p}C_{\rm stab}},
                    \frac{1}{10s}\right\}.                   \label{eq:protocol-linear-precision}
\end{align}
The right-hand side accounts for the additional factor
$\|\Gamma_j\|_{\ell_\infty\to\ell_\infty}=3/2$ when working-rate node
estimates are converted to physical-rate node estimates.

For each fixed reconstructed probe $O_{i,\beta}$, set
$\eta_{\rm app}:=\eta_\epsilon$, $r:=r_{q_\epsilon}$, and
$d_{\rm aux}:=d_\epsilon$ according to
Eqs.~\eqref{eq:approx-working-degree}--
\eqref{eq:effective-complexity-degree}. Corollary
\ref{cor:deg_approx} then supplies an auxiliary MSFS approximant for
\begin{align}
 f_{i,\beta}(u):=2^{-n}\tr[T(0,u)(O_{i,\beta})\overline Q_\beta].
                                                                    \label{eq:protocol-trace-function}
\end{align}
Writing
$C_{\rm aux}:=C_{\rm int}([0,T],d_{\rm aux})
C_{\rm der}([0,T],d_{\rm aux})$, the selected pair obeys
\begin{align}
 C_{\rm aux}(\eta_{\rm app}+\mathfrak b_r\rho_{m_\epsilon})
 +(\eta_{\rm app}+\mathfrak b_r\rho_{m_\epsilon})
 \le\epsilon_{\rm lin}/2.                                    \label{eq:auxiliary-model-bias-condition}
\end{align}
This is precisely the last feasibility condition in
Eq.~\eqref{eq:approx-working-degree}. We use a common deterministic stable grid of
$N_2=\poly{d_{\rm aux}}$ auxiliary times and estimate the values to precision
\begin{align}
 \epsilon_1
 =\frac{\epsilon_{\rm lin}}{2C_{\rm aux}}.                       \label{eq:eps1}
\end{align}
The same three-term budget as in Prop.~\ref{prop.explaintrickMSFS} then
yields
$|\widehat f'_{i,\beta}(t_i)-f'_{i,\beta}(t_i)|\le
\epsilon_{\rm lin}$.

\subsection{Classical preprocessing}

\begin{algorithm}[H]
\caption{Classical preprocessing}\label{alg:preprocessing}
\begin{algorithmic}[1]
\Require Physical generator ansatz in Eq.~\eqref{eq:thegenerator},
MSFS hierarchy $(\mathcal F_m)_{m\ge0}$, certified modulus $\rho_m$,
$(\epsilon,\delta,T,\tau)$, and the graph data $(k,D,C_1)$ from which the
locality constants in Eq.~\eqref{eq:main-LR-constant-choice} are computed
\Ensure Stable pairs, time grids, local SDP regions, and sample budgets
\State Construct the real working basis
$\{\mathcal K^H_\alpha,\widetilde{\mathcal L}_{j,P}\}$ and the normalized
stable pairs $(Q_\beta,\overline Q_\beta)$.
\State Verify that Eq.~\eqref{equ:tau_T_interplay} is part of the input
promise; set the isolator budgets as in Eq.~\eqref{eq:isolator-budgets}.
\State Enumerate $(m,q)$ and select $(m_\epsilon,q_\epsilon)$,
$d_\epsilon$, and $m_{\rm eff}$ as in
Eqs.~\eqref{eq:approx-working-degree}--
\eqref{eq:effective-complexity-degree}; verify
Eqs.~\eqref{eq:primary-model-bias-condition} and
\eqref{eq:auxiliary-model-bias-condition}.
\State Choose $N_1$ primary and $N_2$ auxiliary deterministic stable nodes.
  \State For each target, choose an inner input region $S_{0,\beta}$ and a
larger output region $\Lambda_\beta$ so that
Eq.~\eqref{eq:sdp-uniform-tail} holds with the allocated $\eta_{\rm LR}$ and
$U_\beta\subseteq\Lambda_\beta$ as in
Eq.~\eqref{eq:structural-output-region}.
\State Set the transfer-entry, derivative, and union-bound precisions using
Eqs.~\eqref{eq:correct-sdp-budget}, \eqref{eq:protocol-linear-precision}, and
\eqref{eq:eps1}.
\State \Return the stable pairs, regions, grids, and per-time sample counts.
\end{algorithmic}
\end{algorithm}

The stable-pair construction and all support enumeration are local, so the
preprocessing time is polynomial in $n,m_{\rm eff},\log(\delta^{-1})$, and
$\epsilon^{-1}$.  For fixed $(k,D,T)$ and the constant isolator budgets, all
SDPs have constant dimension.

\subsection{Quantum data acquisition}

\begin{algorithm}[H]
\caption{Forward-channel process-shadow data}\label{alg:quantum_data}
\begin{algorithmic}[1]
\Require The union of primary and auxiliary times and their sample counts
\Ensure Process-shadow snapshots at every requested time
\For{each requested time $u$}
  \For{each prescribed snapshot}
    \State Prepare a uniformly random product Pauli eigenstate.
    \State Apply the physical forward channel $T(0,u)$.
    \State Measure in a uniformly random product Pauli basis and record the
    input labels, input signs, measurement labels, and outcomes.
  \EndFor
\EndFor
\end{algorithmic}
\end{algorithm}

These data estimate all required overlaps
$2^{-n}\tr[Q T(0,u)(P)]$ in parallel.  We use either independent batches for
probe construction and evaluation, or, as in the analysis here, a single
uniform concentration event over every enumerated input/output Pauli entry.
The latter event holds for every feasible SDP output and therefore handles
the data dependence of $O_{i,\beta}$.

\subsection{Local probes and derivative estimates}

\begin{algorithm}[H]
\caption{Two-region local probe reconstruction}\label{alg:rev_from_shadows}
\begin{algorithmic}[1]
\Require Stable pairs, primary-time shadows, $S_{0,\beta}\subset\Lambda_\beta$
\Ensure Probes $O_{i,\beta}$ obeying Eq.~\eqref{eq:dissipative-sdp-budget}
\For{each primary time $t_i$ and row $\beta$}
  \State Estimate all transfer entries with input Paulis supported on
  $S_{0,\beta}$ and output Paulis supported on $\Lambda_\beta$.
  \State Symmetrize the real empirical map $\widehat M_{i,\beta}$ so that it
  maps Hermitian inputs to Hermitian outputs.
  \State Solve
  $\min_{-I\preceq O\preceq I,\,\supp(O)\subseteq S_{0,\beta}}
  \|\widehat M_{i,\beta}(O)-Q_\beta\|_\infty$
  to accuracy $\eta_{\rm opt}$ and store $O_{i,\beta}$.
\EndFor
\end{algorithmic}
\end{algorithm}

The empirical map need not be completely positive; Hermiticity preservation
is sufficient for the SDP linear matrix inequalities.  Proposition
\ref{prop:good_obs} gives a guarantee for the full, rather than silently
truncated, output channel.

\begin{algorithm}[H]
\caption{Derivative estimation by stable interpolation}\label{alg:derivative_estimation}
\begin{algorithmic}[1]
\Require Probes $O_{i,\beta}$, auxiliary-time shadows, and
$\overline Q_\beta$
\Ensure Estimates $\widehat b^i_\beta$ of Eq.~\eqref{eq:correct-working-rhs}
\For{each $(i,\beta)$}
  \State At every auxiliary node $u_j$, use the Pauli expansion of
  $O_{i,\beta}$ and process shadows to estimate
  $f_{i,\beta}(u_j)$ in Eq.~\eqref{eq:protocol-trace-function}.
  \State Stably interpolate these values in the auxiliary MSFS and set
  $\widehat b^i_\beta:=\widehat f'_{i,\beta}(t_i)$.
\EndFor
\end{algorithmic}
\end{algorithm}

\subsection{Working-basis solve and physical output}\label{subsec:classical-post}

For every candidate column $\beta'$, precompute
\begin{align}
 \mathcal K_{\beta'}^*(\overline Q_\beta)
 =a_{\beta,\beta'}P_{\beta,\beta'},\qquad
 a_{\beta,\beta'}\in\{0,\pm1\},                                 \label{eq:real-pauli-neighbor}
\end{align}
where the Pauli is needed only when $a_{\beta,\beta'}\ne0$.  Thus the
correct matrix entry is the known real sign times a single transfer overlap.

\begin{algorithm}[H]
\caption{Recover working coefficients and physical rates}\label{alg:solve_and_interpolate}
\begin{algorithmic}[1]
\Require Primary-time shadows, probes, stable pairs, and derivative estimates
\Ensure Interpolants $\widehat h_\alpha$ and physical rates
$\widehat\ell_{j,P}$
\For{each primary time $t_i$}
  \For{each potentially nonzero $(\beta,\beta')$}
    \State Estimate
    $c^i_{\beta,\beta'}=2^{-n}\tr[T(0,t_i)(O_{i,\beta})
    P_{\beta,\beta'}]$ and set
    $\widehat{\widetilde A}^i_{\beta,\beta'}
    =a_{\beta,\beta'}\widehat c^i_{\beta,\beta'}$.
  \EndFor
  \State Solve the sparse diagonally dominant system
  $\widehat b^i=\widehat{\widetilde A}^i
  \widehat{\widetilde\theta}^{,i}$ directly.
\EndFor
\State At every primary node set
$\widehat\ell_j^{,i}=\Gamma_j
\widehat{\widetilde\ell}_j^{,i}$, then stably interpolate the Hamiltonian
and physical-rate node values in $\mathcal F_{m_\epsilon}$.
\State Clip each $\widehat h_\alpha(t)$ pointwise to $[-1,1]$.
\State Clip each interpolated physical rate pointwise to $[0,\tau]$.
\State \Return $\{\widehat h_\alpha\}$ and $\{\widehat\ell_{j,P}\}$.
\end{algorithmic}
\end{algorithm}

There is no rotation or conjugation of the measured learning matrix.  The
conversion to physical-rate node values occurs only after the working system
is solved.

\subsection{Complexity and completion of the proof}

At a fixed time, process shadows estimate all enumerated constant-weight
transfer entries using
\begin{align}
 \widetilde\cO\!\left(\epsilon_1^{-2}\log(n/\delta)\right)
\end{align}
samples, where constants depend on the fixed local supports.  The number of
primary and auxiliary nodes and the precision inflation in
Eq.~\eqref{eq:eps1} are polynomial in
$m_{\rm eff}$ and $\log(\epsilon^{-1})$.  A union bound over all nodes, rows, columns,
and SDP transfer entries therefore gives
\begin{align}
 S=\widetilde\cO\!\left(
 \epsilon^{-2}\poly{m_{\rm eff}}\log(n\delta^{-1})\right).
\end{align}
Proposition~\ref{prop:stable_systems} propagates the entrywise errors through
the working-basis solve.  Let $e_{\rm node}\le
C_{\rm stab}\epsilon_{\rm lin}$ be the resulting working-coefficient error.
After converting the dissipative node values, their error is at most
$3e_{\rm node}/2$.  For any true physical coefficient $\theta$ and its
witness $p_{\theta,m_\epsilon}$ from
Eq.~\eqref{eq:C1-MSFS-modulus}, stable interpolation therefore gives
\begin{align}
 \|\widehat\theta-\theta\|_\infty
 \le (1+C_{\rm p})\rho_{m_\epsilon}
   +\frac32 C_{\rm p}e_{\rm node}\le\epsilon.                  \label{eq:final-C1-coefficient-error}
\end{align}
For Hamiltonian coefficients the factor $3/2$ can be replaced by one.
Clipping preserves the errors because the true Hamiltonian
coefficients lie in $[-1,1]$ and the true rates lie in $[0,\tau]$. This
proves Theorem~\ref{thm:main}.

The total evolution time is at most $TS$. For the polynomial specialization,
if $N_{\max}$ is the largest size among the primary and auxiliary
deterministic Chebyshev grids, then
\begin{align}
 t_{\min}=T\sin^2\!\left(\frac{\pi}{4N_{\max}}\right)
 =\Theta(T/N_{\max}^2),
\end{align}
For a general MSFS, the protocol instead inherits the effective
deterministic node lower bound in Def.~\ref{def:Markov_stable}.

\section{Robustness}\label{sec:robust}

The theorem already includes statistical, interpolation, and certified
schedule-approximation errors.  This
section records two additional, explicitly restricted robustness statements.
They do not extend the learned generator class beyond Hamiltonians and weak
onsite Pauli dephasing.

\subsection{Small generator perturbations}

Let $\mathcal S_0$ and $\mathcal S_1$ both satisfy the hypothesis of our main theorem, including that every local generator term vanishes on the
identity, and suppose corresponding
local terms differ by at most $\eta$.  Then
Proposition~\ref{prop:locality_bounds} gives, for a local probe $O$ supported in
a ball of radius $r_A$,
\begin{align}
 \|T_{\mathcal S_1}(0,t)(O)-T_{\mathcal S_0}(0,t)(O)\|
 \le C_{\rm cmp}\eta\|O\|(1+r_A)^{2D-1}(e^{\mu't}-1).             \label{eq:local-perturbation}
\end{align}
This is an additive model-error floor.  Generic non-unital perturbations are
not covered because the operator-norm contractivity used by the protocol can
fail for their forward Schrödinger channels.

\subsection{SPAM noise}

Let $\mathcal T_P$ be the preparation channel and $\mathcal T_M$ the
measurement channel.  The correctly oriented noisy transfer overlap is
\begin{align}
 F_{\rm spam}(P,Q;t)
 :=2^{-n}\tr\!\left[
 \mathcal T_M^*(Q)\,T(0,t)(\mathcal T_P(P))\right].               \label{eq:correct-spam-orientation}
\end{align}
For a rigorous unknown-noise statement, assume
$\mathcal T_P=\bigotimes_j\mathcal T_{P,j}$ and
$\mathcal T_M=\bigotimes_j\mathcal T_{M,j}$ are onsite product channels,
with every preparation factor unital. Suppose, in the relevant induced
norms,
\begin{align}
 \|\mathcal T_{P,j}-\mathcal I\|_{1\to1}\le\eta_{\rm SPAM},
 \qquad
 \|\mathcal T_{M,j}^*-\mathcal I\|_{\infty\to\infty}
 \le\eta_{\rm SPAM}.                                           \label{eq:spam-local-norms}
\end{align}
Then
\begin{align}
 |F_{\rm spam}(P,Q;t)-F(P,Q;t)|
 \le [w(P)+w(Q)]\eta_{\rm SPAM}.                              \label{eq:spam-additive-floor}
\end{align}
Indeed, unitality leaves identity factors outside $\supp(P)$ unchanged.
Tensor-product telescoping and trace-norm contractivity give
$2^{-n}\|\mathcal T_P(P)-P\|_1\le w(P)\eta_{\rm SPAM}$, while
operator-norm contractivity of the measurement adjoints gives
$\|\mathcal T_M^*(Q)-Q\|_\infty\le w(Q)\eta_{\rm SPAM}$. Splitting
Eq.~\eqref{eq:correct-spam-orientation} into these two errors, using
trace duality and contractivity of $T(0,t)$, proves
Eq.~\eqref{eq:spam-additive-floor}. This is an additive recovery floor, not
an error removed by more samples.

For characterized single-qubit depolarization, define
\begin{align}
 \mathcal D_p^{(1)}(X)=pX+(1-p)\tr(X)I/2,\qquad
 \mathcal D_p=\bigl(\mathcal D_p^{(1)}\bigr)^{\otimes n}.
\end{align}
Then $\mathcal D_p(P)=p^{w(P)}P$.  Preparation and measurement parameters
$p_P,p_M$ attenuate a transfer entry by
$p_P^{w(P)}p_M^{w(Q)}$, producing variance overhead
\begin{align}
 p_P^{-2w(P)}p_M^{-2w(Q)}.
\end{align}

\subsection{Held-out prediction and coefficient validation}
\label{sec:holdout_validation}

Independent holdout experiments can validate predicted channel overlaps, but
they do not directly reveal the instantaneous coefficients.  Fix local test
pairs $(P_a,Q_a)$ and draw $M$ independent uniform holdout times $U_j$.  Let
$Y_{a,j}$ be independent bounded estimators of
$F_a(U_j)=2^{-n}\tr[Q_aT(0,U_j)(P_a)]$, and let
$\widehat F_a$ be the prediction of the learned model.  If
$|Y_{a,j}-F_a(U_j)|\le\sigma$ on the measurement event and
$|F_a(u)-\widehat F_a(u)|\le B$, then Hoeffding's inequality gives, with
probability at least $1-\delta$,
\begin{align}
 \frac1T\int_0^T|F_a(u)-\widehat F_a(u)|\,du
 \le \frac1M\sum_{j=1}^M|Y_{a,j}-\widehat F_a(U_j)|
 +\sigma+B\sqrt{\frac{\log(2/\delta)}{2M}}.                      \label{eq:honest-holdout-certificate}
\end{align}
A union bound validates a finite family of local predictions.  This
certificate is conditional on the observed empirical residual and makes no
claim about coefficient error.

\paragraph{Optional independent coefficient certificate.}
A coefficient certificate requires independent estimates of the
instantaneous generator, rather than only channel overlaps.  Let
$\Xi_{m_\epsilon}=\{t_i\}$ be the primary stable interpolation nodes and let
$\widehat\theta^{\rm raw}_c\in\mathcal F_{m_\epsilon}$ denote a learned
physical Hamiltonian coefficient or rate before the final clipping step,
with $c$ indexing these physical coefficients.  Reserve an independent
validation batch and rerun the local-probe reconstruction, transfer-matrix,
derivative-estimation, and stable-solve steps at the same nodes.  Convert
working-basis rate estimates to physical rates before defining the validation
error.  Suppose that, with probability at least $1-\delta_{\rm val}$, the
resulting physical estimates obey
\begin{align}
 \max_{c,i}
 |\widehat\theta^{\rm val}_{c,i}-\theta_c(t_i)|
 \le e_{\rm val}.                                               \label{eq:validation-node-error}
\end{align}
Define the observed validation residual
\begin{align}
 R_{\rm val}:=\max_{c,i}
 |\widehat\theta^{\rm val}_{c,i}
       -\widehat\theta^{\rm raw}_c(t_i)|.                       \label{eq:coefficient-validation-residual}
\end{align}
For the witness $p_{\theta_c,m_\epsilon}\in\mathcal F_{m_\epsilon}$
in Eq.~\eqref{eq:C1-MSFS-modulus}, the norming inequality
Eq.~\eqref{eq:stable-node-norming} gives
\begin{align}
 \|\theta_c-\widehat\theta^{\rm raw}_c\|_{\infty,[0,T]}
 \le C_{\rm p}(R_{\rm val}+e_{\rm val})
 +(1+C_{\rm p})\rho_{m_\epsilon}.                              \label{eq:coefficient-validation-certificate}
\end{align}
Indeed, the difference
$p_{\theta_c,m_\epsilon}-\widehat\theta^{\rm raw}_c$ lies in
$\mathcal F_{m_\epsilon}$ and is bounded by
$R_{\rm val}+e_{\rm val}+\rho_{m_\epsilon}$ on every primary node.
Clipping cannot increase the error because the true coefficients lie in their
promised intervals.  For exact schedules in $\mathcal F_m$, the
$\rho_m$ term vanishes.  If a common working-basis node error is
$\widetilde e_{\rm val}$, one may take the common physical error in
Eq.~\eqref{eq:validation-node-error} to be
$e_{\rm val}=3\widetilde e_{\rm val}/2$, using
$\|\Gamma_j\|_{\ell_\infty\to\ell_\infty}=3/2$.

For a target coefficient certificate $\epsilon_{\rm cert}$, the validation
passes only if the observed right-hand side obeys
\begin{align}
 C_{\rm p}(R_{\rm val}+e_{\rm val})
 +(1+C_{\rm p})\rho_{m_\epsilon}
 \le\epsilon_{\rm cert}.                                      \label{eq:coefficient-validation-acceptance}
\end{align}

This independent certificate costs essentially one additional learning run:
at validation precision $e_{\rm val}$ and failure probability
$\delta_{\rm val}$, its sample complexity is
\begin{align}
 S_{\rm val}=\widetilde{\mathcal O}\!\left(
 e_{\rm val}^{-2}\poly{m_{{\rm eff},{\rm val}}}
 \log(n\delta_{\rm val}^{-1})\right).                          \label{eq:coefficient-validation-cost}
\end{align}
Here the main coefficient degree and primary nodes remain fixed, while the
auxiliary validation and localization budgets, together with the fixed degree
$m_\epsilon$, are summarized by $m_{{\rm eff},{\rm val}}$ and chosen to
achieve $e_{\rm val}$.  After
allocating room for the model bias and observed residual in
Eq.~\eqref{eq:coefficient-validation-acceptance}, it suffices to take
$e_{\rm val}=\Theta(\epsilon_{\rm cert}/C_{\rm p})$.  Since
$C_{\rm p}=\poly{m_\epsilon}$, this has the same
$\widetilde{\mathcal O}(\epsilon_{\rm cert}^{-2}
\poly{m_{{\rm eff},{\rm val}}}\log(n\delta_{\rm val}^{-1}))$ form as
Theorem~\ref{thm:main}.  The certificate is conditional on the MSFS
approximation promise; it does not by itself establish that promise
nonparametrically.

\section{Estimating errors in quantum simulation}

\subsection{Physical expectation values}\label{sec:verification}

Beyond benchmarking quantum simulators and estimating time-dependent noise, which is a well-motivated application of our method, another interesting direction is to follow~\cite{joshi2025bounded} and use Hamiltonian and Lindbladian learning to certify the performance of quantum simulators. That is, given some observable
of interest, say the energy of a system after an adiabatic evolution, we can resort to our methods to certify that the energy we measure on the simulator is within a value of the
one we would obtain on an ideal system, say one without noise or with perfectly implemented time-dependent evolution.

\begin{prop}[Verification of local expectation values]
\label{prop:verifying_expectation_values}
Let $T$ and $\widehat T$ be the forward channels generated by two models with $T$ the true model satisfying
Problem~\ref{prob:time_dependent}. Clip the learned Hamiltonian coefficients to
$[-1,1]$ and physical rates to $[0,\tau]$. Suppose every corresponding local coefficient differs
by at most $\epsilon$ uniformly on $[0,T]$.  For a local observable $O$
supported in a ball of radius $r_A$ and any state $\rho$,
\begin{align}
 &\left|\tr[O T(0,t)(\rho)]-\tr[O\widehat T(0,t)(\rho)]\right|       \nonumber\\
 &\qquad\le
 \|T(0,t)^*(O)-\widehat T(0,t)^*(O)\|                             \nonumber\\
 &\qquad\le C_{\rm cmp}\epsilon\|O\|(1+r_A)^{2D-1}
 (e^{\mu't}-1),\qquad 0\le t\le T.                              \label{eq:physical-verification-bound}
\end{align}
\end{prop}
\begin{proof}
The first inequality is trace duality.  Apply the Heisenberg comparison bound
in Proposition~\ref{prop:locality_bounds} directly to $T^*$ and $\widehat T^*$;
bounded local multiplicities and generator-basis norms are absorbed in
$C_{\rm cmp}$. Clipping is
nonexpansive relative to true coefficients in the promised intervals and
keeps the reconstructed channel physical.
\end{proof}

\subsection{Extrapolating polynomial schedules}\label{sec:extrapolation}

This optional subsection is restricted to true degree-$m$ polynomial
schedules (the exact case $\rho_m=0$) whose promised bounds hold on the
larger interval $[0,T_f]$.  A $C^1$ approximation certificate only on the
training interval $[0,T]$ gives no control beyond $T$ and is not sufficient
for extrapolation.  Extrapolation is intrinsically ill-conditioned.  The
exact amplification factor is a rescaled Chebyshev polynomial, rather than
an unspecified big-$\cO$ power.
\begin{lemma}\label{lemma:extrapolation_bound}
Let $q$ be a real polynomial of degree at most $m$ with
$\|q\|_{\infty,[0,T]}\le\epsilon$.  For $T_f\ge T$,
\begin{align}
 \|q\|_{\infty,[0,T_f]}
 \le \epsilon\,\chi_m(T_f/T),\qquad
 \chi_m(\lambda):=
 \left|\mathrm T_m(2\lambda-1)\right|,                           \label{equ:extrapolation}
\end{align}
where $\mathrm T_m$ is the degree-$m$ Chebyshev polynomial.  Moreover, for
every $\lambda\ge1$,
\begin{align}
 \chi_m(\lambda)^2
 &\le\bigl(\sqrt\lambda+\sqrt{\lambda-1}\bigr)^{4m}
 \le(4\lambda)^{2m}.                                           \label{eq:chi-explicit-bound}
\end{align}
\end{lemma}
\begin{proof}
Map $[0,T]$ affinely to $[-1,1]$ and apply the Chebyshev extremal property in
Lemma~\ref{lem:cheby-extremal}.  The maximum on $[0,T_f]$ occurs at the
right endpoint of the rescaled interval.  For $x\ge1$,
$\mathrm T_m(x)=\cosh(m\operatorname{arcosh}x)$.  Writing
$u=\operatorname{arcosh}\sqrt\lambda$ gives
$\operatorname{arcosh}(2\lambda-1)=2u$ and
\begin{align}
 \chi_m(\lambda)^2=\cosh^2(2mu)
 \le e^{4mu}
 =\bigl(\sqrt\lambda+\sqrt{\lambda-1}\bigr)^{4m}.
\end{align}
The last inequality in Eq.~\eqref{eq:chi-explicit-bound} follows from
$\sqrt\lambda+\sqrt{\lambda-1}\le2\sqrt\lambda$.
\end{proof}

To use this factor in a dynamical bound, assume additionally that the true
and learned local interaction strengths remain bounded on $[0,T_f]$ by stated
constants, so that the locality constants do not change uncontrolledly.
The raw polynomial estimates may be clipped to the promised Hamiltonian and
rate intervals, preserving physical nonnegative dissipative rates.

\begin{prop}[Extrapolated verification]
\label{prop:verifying_expectation_values_extrapolation}
Suppose the differences between the true coefficients and the \emph{raw}
polynomial interpolants have degree at most $m$ and are bounded by
$\epsilon$ on $[0,T]$. Before generating $\widehat T$, clip the raw
Hamiltonian coefficients and rates pointwise to their promised intervals on
$[0,T_f]$. Under the preceding long-interval strength assumptions, for
$0\le t\le T_f$,
\begin{align}
 \|T(0,t)^*(O)-\widehat T(0,t)^*(O)\|
 \le\min\Bigl\{&C_f\epsilon\chi_m(T_f/T)\|O\|
 (1+r_A)^{2D-1}(e^{\mu'_f t}-1),\nonumber\\
 &M t\epsilon\chi_m(T_f/T)\|O\|\Bigr\},                         \label{eq:correct-extrapolated-verification}
\end{align}
where $M$ is the number of local terms and $C_f,\mu'_f$ are the promised
constants furnished by Proposition~\ref{prop:locality_bounds} from the stated
long-interval interaction-strength bounds.
\end{prop}
\begin{proof}
Lemma~\ref{lemma:extrapolation_bound} bounds every raw local coefficient error
on $[0,T_f]$. Pointwise projection onto the promised Hamiltonian and rate
intervals is nonexpansive relative to the true coefficients, so the same
bound holds after clipping. The first estimate follows from the physical
Heisenberg comparison bound; the second is the global Duhamel bound.
\end{proof}

Consequently, a final verification target $\eta$ obtained from a learning
routine whose cost is quadratic in the training precision has sample cost
\begin{align}
 \widetilde\cO\!\left(
 \eta^{-2}M^2T_f^2\chi_m(T_f/T)^2\poly{m}\right),
\end{align}
up to the local alternative in
Eq.~\eqref{eq:correct-extrapolated-verification}.  Writing
$\lambda=T_f/T$, Eq.~\eqref{eq:chi-explicit-bound} gives the fully explicit
upper bound
\begin{align}
 \widetilde\cO\!\left(
 \eta^{-2}M^2T_f^2
 \bigl(\sqrt\lambda+\sqrt{\lambda-1}\bigr)^{4m}\poly{m}\right).
                                                                    \label{eq:explicit-extrapolation-sample-cost}
\end{align}
More coarsely, it is bounded by
\begin{align}
 \widetilde\cO\bigl(\eta^{-2}M^2T_f^2
 (4T_f/T)^{2m}\poly{m}\bigr).                                  \label{eq:coarse-extrapolation-sample-cost}
\end{align}
Thus, at fixed $T$, this coarse bound grows
as $T_f^{2m+2}$;
extrapolation is useful only for small fixed degree.

\section{Numerical illustrations}
\label{sec:numerics}

\paragraph{Use of artificial intelligence tools.}
OpenAI's ChatGPT, using GPT-5.5 and GPT-5.6, was used extensively to
generate and debug the numerical and plotting code and to assist in producing
the figures for this work. The authors directed the models to implement the
protocol developed by the authors, reviewed the generated code, ran several
unit tests, and checked the resulting numerical outputs and figures. The
authors take full responsibility for the implementation and the reported
numerical results.

We complement the worst-case guarantees of Theorem~\ref{thm:main} with three
numerical tests of the parts of the protocol that determine its practical
cost, i.e.~whether the constants in the protocol are well-behaved enough to ensure that the number of samples required for moderate precision remains feasible.  The tests address three distinct questions. First, whether the process-shadow pipeline remains usable beyond a few qubits. Then, whether smooth schedules outside
the fitted polynomial class can still be recovered, and then how the sampling cost
grows with the dimension of the time-dependent function class.  All errors
below are evaluated on the physical coefficient functions, rather than on
their monomial coefficients.  In particular, for a true schedule $h_\alpha$
and reconstruction $\widehat h_\alpha$, we report the uniform error
\begin{align}
 E_\alpha:=\sup_{t\in[0,1]}
 \bigl|\widehat h_\alpha(t)-h_\alpha(t)\bigr|.                 \label{eq:numerical-uniform-error}
\end{align}

\paragraph{A nontrivial process-shadow instance.}
We first consider a time-dependent transverse-field Ising-type chain of
$n=15$ qubits.  The generator basis contains onsite $X$ and $Z$ terms and
nearest-neighbor $ZZ$ terms, for a total of $44$ local coefficient schedules.
Process-shadow data are collected at $20$ uniformly spaced times in $[0,1]$,
and the schedules are reconstructed with a centered local derivative fit and
a degree-one time model.  Because the difference of two linear schedules is
linear, Eq.~\eqref{eq:numerical-uniform-error} is obtained exactly by checking
the two endpoints of the interval. We used tensor network methods~\cite{Ors2014} to generate process shadows synthetically. Note that we did not add dissipation due to the large computational overheads required to simulate the time evolution with dissipation.

Figure~\ref{fig:numerics-tfim-n15} shows both the median error over the $44$
schedules and the worst schedule.  The fitted slopes are $-0.52$ and $-0.51$,
respectively, in close agreement with the $S^{-1/2}$ statistical scaling.
At the largest data set, containing $527{,}650$ process-shadow samples at each
of the $20$ queried times, the median uniform error is $1.47\times10^{-2}$ and
the worst-schedule error is $6.64\times10^{-2}$.  This test therefore retains
percent-level accuracy on a system large enough that dense process tomography
is unavailable, while using only local Pauli preparations and measurements.

\begin{figure}[tbp]
  \centering
  \includegraphics[width=0.90\linewidth]{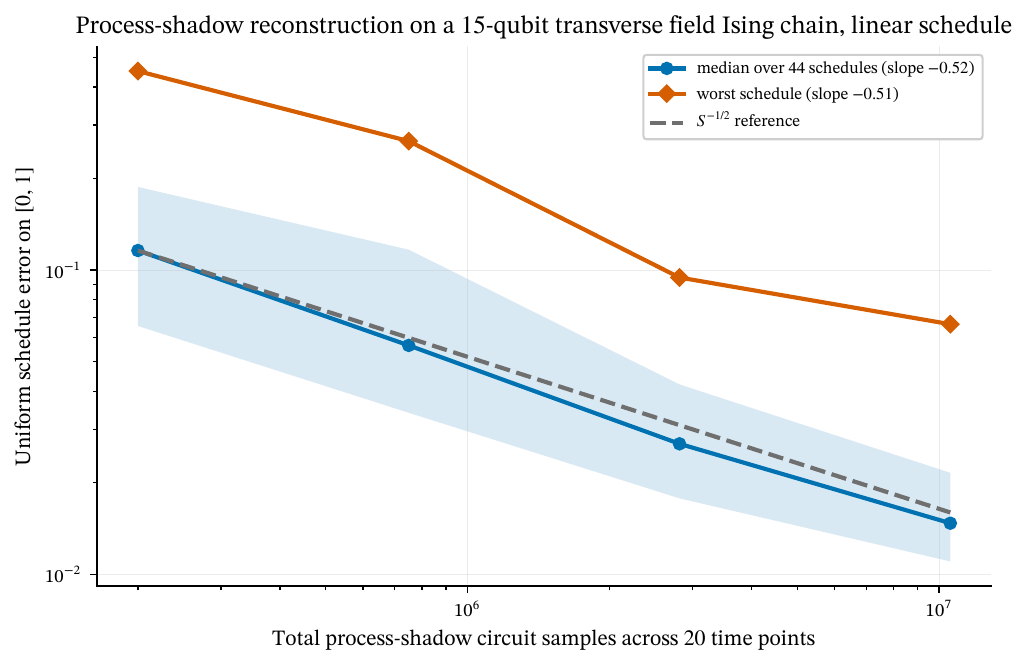}
  \caption{Uniform recovery error for a $15$-qubit time-dependent
  transverse-field Ising-type chain with $44$ local schedules.  The horizontal
  axis is the total number of recorded process-shadow samples over the $20$
  queried times.  The blue curve is the median of $E_\alpha$ over the $44$
  schedules, the blue band is their interquartile range, and the orange curve
  is their maximum. Both aggregate
  errors follow the expected inverse-square-root sampling law. The plotting
  code for this figure was prepared with substantial assistance from OpenAI's
  ChatGPT (GPT-5.5 and GPT-5.6); the authors checked the code and the displayed
  results.}
  \label{fig:numerics-tfim-n15}
\end{figure}

\paragraph{Smooth schedules outside the fitted model class.}
The theorem allows approximate time models provided that their deterministic
$C^1$ approximation bias is certified.  To probe this regime directly, we
simulate a two-qubit generator whose Hamiltonian couplings and onsite
dephasing rates contain functions of the form $\sin(kt+\phi)$, with
$k\in\{4,8,16,24\}$.  These schedules are not polynomials on $[0,1]$.  We fit
polynomials of degree $d=0,\ldots,20$ and compare the recovered schedules with
the best polynomial approximants evaluated on the same dense time grid.

As shown in Fig.~\ref{fig:numerics-nonpolynomial}, the learned error initially
tracks the deterministic approximation error closely.  An error of order
$10^{-2}$ is reached around degrees $4$, $7$, $14$, and $20$ for
$k=4$, $8$, $16$, and $24$, respectively.  Thus the working degree grows with
the frequency content, as expected from the approximation results in
Appendix~\ref{app:MSFS}, but exact membership in a polynomial class is not
required.  The synthetic shot count in this diagnostic is deliberately very
large ($10^{12}$ per scalar expectation) so that approximation bias and
high-degree numerical conditioning, rather than sampling noise, are visible.
For the lower frequencies, increasing $d$ past the optimal working range no
longer improves the recovered curve even though the best polynomial error
continues to decrease.  This identifies that at large enough degrees, the amplification of the uncertainty caused by the derivative estimation step dominates the improvement in precision given by the higher degree.

\begin{figure}[tbp]
  \centering
  \includegraphics[width=0.90\linewidth]{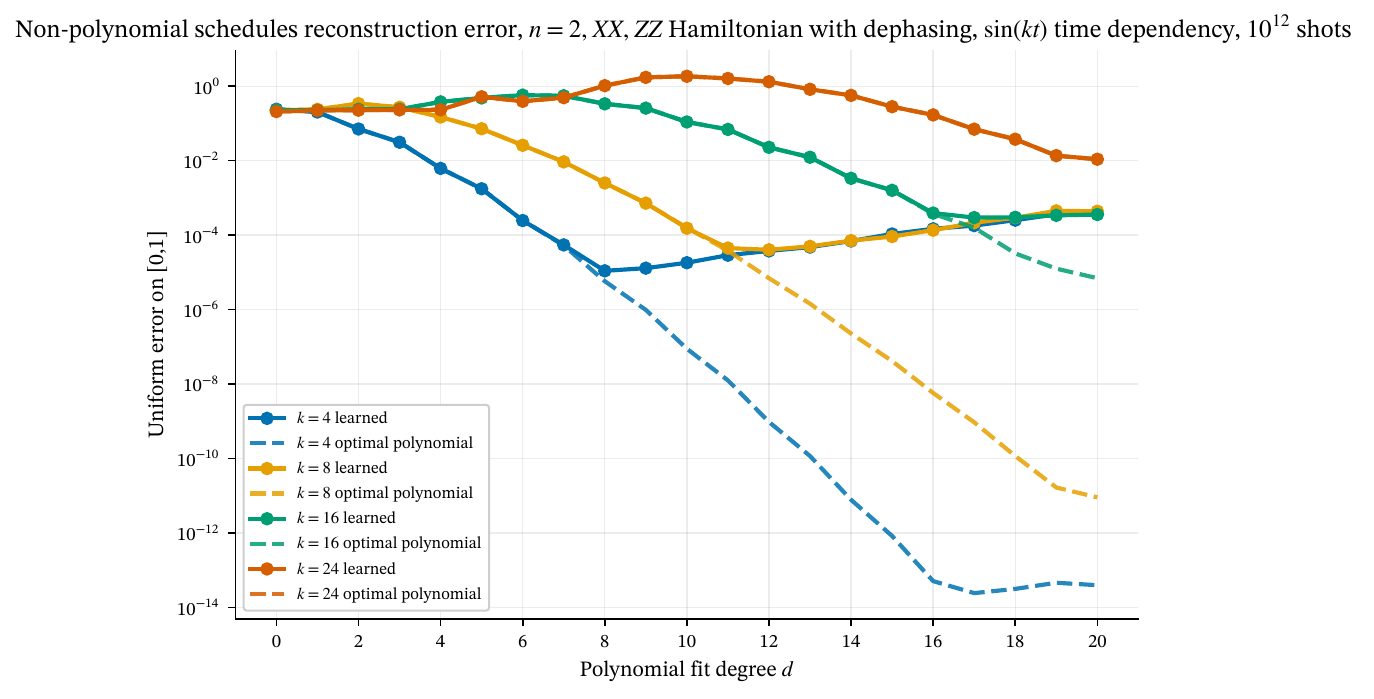}
  \caption{Recovery of non-polynomial time dependence.  Solid curves show the
  maximum, over four Hamiltonian and dephasing schedules, of the median uniform
  recovery error across $31$ noise realizations.  Dashed curves show the
  corresponding best-polynomial approximation error on the dense evaluation
  grid.  Higher temporal frequency requires a larger polynomial degree, while
  the learned schedules follow the approximation-limited behavior up to the
  numerical noise floor.  The intentionally large $10^{12}$ synthetic shot count
  isolates approximation and conditioning effects; it is not a proposed
  experimental sampling budget. The plotting code for this figure was prepared
  with substantial assistance from OpenAI's ChatGPT (GPT-5.5 and GPT-5.6); the
  authors checked the code and the displayed results.}
  \label{fig:numerics-nonpolynomial}
\end{figure}

\paragraph{Dependence on the function-space dimension.}
Finally, we isolate the degree-dependent sampling overhead.  We use an exact
dense simulation of a three-qubit Hamiltonian with polynomial coefficient
schedules of degree $d=0,\ldots,6$, write $m=d+1$ for the corresponding
function-space dimension, and add independent Gaussian shot noise to the
exact scalar Pauli expectations.  For every pair $(d,S)$ we use $51$ noise
realizations and evaluate the recovered schedules on a grid of $2001$ times.
The left panel of Fig.~\ref{fig:numerics-degree-sampling} reports the median,
over realizations, of the worst $E_\alpha$ over all learned parameters.  Across
all degrees, the fitted sampling exponents lie between $-0.59$ and $-0.48$.
Thus increasing the model degree changes the prefactor but does not spoil the
inverse-square-root dependence on the shots per scalar expectation.

\clearpage
\begin{figure}[t]
  \centering
  \includegraphics[width=\linewidth]{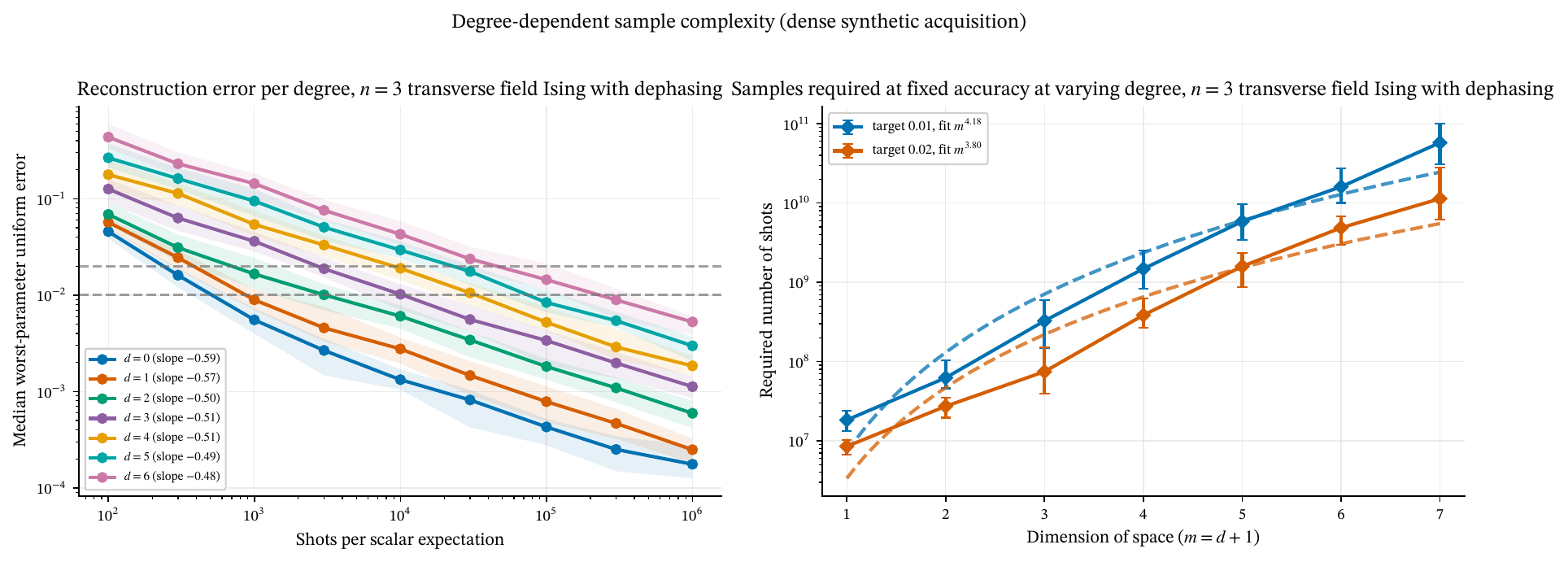}
  \caption{Degree-dependent sample complexity for exact three-qubit dynamics
  with synthetic scalar-expectation noise.  \emph{Left:} median
  worst-parameter uniform error, with interquartile bands over $51$ noise
  realizations.  Dashed horizontal lines mark errors $10^{-2}$ and
  $2\times10^{-2}$.  \emph{Right:} total scalar measurements required to reach
  those targets; error bars are obtained by applying the same inversion to the
  realization quartiles.  Dashed curves are power-law fits over $m=1,\ldots,7$.
  The plotting code for this figure was prepared with substantial assistance
  from OpenAI's ChatGPT (GPT-5.5 and GPT-5.6); the authors checked the code and
  the displayed results.}
  \label{fig:numerics-degree-sampling}
\end{figure}

The right panel inverts these curves to obtain the number of measurements
needed for target errors $10^{-2}$ and $2\times10^{-2}$.  Multiplying the
required shots per scalar expectation by the complete scalar-expectation count
gives empirical dependences $m^{4.18}$ and $m^{3.80}$, respectively.  Every
target crossing lies inside the simulated shot range and is obtained by
log--log interpolation between adjacent shot counts.  These exponents are
finite-range empirical fits, not claims about a universal asymptotic exponent, but do suggest that the scaling for our scheme is approximately quartic in $m$ over the tested range.
Moreover, this dense acquisition deliberately counts each scalar expectation
separately and does not exploit the parallel estimation furnished by process
shadows.  It should therefore be read as a conservative stress test of the
degree overhead; Fig.~\ref{fig:numerics-tfim-n15} is the direct process-shadow
sampling benchmark.

Taken together, these experiments support the three practical features of the
analysis: local process shadows reach a nontrivial qubit number with modest
uniform error, smooth model misspecification can be absorbed by increasing the
working degree, and the measured overhead in that degree is polynomial over
the tested range.

\section{Outlook and open problems}\label{sec:outlook}

This work gives an efficient protocol for learning geometrically local
time-dependent Hamiltonians together with weak onsite Pauli dephasing, under
a certified simultaneous $C^1$ MSFS approximation model and the locality
estimates proved in Appendix~\ref{app:LR}. Its experimental access model uses product Pauli
preparations and measurements. The restricted result leaves several natural
open problems.

\paragraph{Structure learning:} It would be interesting to remove the assumption that the generator support is known in advance, as has been achieved for some time-independent models~\cite{StilckFranca.2024,tang_structure}. Our protocol uses the candidate terms $\cP_\alpha,\cL_{j,P}$ when constructing stable rows. Searching over local neighborhoods is a possible starting point, but we leave a proof and its computational analysis to future work.

\paragraph{Algebraic tails and general dissipation:} our protocol relies on
$s$-stable Pauli probes and operator-norm localization. Extending it to
algebraically decaying interactions requires summable tail bounds. General
non-unital onsite dissipation is also open because a forward Schrödinger
channel need not be operator-norm contractive. Multiqubit dissipative terms,
including unital ones, require a new sparse stable-probe construction. We make
no claim for either class here.

\paragraph{Heisenberg scaling and efficient verification of long-term behavior:} mirroring the time-independent case~\cite{tang_structure,PhysRevLett.130.200403}, it would be interesting to achieve Heisenberg scaling for our problem in precision, i.e.~a total evolution time that scales like $\widetilde{\mathcal{O}}(\epsilon^{-1})$ instead of $\widetilde{\mathcal{O}}(\epsilon^{-2})$. In addition, it would be valuable to probe the evolution at arbitrarily long times efficiently regardless of the approximation degree. For both approaches, the ``isolation through dissipation'' methods introduced in~\cite{Moebus.2025} in the bosonic setting may offer a useful starting point. There, single-body dissipation is used to decouple different parts of the system. This could help control light-cone growth and facilitate the construction of inverse probes. We leave this investigation to future work.

\paragraph{Other stable families and approximation certificates:} the theorem
allows approximate schedules once a uniform $C^1$ approximation modulus is
certified. It would be useful to derive sharp, readily checkable moduli and
stable interpolation constants for trigonometric-polynomial, spline, and
mixed families.
In addition, while our methods allow for the efficient approximate evaluation of the functions, it would be interesting to explore other techniques to then estimate the underlying functions, such as compressed sensing methods~\cite{ma2024learningkbodyhamiltonianscompressed,Foucart2013}. These often allow for more efficient reconstruction of functions from approximate evaluation assuming that they are e.g.~Fourier sparse, which is likely a natural assumption for various quantum systems.

\paragraph{Experimental demonstration and large scale verification:} the experimental requirements for our protocol are minimal and supported by many platforms. It would be interesting to run our protocol to verify e.g.~the adiabatic preparation of a quantum state on multiple qubits.

\section{Acknowledgments}
D.S.F. acknowledges financial support from the Novo Nordisk Foundation (Grant No. NNF20OC0059939 Quantum for Life) and the ERC grant GIFNEQ 101163938.
T.M. acknowledges support from the Deutsche Forschungsgemeinschaft (DFG, German Research Foundation), Project-ID 470903074, TRR 352. This project was funded within the QuantERA II Programme, which received funding from the EU's H2020 research and innovation programme under Grant Agreement No.~101017733.
D.S.F. and C.R. are
supported by France 2030 under the French National Research Agency award number ``ANR-22-
PNCQ-0002.''
A.H.W. thanks the VILLUM FONDEN for its support with a Villum Young Investigator (Plus) Grant (Grant No. 25452 and Grant No. 60842) as well as via the QMATH Centre of Excellence (Grant No. 10059).

OpenAI's ChatGPT, using GPT-5.5 and GPT-5.6, was also used to check the
manuscript for typographical and mathematical errors. The protocol and proofs
were developed by the authors, who reviewed the AI-assisted suggestions and
take full responsibility for the manuscript.

\section*{Conflict of interest}
The authors declare that they have no conflicts of interest.

\section*{Data and code availability}
The source code, parameters, and scripts used to produce the numerical results
and numerical figures in this article, together with all processed numerical
data underlying those figures, are openly available in the versioned Zenodo
archive in Ref.~\cite{MobusFrancaNumerics2026}. The corresponding development
repository is available on
\href{https://github.com/moebustim/numerics-learning-time-dependent-quantum-dynamics}{GitHub}.
The archive documents the data lineage and reproducibility scope and gives
instructions for seeded reruns. Historical raw simulation trajectories and
intermediate solver outputs are not archived separately; new seeded runs can
be generated using the deposited code and parameters.

\bibliographystyle{alphaurl}
\bibliography{bibliography}

\newcommand{\etalchar}[1]{$^{#1}$}
\begin{thebibliography}{CLMPG15}

\bibitem[AL18]{RevModPhys.90.015002}
Tameem Albash and Daniel~A. Lidar.
\newblock Adiabatic quantum computation.
\newblock {\em Rev. Mod. Phys.}, 90:015002, Jan 2018.
\newblock URL: \url{https://link.aps.org/doi/10.1103/RevModPhys.90.015002}, \href {https://doi.org/10.1103/RevModPhys.90.015002} {\path{doi:10.1103/RevModPhys.90.015002}}.

\bibitem[BAL19]{PhysRevLett.122.020504}
Eyal Bairey, Itai Arad, and Netanel~H. Lindner.
\newblock Learning a local {H}amiltonian from local measurements.
\newblock {\em Phys. Rev. Lett.}, 122:020504, Jan 2019.
\newblock URL: \url{https://link.aps.org/doi/10.1103/PhysRevLett.122.020504}, \href {https://doi.org/10.1103/PhysRevLett.122.020504} {\path{doi:10.1103/PhysRevLett.122.020504}}.

\bibitem[BBS{\etalchar{+}}19]{Burnett2019}
Jonathan~J. Burnett, Andreas Bengtsson, Marco Scigliuzzo, David Niepce, Marina Kudra, Per Delsing, and Jonas Bylander.
\newblock Decoherence benchmarking of superconducting qubits.
\newblock {\em npj Quantum Information}, 5(1), June 2019.
\newblock URL: \url{http://dx.doi.org/10.1038/s41534-019-0168-5}, \href {https://doi.org/10.1038/s41534-019-0168-5} {\path{doi:10.1038/s41534-019-0168-5}}.

\bibitem[BE95]{Borwein1995}
Peter Borwein and Tamás Erdélyi.
\newblock {\em Polynomials and Polynomial Inequalities}.
\newblock Springer New York, 1995.
\newblock URL: \url{http://dx.doi.org/10.1007/978-1-4612-0793-1}, \href {https://doi.org/10.1007/978-1-4612-0793-1} {\path{doi:10.1007/978-1-4612-0793-1}}.

\bibitem[BK12]{BarthelKliesch2012}
Thomas Barthel and Martin Kliesch.
\newblock Quasilocality and efficient simulation of markovian quantum dynamics.
\newblock {\em Physical Review Letters}, 108(23):230504, 2012.
\newblock \href {https://arxiv.org/abs/1111.4210} {\path{arXiv:1111.4210}}, \href {https://doi.org/10.1103/PhysRevLett.108.230504} {\path{doi:10.1103/PhysRevLett.108.230504}}.

\bibitem[BLMT24]{tang_structure}
Ainesh Bakshi, Allen Liu, Ankur Moitra, and Ewin Tang.
\newblock Structure learning of {H}amiltonians from real-time evolution.
\newblock In {\em 2024 IEEE 65th Annual Symposium on Foundations of Computer Science (FOCS)}, pages 1037--1050, 2024.
\newblock \href {https://doi.org/10.1109/FOCS61266.2024.00069} {\path{doi:10.1109/FOCS61266.2024.00069}}.

\bibitem[BP07]{Breuer2007}
Heinz-Peter Breuer and Francesco Petruccione.
\newblock {\em The Theory of Open Quantum Systems}.
\newblock Oxford University PressOxford, January 2007.
\newblock URL: \url{http://dx.doi.org/10.1093/acprof:oso/9780199213900.001.0001}, \href {https://doi.org/10.1093/acprof:oso/9780199213900.001.0001} {\path{doi:10.1093/acprof:oso/9780199213900.001.0001}}.

\bibitem[Car24]{10.1145/3670418}
Matthias~C. Caro.
\newblock Learning quantum processes and {H}amiltonians via the pauli transfer matrix.
\newblock {\em ACM Transactions on Quantum Computing}, 5(2), June 2024.
\newblock \href {https://doi.org/10.1145/3670418} {\path{doi:10.1145/3670418}}.

\bibitem[CL21]{chen2021operator}
Chi-Fang Chen and Andrew Lucas.
\newblock Operator growth bounds from graph theory.
\newblock {\em Communications in Mathematical Physics}, 385(3):1273--1323, 2021.

\bibitem[CLMPG15]{Cubitt2015}
Toby~S. Cubitt, Angelo Lucia, Spyridon Michalakis, and David Perez-Garcia.
\newblock Stability of local quantum dissipative systems.
\newblock {\em Communications in Mathematical Physics}, 337(3):1275–1315, April 2015.
\newblock URL: \url{http://dx.doi.org/10.1007/s00220-015-2355-3}, \href {https://doi.org/10.1007/s00220-015-2355-3} {\path{doi:10.1007/s00220-015-2355-3}}.

\bibitem[CLY23]{chen2023speed}
Chi-Fang~Anthony Chen, Andrew Lucas, and Chao Yin.
\newblock Speed limits and locality in many-body quantum dynamics.
\newblock {\em Reports on Progress in Physics}, 86(11):116001, 2023.

\bibitem[DOS24]{Dutkiewicz2024}
Alicja Dutkiewicz, Thomas~E. O\'Brien, and Thomas Schuster.
\newblock The advantage of quantum control in many-body {H}amiltonian learning.
\newblock {\em Quantum}, 8:1537, November 2024.
\newblock URL: \url{http://dx.doi.org/10.22331/q-2024-11-26-1537}, \href {https://doi.org/10.22331/q-2024-11-26-1537} {\path{doi:10.22331/q-2024-11-26-1537}}.

\bibitem[dSLCP11]{PhysRevLett.107.210404}
Marcus~P. da~Silva, Olivier Landon-Cardinal, and David Poulin.
\newblock Practical characterization of quantum devices without tomography.
\newblock {\em Phys. Rev. Lett.}, 107:210404, Nov 2011.
\newblock URL: \url{https://link.aps.org/doi/10.1103/PhysRevLett.107.210404}, \href {https://doi.org/10.1103/PhysRevLett.107.210404} {\path{doi:10.1103/PhysRevLett.107.210404}}.

\bibitem[Dzi10]{Dziarmaga2010}
Jacek Dziarmaga.
\newblock Dynamics of a quantum phase transition and relaxation to a steady state.
\newblock {\em Advances in Physics}, 59(6):1063–1189, September 2010.
\newblock URL: \url{http://dx.doi.org/10.1080/00018732.2010.514702}, \href {https://doi.org/10.1080/00018732.2010.514702} {\path{doi:10.1080/00018732.2010.514702}}.

\bibitem[FFPPY25]{Foss_Feig_2025}
Michael Foss-Feig, Guido Pagano, Andrew~C. Potter, and Norman~Y. Yao.
\newblock Progress in trapped-ion quantum simulation.
\newblock {\em Annual Review of Condensed Matter Physics}, 16(1):145–172, March 2025.
\newblock URL: \url{http://dx.doi.org/10.1146/annurev-conmatphys-032822-045619}, \href {https://doi.org/10.1146/annurev-conmatphys-032822-045619} {\path{doi:10.1146/annurev-conmatphys-032822-045619}}.

\bibitem[FGW{\etalchar{+}}22]{Flynn2022}
Brian Flynn, Antonio~A Gentile, Nathan Wiebe, Raffaele Santagati, and Anthony Laing.
\newblock Quantum model learning agent: characterisation of quantum systems through machine learning.
\newblock {\em New Journal of Physics}, 24(5):053034, May 2022.
\newblock URL: \url{http://dx.doi.org/10.1088/1367-2630/ac68ff}, \href {https://doi.org/10.1088/1367-2630/ac68ff} {\path{doi:10.1088/1367-2630/ac68ff}}.

\bibitem[FR13]{Foucart2013}
Simon Foucart and Holger Rauhut.
\newblock {\em A Mathematical Introduction to Compressive Sensing}.
\newblock Springer New York, 2013.
\newblock URL: \url{http://dx.doi.org/10.1007/978-0-8176-4948-7}, \href {https://doi.org/10.1007/978-0-8176-4948-7} {\path{doi:10.1007/978-0-8176-4948-7}}.

\bibitem[GCC24]{Gu2024}
Andi Gu, Lukasz Cincio, and Patrick~J. Coles.
\newblock Practical {H}amiltonian learning with unitary dynamics and gibbs states.
\newblock {\em Nature Communications}, 15(1), January 2024.
\newblock URL: \url{http://dx.doi.org/10.1038/s41467-023-44008-1}, \href {https://doi.org/10.1038/s41467-023-44008-1} {\path{doi:10.1038/s41467-023-44008-1}}.

\bibitem[GCM23]{Guillaud.2023}
Jérémie Guillaud, Joachim Cohen, and Mazyar Mirrahimi.
\newblock Quantum computation with cat qubits.
\newblock {\em SciPost Physics Lecture Notes}, 2023.
\newblock \href {https://doi.org/10.21468/scipostphyslectnotes.72} {\path{doi:10.21468/scipostphyslectnotes.72}}.

\bibitem[GD14]{Goldman2014}
N.~Goldman and J.~Dalibard.
\newblock Periodically driven quantum systems: Effective {H}amiltonians and engineered gauge fields.
\newblock {\em Physical Review X}, 4(3), August 2014.
\newblock URL: \url{http://dx.doi.org/10.1103/PhysRevX.4.031027}, \href {https://doi.org/10.1103/physrevx.4.031027} {\path{doi:10.1103/physrevx.4.031027}}.

\bibitem[HBCP15]{Holzpfel2015}
M.~Holz\"{a}pfel, T.~Baumgratz, M.~Cramer, and M.~B. Plenio.
\newblock Scalable reconstruction of unitary processes and {H}amiltonians.
\newblock {\em Physical Review A}, 91(4), April 2015.
\newblock URL: \url{http://dx.doi.org/10.1103/PhysRevA.91.042129}, \href {https://doi.org/10.1103/physreva.91.042129} {\path{doi:10.1103/physreva.91.042129}}.

\bibitem[HRO{\etalchar{+}}22]{PRXQuantum.3.020357}
J.~Helsen, I.~Roth, E.~Onorati, A.H. Werner, and J.~Eisert.
\newblock General framework for randomized benchmarking.
\newblock {\em PRX Quantum}, 3:020357, Jun 2022.
\newblock URL: \url{https://link.aps.org/doi/10.1103/PRXQuantum.3.020357}, \href {https://doi.org/10.1103/PRXQuantum.3.020357} {\path{doi:10.1103/PRXQuantum.3.020357}}.

\bibitem[HTFS23]{PhysRevLett.130.200403}
Hsin-Yuan Huang, Yu~Tong, Di~Fang, and Yuan Su.
\newblock Learning many-body {H}amiltonians with heisenberg-limited scaling.
\newblock {\em Phys. Rev. Lett.}, 130:200403, May 2023.
\newblock URL: \url{https://link.aps.org/doi/10.1103/PhysRevLett.130.200403}, \href {https://doi.org/10.1103/PhysRevLett.130.200403} {\path{doi:10.1103/PhysRevLett.130.200403}}.

\bibitem[JKL{\etalchar{+}}25]{joshi2025bounded}
Manoj~K. Joshi, Tristan Kraft, William Lam, Tobias Olsacher, Johannes Franke, Florian Kranzl, Lata~Kh Joshi, Daniel~Stilck Fran{\c{c}}a, Rainer Blatt, Beno{\^i}t Vermersch, Barbara Kraus, Peter Zoller, and Christian~F. Roos.
\newblock Bounded-error quantum simulation via hamiltonian and lindbladian learning.
\newblock In preparation, 2025.

\bibitem[KKP17]{Kane2017-dn}
Daniel Kane, Sushrut Karmalkar, and Eric Price.
\newblock Robust polynomial regression up to the information theoretic limit.
\newblock In {\em 2017 {IEEE} 58th Annual Symposium on Foundations of Computer Science ({FOCS})}. IEEE, October 2017.

\bibitem[KKY{\etalchar{+}}19]{Krantz_2019}
P.~Krantz, M.~Kjaergaard, F.~Yan, T.~P. Orlando, S.~Gustavsson, and W.~D. Oliver.
\newblock A quantum engineer’s guide to superconducting qubits.
\newblock {\em Applied Physics Reviews}, 6(2), June 2019.
\newblock URL: \url{http://dx.doi.org/10.1063/1.5089550}, \href {https://doi.org/10.1063/1.5089550} {\path{doi:10.1063/1.5089550}}.

\bibitem[LR72]{lieb1972finite}
Elliott~H Lieb and Derek~W Robinson.
\newblock The finite group velocity of quantum spin systems.
\newblock {\em Communications in mathematical physics}, 28(3):251--257, 1972.

\bibitem[LTN{\etalchar{+}}23]{li2023heisenberglimitedhamiltonianlearninginteracting}
Haoya Li, Yu~Tong, Hongkang Ni, Tuvia Gefen, and Lexing Ying.
\newblock Heisenberg-limited {H}amiltonian learning for interacting bosons, 2023.
\newblock URL: \url{https://arxiv.org/abs/2307.04690}, \href {https://arxiv.org/abs/2307.04690} {\path{arXiv:2307.04690}}.

\bibitem[MBC{\etalchar{+}}23]{Moebus.2023}
Tim M\"{o}bus, Andreas Bluhm, Matthias~C. Caro, Albert~H. Werner, and Cambyse Rouzé.
\newblock Dissipation-enabled bosonic {H}amiltonian learning via new information-propagation bounds, 2023.
\newblock \href {https://doi.org/10.48550/ARXIV.2307.15026} {\path{doi:10.48550/ARXIV.2307.15026}}.

\bibitem[MBG{\etalchar{+}}25]{Moebus.2025}
Tim M\"{o}bus, Andreas Bluhm, Tuvia Gefen, Yu~Tong, Albert~H. Werner, and Cambyse Rouzé.
\newblock Heisenberg-limited {H}amiltonian learning continuous variable systems via engineered dissipation, 2025.
\newblock \href {https://doi.org/10.48550/ARXIV.2506.00606} {\path{doi:10.48550/ARXIV.2506.00606}}.

\bibitem[MBSFR26]{MobusBergamaschi2026}
Tim M\"{o}bus, Thiago Bergamaschi, Daniel Stilck~Fran\c{c}a, and Cambyse Rouz\'{e}.
\newblock Robust structure learning of {$k$}-local lindbladians, 2026.
\newblock URL: \url{https://arxiv.org/abs/2606.23652}, \href {https://arxiv.org/abs/2606.23652} {\path{arXiv:2606.23652}}, \href {https://doi.org/10.48550/arXiv.2606.23652} {\path{doi:10.48550/arXiv.2606.23652}}.

\bibitem[MFPT24]{ma2024learningkbodyhamiltonianscompressed}
Muzhou Ma, Steven~T. Flammia, John Preskill, and Yu~Tong.
\newblock Learning $k$-body {H}amiltonians via compressed sensing, 2024.
\newblock URL: \url{https://arxiv.org/abs/2410.18928}, \href {https://arxiv.org/abs/2410.18928} {\path{arXiv:2410.18928}}.

\bibitem[MG16]{Markoff1916}
Wladimir Markoff and J.~Grossmann.
\newblock Uber polynome, die in einem gegebenen intervalle moglichst wenig von null abweichen.
\newblock {\em Mathematische Annalen}, 77(2):213–258, June 1916.
\newblock URL: \url{http://dx.doi.org/10.1007/BF01456902}, \href {https://doi.org/10.1007/bf01456902} {\path{doi:10.1007/bf01456902}}.

\bibitem[MSF26]{MobusFrancaNumerics2026}
Tim M\"{o}bus and Daniel Stilck~Fran\c{c}a.
\newblock Numerics for learning time-dependent {H}amiltonians.
\newblock Zenodo, 2026.
\newblock Version v1.0.0 [Software and processed data].
\newblock \href {https://doi.org/10.5281/zenodo.22705028} {\path{doi:10.5281/zenodo.22705028}}.

\bibitem[Oru14]{Ors2014}
Roman Orus.
\newblock A practical introduction to tensor networks: Matrix product states and projected entangled pair states.
\newblock {\em Annals of Physics}, 349:117–158, October 2014.
\newblock URL: \url{http://dx.doi.org/10.1016/j.aop.2014.06.013}, \href {https://doi.org/10.1016/j.aop.2014.06.013} {\path{doi:10.1016/j.aop.2014.06.013}}.

\bibitem[PSSV11]{RevModPhys.83.863}
Anatoli Polkovnikov, Krishnendu Sengupta, Alessandro Silva, and Mukund Vengalattore.
\newblock Colloquium: Nonequilibrium dynamics of closed interacting quantum systems.
\newblock {\em Rev. Mod. Phys.}, 83:863--883, Aug 2011.
\newblock URL: \url{https://link.aps.org/doi/10.1103/RevModPhys.83.863}, \href {https://doi.org/10.1103/RevModPhys.83.863} {\path{doi:10.1103/RevModPhys.83.863}}.

\bibitem[Riv20]{rivlin2020chebyshev}
T.J. Rivlin.
\newblock {\em Chebyshev Polynomials}.
\newblock Dover Books on Mathematics. Dover Publications, 2020.
\newblock URL: \url{https://books.google.com.br/books?id=3s0mygEACAAJ}.

\bibitem[SFMD{\etalchar{+}}24]{StilckFranca.2024}
Daniel Stilck~Fran\c{c}a, Liubov~A. Markovich, V.~V. Dobrovitski, Albert~H. Werner, and Johannes Borregaard.
\newblock Efficient and robust estimation of many-qubit {H}amiltonians.
\newblock {\em Nature Communications}, 15(1), 2024.
\newblock \href {https://doi.org/10.1038/s41467-023-44012-5} {\path{doi:10.1038/s41467-023-44012-5}}.

\bibitem[Var75]{Varah1975}
J.M. Varah.
\newblock A lower bound for the smallest singular value of a matrix.
\newblock {\em Linear Algebra and its Applications}, 11(1):3–5, 1975.
\newblock URL: \url{http://dx.doi.org/10.1016/0024-3795(75)90112-3}, \href {https://doi.org/10.1016/0024-3795(75)90112-3} {\path{doi:10.1016/0024-3795(75)90112-3}}.

\bibitem[Ver26]{Vershynin.2026}
Roman Vershynin.
\newblock {\em High-Dimensional Probability: An Introduction with Applications in Data Science}.
\newblock Cambridge University Press, 2026.
\newblock \href {https://doi.org/10.1017/9781009490672} {\path{doi:10.1017/9781009490672}}.

\bibitem[VKL99]{Viola1999}
Lorenza Viola, Emanuel Knill, and Seth Lloyd.
\newblock Dynamical decoupling of open quantum systems.
\newblock {\em Physical Review Letters}, 82(12):2417–2421, March 1999.
\newblock URL: \url{http://dx.doi.org/10.1103/PhysRevLett.82.2417}, \href {https://doi.org/10.1103/physrevlett.82.2417} {\path{doi:10.1103/physrevlett.82.2417}}.

\bibitem[Wat18]{watrous}
John Watrous.
\newblock {\em The Theory of Quantum Information}.
\newblock Cambridge University Press, USA, 1st edition, 2018.

\bibitem[WGFC14]{PhysRevLett.112.190501}
Nathan Wiebe, Christopher Granade, Christopher Ferrie, and D.~G. Cory.
\newblock Hamiltonian learning and certification using quantum resources.
\newblock {\em Phys. Rev. Lett.}, 112:190501, May 2014.
\newblock URL: \url{https://link.aps.org/doi/10.1103/PhysRevLett.112.190501}, \href {https://doi.org/10.1103/PhysRevLett.112.190501} {\path{doi:10.1103/PhysRevLett.112.190501}}.

\bibitem[WZAD24]{PhysRevApplied.22.054065}
Filip Wudarski, Yaxing Zhang, Juan Atalaya, and M.I. Dykman.
\newblock Revealing inadvertent periodic modulation of qubit frequency.
\newblock {\em Phys. Rev. Appl.}, 22:054065, Nov 2024.
\newblock URL: \url{https://link.aps.org/doi/10.1103/PhysRevApplied.22.054065}, \href {https://doi.org/10.1103/PhysRevApplied.22.054065} {\path{doi:10.1103/PhysRevApplied.22.054065}}.

\bibitem[ZYLB21]{zubida2021optimalshorttimemeasurementshamiltonian}
Assaf Zubida, Elad Yitzhaki, Netanel~H. Lindner, and Eyal Bairey.
\newblock Optimal short-time measurements for {H}amiltonian learning, 2021.
\newblock URL: \url{https://arxiv.org/abs/2108.08824}, \href {https://arxiv.org/abs/2108.08824} {\path{arXiv:2108.08824}}.

\end{thebibliography}
\appendix
\section{Smooth schedules and approximate MSFS models}\label{app:MSFS}

This appendix proves the approximate-schedule extension used in
Theorem~\ref{thm:main}.  The promise is deliberately stated in the
simultaneous $C^1$ form of Def.~\ref{def:C1-MSFS-approximation}: the same
low-degree witness approximates a coefficient and its derivative.  The
learner needs only the certified modulus $\rho_m$, not the witnesses.

\subsection{From coefficient bias to transfer-function bias}

Let $\mathcal S_R$ and $\widetilde{\mathcal S}_R$ be two bounded generators
on the same finite localization region and define
\begin{align}
 \Delta_R&:=\sup_{t\in[0,T]}
 \|\mathcal S_R(t)-\widetilde{\mathcal S}_R(t)\|_{\infty\to\infty},\\
 M_R&:=\sup_{t\in[0,T]}\max\{\|\mathcal S_R(t)\|_{\infty\to\infty},
 \|\widetilde{\mathcal S}_R(t)\|_{\infty\to\infty}\}.
\end{align}
Neither generator is required to be contractive.  The finite-dimensional
Duhamel formula and $\|T(s,t)\|_{\infty\to\infty}\le e^{M_R(t-s)}$ give
\begin{align}
 \|T_{\mathcal S_R}(0,t)-T_{\widetilde{\mathcal S}_R}(0,t)\|_{
 \infty\to\infty}
 \le t\Delta_Re^{M_Rt}.                                         \label{eq:approx-generator-duhamel}
\end{align}
Consequently, for $\|O\|\le1$,
\begin{align}
 &\|\mathcal S_R(t)T_{\mathcal S_R}(0,t)(O)
 -\widetilde{\mathcal S}_R(t)
  T_{\widetilde{\mathcal S}_R}(0,t)(O)\|                       \nonumber\\
 &\quad\le
 \|\mathcal S_R(t)-\widetilde{\mathcal S}_R(t)\|
 \|T_{\mathcal S_R}(0,t)(O)\|
 +\|\widetilde{\mathcal S}_R(t)\|
 \|(T_{\mathcal S_R}-T_{\widetilde{\mathcal S}_R})(0,t)(O)\|
 \nonumber\\
 &\quad\le \Delta_Re^{M_Rt}(1+M_Rt).                           \label{eq:approx-generator-derivative}
\end{align}

For the localized physical generator in Cor.~\ref{cor:deg_approx}, replace
each coefficient by a witness from Eq.~\eqref{eq:C1-MSFS-modulus}.  Because
the physical Hamiltonian and Pauli-dephasing directions
$\mathcal G_\beta\in\mathfrak G_{\rm phys}$ have
$\|\mathcal G_\beta\|_{\infty\to\infty}\le1$,
\begin{align}
 \Delta_R\le J_R\rho_m,
 \qquad M_R\le\overline M_R=(2+\tau)J_R.                        \label{eq:coefficient-to-generator-bias}
\end{align}
The polynomial rate witnesses may be signed, so the auxiliary generator need
not be physical.  This is why Lemma~\ref{lem:dyson_series_conv} and the two
estimates above retain the harmless factor $e^{\overline M_RT}$ rather than
invoking channel contractivity.  Combining
Eqs.~\eqref{eq:approx-generator-duhamel}--
\eqref{eq:coefficient-to-generator-bias} yields precisely the comparison
term $\mathfrak b_R\rho_m$ in
Eq.~\eqref{eq:C1-transfer-approximation}.  The remaining terms are the
physical localization and generator-tail budgets from
Proposition~\ref{prop:locality_bounds}, followed by the noncontractive Dyson
tail in Eq.~\eqref{equ:estimate_tail}.

Notice that the Duhamel comparison itself uses only the zeroth-order part of
Eq.~\eqref{eq:C1-MSFS-modulus}: the derivative of a forward overlap is
$\mathcal S(t)T(0,t)$ and contains no $\dot{\mathcal S}(t)$.  The joint
$C^1$ promise is a convenient stronger smoothness assumption ensuring that
one low-degree model also represents the coefficient derivatives stably.

\subsection{Selecting the effective degree}

Let $R_{\rm probe}$ be the largest radius of any permitted SDP input region
in Algorithm~\ref{alg:rev_from_shadows}.  It is known in preprocessing and
is independent of $n$.  In what follows, all localization bounds and $J_r$
are evaluated at this worst-case probe radius.  Write $L_0(r)$ and $L_1(r)$
for the corresponding right-hand sides of
Eqs.~\eqref{eq:locality-assumption} and
\eqref{eq:generator-tail-assumption}, respectively, maximized over
$0\le s\le t\le T$ and unit-norm probes.

For integers $q\ge1$, set $\eta_q=2^{-q}$ and let $r_q$ be the least integer
$r\ge r_0$ satisfying
\begin{align}
 L_1(r)\le\eta_q/3,
 \qquad L_0(r)\le\frac{\eta_q}{3(1+\overline M_r)}.             \label{eq:explicit-localization-choice}
\end{align}
Let $K_q$ be the least integer such that both noncontractive Dyson tails obey
\begin{align}
 e^{\overline M_{r_q}T}
 \frac{(\overline M_{r_q}T)^{K_q+1}}{(K_q+1)!}&\le\eta_q/3,\nonumber\\
 \overline M_{r_q}e^{\overline M_{r_q}T}
 \frac{(\overline M_{r_q}T)^{K_q}}{K_q!}&\le\eta_q/3.          \label{eq:explicit-Dyson-choice}
\end{align}
We use known monotone upper envelopes for the MSFS constants and for $G$.
For a candidate coefficient degree $m$, define
\begin{align}
 C_{\rm p}(m)&:=C_{\rm int}([0,T],m),\nonumber\\
 d_{m,q}&:=G(m,K_q),\nonumber\\
 C_{\rm aux}(m,q)&:=C_{\rm int}([0,T],d_{m,q})
 C_{\rm der}([0,T],d_{m,q}),\nonumber\\
 \epsilon_{\rm lin}(m)&:=\min\!\left\{
 \frac{\epsilon}{3C_{\rm stab}C_{\rm p}(m)},
 \frac{1}{10s}\right\}.                                      \label{eq:explicit-degree-quantities}
\end{align}
All these quantities are computable from the promises.  Enumerate integer
pairs $(m,q)$ in increasing order of $m+q$ (breaking ties
lexicographically), and let $(m_\epsilon,q_\epsilon)$ be the first pair for
which
\begin{align}
 &\rho_m\le1,\qquad
 (1+C_{\rm p}(m))\rho_m\le\epsilon/2,\nonumber\\
 &\bigl(1+C_{\rm aux}(m,q)\bigr)
 \bigl(\eta_q+\mathfrak b_{r_q}\rho_m\bigr)
 \le\epsilon_{\rm lin}(m)/2.                                  \label{eq:approx-working-degree}
\end{align}
Finally set
\begin{align}
 \eta_\epsilon:=\eta_{q_\epsilon},\qquad
 d_\epsilon:=d_{m_\epsilon,q_\epsilon},\qquad
 m_{\rm eff}:=\max\{m_\epsilon,d_\epsilon,q_\epsilon\}.       \label{eq:effective-complexity-degree}
\end{align}
This is a deterministic enumeration, not an assumed monotonicity property.
For a general modulus the theorem is conditional on the existence of a
feasible pair; the exponential modulus below guarantees one.
The last inequality is exactly the auxiliary derivative-bias budget
Eq.~\eqref{eq:auxiliary-model-bias-condition}; the preceding one is the
primary coefficient-bias budget
Eq.~\eqref{eq:primary-model-bias-condition}.

To verify the latter directly, let $e_{\rm node}$ be the uniform error of the
solved working coefficients at the primary nodes.  The algorithm first
converts dissipative values by $\Gamma$, so physical-rate node errors are at
most $3e_{\rm node}/2$.  For a true physical coefficient $\theta$ and its
witness $p_{\theta,m}$, stable interpolation gives
\begin{align}
 \|\widehat\theta-\theta\|_\infty
 \le \rho_m+C_{\rm p}(m)
       \bigl(\rho_m+\tfrac32e_{\rm node}\bigr).                 \label{eq:approx-primary-interpolation}
\end{align}
Thus Eqs.~\eqref{eq:primary-model-bias-condition} and
\eqref{eq:protocol-linear-precision} imply the target error $\epsilon$.
For Hamiltonian coefficients the factor $3/2$ is one.  This also explains why
the true physical rates, rather than the transformed working rates, are
interpolated in Algorithm~\ref{alg:solve_and_interpolate}.

Under the exponentially stable modulus
Eq.~\eqref{eq:exp-C1-modulus}, a feasible pair exists and
\begin{align}
 m_{\rm eff}=\poly{m_0,\log(\epsilon^{-1})}.                   \label{eq:effective-degree-polylog}
\end{align}
Indeed, $r_q=\mathcal O(\log(\eta_q^{-1}))$,
$J_{r_q}$ and $\overline M_{r_q}$ are polynomial in that radius, and hence
$\log\mathfrak b_{r_q}=\polylog(\eta_q^{-1})$.  The MSFS constants are
polynomial in their degree.  A stretched-exponentially decreasing $\rho_m$
therefore dominates every model-bias amplification in
Eq.~\eqref{eq:approx-working-degree} at a degree polynomial in
$m_0$ and $\log(\epsilon^{-1})$.  The auxiliary degree, both interpolation
grids, and all precision inflations have the same scaling.  Substitution in
the exact statistical analysis gives the sample and time bounds of
Theorem~\ref{thm:main}.

\subsection{A concrete smoothness condition}

For polynomial MSFS, a standard sufficient condition for
Eq.~\eqref{eq:exp-C1-modulus} is that every schedule extend analytically to a
common complex neighbourhood of $[0,T]$ and obey a common bound there.
Chebyshev truncation on a Bernstein ellipse of parameter $\varrho>1$ then
gives, for suitable uniform $C>0$,
\begin{align}
 \|\theta-p_m\|_{\infty,[0,T]}
 +T\|\theta'-p_m'\|_{\infty,[0,T]}
 \le C m^2\varrho^{-m}.                                        \label{eq:analytic-C1-rate}
\end{align}
The polynomial factor is absorbed by a slightly smaller exponential rate.
This covers, on every fixed interval with a common analytic bound,
polynomials, exponentials, Gaussians, trigonometric schedules of bounded
bandwidth, and analytic pulse shapes.  Mere smoothness without a quantitative
modulus is not enough to infer a logarithmic degree: finitely differentiable
schedules generally yield algebraic $C^1$ approximation rates.  They are
covered whenever their certified decay is strong enough to make the set in
Eq.~\eqref{eq:approx-working-degree} nonempty; the resulting $m_{\rm eff}$ is
then determined by that algebraic modulus rather than
Eq.~\eqref{eq:effective-degree-polylog}.  See, e.g., the approximation
results reviewed in~\cite{Borwein1995}.

\subsection{Extrapolation bounds}
To show the extrapolation bounds in Sec.~\ref{sec:extrapolation}, we resort to the following result about Chebyshev polynomials:
\begin{lemma}[Thm.~2.20 of~\cite{rivlin2020chebyshev}]\label{lem:cheby-extremal}
Let $p\in\cP_m$ be a polynomial of degree $m$ s.t. $\|p\|_{L_\infty,[-1,1]}\leq 1$. Then for $|t|>1$ we have:
\begin{align}
    |p(t)|\leq |T_m(t)|\le\bigl(\sqrt t+\sqrt{t-1}\bigr)^{4m}
 \le(4t)^{2m}. ,
\end{align}
where $T_m$ is the Chebyshev polynomial of degree $m$.
\end{lemma}
\section{Stability of the sparse linear solve}\label{app:stab_infty}

We prove Proposition~\ref{prop:stable_systems}.  Its row diagonal-dominance
margin is at least $1/2$, so Varah's bound~\cite{Varah1975} gives
\begin{align}
 \|A^{-1}\|_{\ell_\infty\to\ell_\infty}\le2.
\end{align}
Because $B$ has at most $s$ nonzero entries per row and each satisfies
$|B_{ij}|\le\zeta$,
\begin{align}
 \|B\|_{\ell_\infty\to\ell_\infty}\le s\zeta.
\end{align}
For $\zeta\le1/(10s)$, the Neumann series in the induced
$\ell_\infty$ norm therefore yields
\begin{align}
 \|(A+B)^{-1}\|_{\ell_\infty\to\ell_\infty}
 \le\frac{\|A^{-1}\|_{\ell_\infty\to\ell_\infty}}
 {1-\|A^{-1}B\|_{\ell_\infty\to\ell_\infty}}
 \le\frac{2}{1-2s\zeta}\le\frac52.
\end{align}
Let $x=A^{-1}b$ and $\widehat x=(A+B)^{-1}(b+y)$.  Then
\begin{align}
 \widehat x-x=(A+B)^{-1}(y-Bx),
\end{align}
and $\|x\|_\infty\le2\|b\|_\infty=\cO(1)$.  Consequently,
\begin{align}
 \|\widehat x-x\|_\infty
 \le\frac52\bigl(\zeta+s\zeta\|x\|_\infty\bigr)
 =\cO(s\zeta),
\end{align}
as claimed.  Column diagonal dominance is used in the manuscript to retain
the same guarantee after assembling all sparse rows and to control the
transpose-oriented estimates; the displayed $\ell_\infty$ perturbation bound
itself needs only the row margin.

\section{Stable Pauli sets for the working basis}\label{app:stable_pauli}

\begin{prop}[Hamiltonian stable pairs]\label{prop:size_stable}
Let the Hamiltonian working terms be
$\{\mathcal K^H_\alpha=-i\mathcal P_\alpha\}_{\alpha\in\mathcal A}$ on a
graph of effective dimension $D$.  There are normalized Pauli pairs
$(P_{\alpha,1},P_{\alpha,2})$ satisfying
Def.~\ref{def:stable_paulis}, with $P_{\alpha,1}$ one-local,
$P_{\alpha,2}$ supported on the support of $P_\alpha$, and
\begin{align}
 s=\cO(b_k4^{b_k}).
\end{align}
\end{prop}
\begin{proof}
Choose a one-qubit Hermitian Pauli $P_{\alpha,1}$ on the support of
$P_\alpha$ that anticommutes with its local factor, and define the signed
Hermitian Pauli
\begin{align}
 P_{\alpha,2}:=\mathcal K^H_\alpha(P_{\alpha,1}).
\end{align}
Then $2^{-n}\tr[\mathcal K^H_\alpha(P_{\alpha,1})P_{\alpha,2}]=1$.
For a distinct Pauli string $P_{\alpha'}$, orthogonality gives
$\tr[\mathcal K^H_{\alpha'}(P_{\alpha,1})P_{\alpha,2}]=0$.
Each output has support in a $k$-diameter region.  For a fixed output, only
terms intersecting one of its at most $\cO(k^D)$ sites can act nontrivially,
and at each site there are at most $4^{b_k}$ candidate local Pauli
strings.  The same count with rows and columns interchanged gives the stated
sparsity.
\end{proof}

\begin{prop}[Stable Pauli-dephasing basis]\label{prop:size_stabledissip}
Adjoin the algebraic terms
$\{\widetilde{\mathcal L}_{j,P}\}_{j,P}$ from
Eq.~\eqref{eq:tilde-L-definition} and the pairs $(P_j,P_j)$.  Together with
the Hamiltonian pairs of Prop.~\ref{prop:size_stable}, the ideal working-basis
matrix is the identity and has at most
$s=\cO(b_k4^{b_k})$ nonzero entries in every row and column.
\end{prop}
\begin{proof}
On Pauli strings,
\begin{align}
 \widetilde{\mathcal L}_{j,P}(P'_k)
 =\delta_{jk}\delta_{PP'}P_j.
\end{align}
Thus the dissipative block is the identity.  A transformed onsite dissipator
maps $P_{\alpha,1}$ either to itself or to zero, while
$P_{\alpha,2}\ne P_{\alpha,1}$, so the ideal dissipative-to-Hamiltonian
cross entries vanish.  Likewise
$2^{-n}\tr[\mathcal K^H_\alpha(P_j)P_j]=0$.  Finally, the Hamiltonian
support count is that of Prop.~\ref{prop:size_stable}, and each onsite
projector contributes only constant additional multiplicity.
\end{proof}

\section{Lieb--Robinson bounds for time-dependent dissipation}
\label{app:LR}

This appendix first proves a Lieb--Robinson theorem for arbitrary finite
local CP-unital evolutions of uniformly bounded labelled interaction degree;
no growth condition on a physical graph is used.  Finite-range and polynomial-growth
estimates are then obtained as corollaries, and the latter supplies the three
locality bounds used in the main text.  The proof is entirely in the
Heisenberg picture and starts from the physical right-generated cocycle.
Only afterward, in Corollary~\ref{cor:route2-LR}, do we apply it to the
learning channel, using the fact that every evolution in our restricted
model is bistochastic.  This separation is important: bistochasticity does
not change the ordering of the physical Heisenberg equation.

\subsection{Local unital evolutions and causal paths}

Let $V$ be a finite set of sites and let $\Gamma$ be a finite labelled family of
nonempty interaction supports $\gamma\subseteq V$.  Different labels may
have the same support.  No volume-growth assumption is imposed.  We only
assume that the interaction-overlap graph has uniformly bounded degree,
and write
\begin{align}
 z:=\sup_{\gamma\in\Gamma}
 |\{\delta\in\Gamma:\delta\cap\gamma\ne\varnothing\}|<\infty . \label{eq:LR-multiplicity}
\end{align}
Thus, if the usual interaction-graph degree excludes the vertex itself and
is bounded by $\Delta_I$, then $z\le\Delta_I+1$.
For a family of growing systems, $z$ is required to be bounded uniformly in
the system size.  This is the bounded-interaction-degree hypothesis; the
underlying physical graph, when one is specified later, may have arbitrary
volume growth.
Consider a time-dependent Heisenberg generator
\begin{align}
 \mathcal L(t)=\sum_{\gamma\in\Gamma}\mathcal L_\gamma(t),\qquad
 \mathcal L_\gamma(t)(I)=0,\qquad
 \sup_{\gamma,t}\|\mathcal L_\gamma(t)\|_{\mathrm{cb},\infty\to\infty}
 \le J.                                                          \label{eq:LR-local-generator}
\end{align}
Here $\|\cdot\|_{\mathrm{cb},\infty\to\infty}$ is the completely bounded
induced operator norm, and each $\mathcal L_\gamma(t)$ is supported on
$\gamma$.  We require that, for every subfamily
$\mathcal E\subseteq\Gamma$, the sum
$\sum_{\gamma\in\mathcal E}\mathcal L_\gamma(t)$ generate a completely
positive unital propagator.  This is automatic when the local terms are
Heisenberg GKSL terms with nonnegative dissipative rates.  Consequently, all
full and subfamily propagators used below are contractions in operator norm.

The physical Heisenberg cocycle is
\begin{align}
 \partial_t\tau_{\mathcal L}(s,t)
 &=\tau_{\mathcal L}(s,t)\circ\mathcal L(t),&
 \tau_{\mathcal L}(s,s)&=\mathcal I,                             \label{eq:heis-right}
\end{align}
and hence
$\partial_s\tau_{\mathcal L}(s,t)
=-\mathcal L(s)\circ\tau_{\mathcal L}(s,t)$.
For $A\subseteq V$, put
\begin{align}
 \Gamma_A:=\{\gamma\in\Gamma:\gamma\cap A\ne\varnothing\},
 \qquad a_A:=|\Gamma_A|.
\end{align}
Because all labels containing a fixed site form a clique in the labelled
interaction graph, $a_A\le z|A|$.  Thus the prefactor is local even when the
underlying physical graph has arbitrary volume growth.
For $F\subseteq\Gamma$, define the interaction-graph distance
$d_I(A,F)$ as the least $\ell\ge0$ for which there are
$\gamma_0\in F$ and $\gamma_1,\ldots,\gamma_\ell\in\Gamma$ such that
\begin{align}
 \gamma_{j-1}\cap\gamma_j\ne\varnothing\quad(1\le j\le\ell),
 \qquad \gamma_\ell\cap A\ne\varnothing.                         \label{eq:interaction-distance}
\end{align}
The value is $+\infty$ when no such path exists.  In particular,
$d_I(A,F)=0$ when some member of $F$ meets $A$.

\begin{lemma}[Causal-path estimate]\label{lem:causal-path}
Let $X$ be supported on $A$.  For every $\gamma\in F$, let
$\mathcal K_\gamma$ be supported on $\gamma$, satisfy
$\mathcal K_\gamma(I)=0$, and obey
$\|\mathcal K_\gamma\|_{\mathrm{cb},\infty\to\infty}\le\kappa$.
Then
\begin{align}
 \left\|\sum_{\gamma\in F}\mathcal K_\gamma
       \tau_{\mathcal L}(s,t)(X)\right\|_\infty
 \le a_A\kappa\|X\|_\infty
 \sum_{\ell\ge d_I(A,F)}
 \frac{[zJ(t-s)]^\ell}{\ell!}.                                  \label{eq:causal-path}
\end{align}
The same estimate holds when the evolution is generated by any subfamily of
$\Gamma$.
\end{lemma}

\begin{proof}
Fix $\gamma$ and a time-independent endpoint map $\mathcal K_\gamma$, and
set $f_\gamma(u)=\mathcal K_\gamma
\tau_{\mathcal L}(u,t)(X)$.  Split off the terms disjoint from $\gamma$:
\begin{align}
 \mathcal L_{\perp\gamma}(u)
 :=\sum_{\delta:\delta\cap\gamma=\varnothing}\mathcal L_\delta(u).
\end{align}
Disjoint local maps commute.  Differentiating in the lower endpoint of the
right-generated cocycle therefore gives
\begin{align}
 \partial_u f_\gamma(u)
 =-\mathcal L_{\perp\gamma}(u)f_\gamma(u)
 -\mathcal K_\gamma
 \sum_{\delta:\delta\cap\gamma\ne\varnothing}
 \mathcal L_\delta(u)\tau_{\mathcal L}(u,t)(X).
\end{align}
Variation of constants and contractivity of the propagator generated by
$\mathcal L_{\perp\gamma}$ imply
\begin{align}
 \|\mathcal K_\gamma\tau_{\mathcal L}(u,t)(X)\|_\infty
 &\le \|\mathcal K_\gamma\|\|X\|_\infty
       \mathbf 1_{\gamma\cap A\ne\varnothing}\nonumber\\
 &\quad+\|\mathcal K_\gamma\|
 \sum_{\delta:\delta\cap\gamma\ne\varnothing}\int_u^t
 \|\mathcal L_\delta(v)\tau_{\mathcal L}(v,t)(X)\|_\infty\,dv .
                                                                    \label{eq:path-recursion}
\end{align}
Iterating Eq.~\eqref{eq:path-recursion} produces only overlap paths.  A path
with $\ell$ generator insertions contributes at most
$\kappa J^\ell\|X\|_\infty(t-s)^\ell/\ell!$.  After summing its endpoint
over $F$, the number of length-$\ell$ paths is at most $a_Az^\ell$: choose
the last support meeting $A$ in at most $a_A$ ways and each preceding
overlapping support in at most $z$ ways.  The iterative remainder vanishes
because, on the finite graph, it is bounded by
$|F|\kappa\|X\|_\infty[zJ(t-s)]^L/L!\to0$ as $L\to\infty$.
This proves Eq.~\eqref{eq:causal-path}.  This Picard/path recursion is the
time-dependent open-system analogue of the standard LR argument; see also
Refs.~\cite{BarthelKliesch2012,Cubitt2015,chen2021operator,chen2023speed}.
\end{proof}

\subsection{Growth-free interaction-graph bounds}

For a subfamily $\Lambda\subseteq\Gamma$, write
\begin{align}
 \mathcal L_\Lambda(t):=\sum_{\gamma\in\Lambda}\mathcal L_\gamma(t),
 \qquad F:=\Gamma\setminus\Lambda,
 \qquad R:=d_I(A,F).                                             \label{eq:general-truncation}
\end{align}
Thus $R$ measures the distance, in the interaction-overlap graph, from the
initial support to the terms omitted from the truncated generator.  We use
the convention that a sum whose lower limit is $+\infty$ vanishes.

We will also use the elementary tilted-Poisson bound: for every $\alpha>0$,
$R\ge1$, and $x\ge0$,
\begin{align}
 \sum_{\ell\ge R}\frac{x^\ell}{\ell!}
 \le e^{-\alpha R}\bigl(e^{e^\alpha x}-1\bigr).                 \label{eq:tilted-poisson-tail}
\end{align}
It follows by multiplying every term with
$1\le e^{\alpha(\ell-R)}$.  This replaces the invalid factorial
reindexing in the earlier argument.

\begin{thm}[Heisenberg Lieb--Robinson bounds on general interaction graphs]
\label{thm:local-unital-LR}
Assume Eqs.~\eqref{eq:LR-multiplicity} and
\eqref{eq:LR-local-generator}, together with CP-unitality of the subfamily
propagators stated above.  Let $X$ be supported on $A$, put
$\omega:=zJ$ and $\Theta:=t-s$, and let $\Lambda,F,R$ be as in
Eq.~\eqref{eq:general-truncation}.  Then
\begin{align}
 \|\tau_{\mathcal L}(s,t)(X)-
       \tau_{\mathcal L_\Lambda}(s,t)(X)\|_\infty
 &\le a_AJ\Theta\|X\|_\infty
 \sum_{\ell\ge R}\frac{(\omega\Theta)^\ell}{(\ell+1)!},
                                                                    \label{eq:general-heis-truncation}\\
 \|[\mathcal L(t)-\mathcal L_\Lambda(t)]
       \tau_{\mathcal L}(s,t)(X)\|_\infty
 &\le a_AJ\|X\|_\infty
 \sum_{\ell\ge R}\frac{(\omega\Theta)^\ell}{\ell!}.
                                                                    \label{eq:general-heis-generator-tail}
\end{align}
In particular, for every $\alpha>0$ and $R\ge1$, with
$v_\alpha:=e^\alpha\omega$,
\begin{align}
 \|\tau_{\mathcal L}(s,t)(X)-
       \tau_{\mathcal L_\Lambda}(s,t)(X)\|_\infty
 &\le \frac{a_AJ}{v_\alpha}e^{-\alpha R}
       (e^{v_\alpha\Theta}-1)\|X\|_\infty,                     \label{eq:general-heis-truncation-exp}\\
 \|[\mathcal L(t)-\mathcal L_\Lambda(t)]
       \tau_{\mathcal L}(s,t)(X)\|_\infty
 &\le a_AJ e^{-\alpha R}
       (e^{v_\alpha\Theta}-1)\|X\|_\infty.                    \label{eq:general-heis-generator-tail-exp}
\end{align}

Moreover, let
$\mathcal L_i(t)=\sum_{\gamma\in\Gamma}\mathcal L_{i,\gamma}(t)$,
$i=1,2$, satisfy the same hypotheses and bounds $J,z$, padding either
decomposition by zero terms if necessary.  If
\begin{align}
 \sup_{\gamma,t}\|\mathcal L_{1,\gamma}(t)-
 \mathcal L_{2,\gamma}(t)\|_{\mathrm{cb},\infty\to\infty}\le\eta,
\end{align}
then, with
$g_\omega(q):=(e^{\omega q}-1)/\omega$ and $g_0(q):=q$,
\begin{align}
 \|\tau_{\mathcal L_1}(s,t)(X)-
       \tau_{\mathcal L_2}(s,t)(X)\|_\infty
 \le a_A\eta g_\omega(\Theta)\|X\|_\infty .                  \label{eq:general-heis-comparison}
\end{align}
More generally, suppose $\Gamma$ is partitioned into endpoint families
$\{F_p\}_{p\in\mathcal Q}$, put $R_p=d_I(A,F_p)$, and assume
\begin{align}
 \sup_{\gamma\in F_p,t}\|\mathcal L_{1,\gamma}(t)-
 \mathcal L_{2,\gamma}(t)\|_{\mathrm{cb},\infty\to\infty}
 \le\eta_p.
\end{align}
Then the distance-resolved comparison estimate is
\begin{align}
 \|\tau_{\mathcal L_1}(s,t)(X)-
       \tau_{\mathcal L_2}(s,t)(X)\|_\infty
 \le a_A\Theta\|X\|_\infty
 \sum_{p\in\mathcal Q}\eta_p
 \sum_{\ell\ge R_p}\frac{(\omega\Theta)^\ell}{(\ell+1)!}.
                                                                    \label{eq:general-heis-comparison-resolved}
\end{align}
No volume-growth assumption on a physical graph, and no separation between
the perturbation and $A$, is required for the uniform bound
Eq.~\eqref{eq:general-heis-comparison}.
All expressions at $\omega=zJ=0$ are understood by continuous extension;
the exact Poisson bounds above are unambiguous in this case.
\end{thm}

\begin{proof}
The right-generated Duhamel identity is
\begin{align}
 \tau_1(s,t)-\tau_2(s,t)
 =\int_s^t\tau_1(s,u)\circ[\mathcal L_1(u)-\mathcal L_2(u)]
             \circ\tau_2(u,t)\,du .                              \label{eq:right-duhamel}
\end{align}
For truncation take $\mathcal L_1=\mathcal L$ and
$\mathcal L_2=\mathcal L_\Lambda$.  The outer propagator is contractive.
For each fixed $u$, Lemma~\ref{lem:causal-path}, applied to the inner
truncated evolution with endpoint family $F$ and endpoint maps
$\mathcal K_\gamma=\mathcal L_\gamma(u)$, gives, with $q=t-u$,
\begin{align}
 \|[\mathcal L(u)-\mathcal L_\Lambda(u)]
       \tau_{\mathcal L_\Lambda}(u,t)(X)\|_\infty
 \le a_AJ\|X\|_\infty
 \sum_{\ell\ge R}\frac{(\omega q)^\ell}{\ell!}.              \label{eq:general-inner-tail}
\end{align}
Integrating this bound proves Eq.~\eqref{eq:general-heis-truncation}.
Equation~\eqref{eq:tilted-poisson-tail} bounds the integrand by
$a_AJ e^{-\alpha R}(e^{v_\alpha q}-1)\|X\|_\infty$.
Its integral is
$a_AJ e^{-\alpha R}[(e^{v_\alpha\Theta}-1)/v_\alpha-\Theta]$
times $\|X\|_\infty$, which proves
Eq.~\eqref{eq:general-heis-truncation-exp}.

For Eqs.~\eqref{eq:general-heis-generator-tail} and
\eqref{eq:general-heis-generator-tail-exp}, apply
Lemma~\ref{lem:causal-path} directly to the full evolution with endpoint
family $F$ and the fixed maps
$\mathcal K_\gamma=\mathcal L_\gamma(t)$, and then use
Eq.~\eqref{eq:tilted-poisson-tail}.  The endpoint maps in the lemma need not
be terms of the evolution at its terminal time; this will be used below.

Finally use Eq.~\eqref{eq:right-duhamel}, contractivity of
$\tau_1(s,u)$, and Lemma~\ref{lem:causal-path} for the $\mathcal L_2$
evolution.  Apply the lemma separately to each endpoint family $F_p$, with
$\mathcal K_\gamma=\mathcal L_{1,\gamma}(u)-
\mathcal L_{2,\gamma}(u)$ and $\kappa=\eta_p$, and sum over $p$.
Integration gives Eq.~\eqref{eq:general-heis-comparison-resolved}.
Taking the single family $F=\Gamma$, for which $R=0$, yields
\begin{align}
 \|[\mathcal L_1(u)-\mathcal L_2(u)]
       \tau_{\mathcal L_2}(u,t)(X)\|_\infty
 \le a_A\eta e^{\omega(t-u)}\|X\|_\infty,
\end{align}
and proves Eq.~\eqref{eq:general-heis-comparison}.
\end{proof}

\subsection{Finite-range and polynomial-growth specializations}

We first convert interaction distance into physical graph distance without
assuming anything about volume growth.  Let $\mathsf G=(V,E)$ be any finite
graph, let $k\ge1$, and suppose that
\begin{align}
 \operatorname{diam}_{\mathsf G}(\gamma)\le k
 \qquad(\gamma\in\Gamma),                                      \label{eq:LR-range}
\end{align}
and, for $A(r):=\{x:d_{\mathsf G}(x,A)\le r\}$, use the
support-contained truncation
\begin{align}
 \mathcal L_{A(r)}(t)
 :=\sum_{\gamma:\gamma\subseteq A(r)}\mathcal L_\gamma(t),\qquad
 F_r:=\{\gamma:\gamma\not\subseteq A(r)\}.                    \label{eq:consistent-truncation}
\end{align}
If $r\ge k$ is integral, then
\begin{align}
 d_I(A,F_r)\ge R(r):=\left\lfloor\frac r k\right\rfloor.       \label{eq:metric-conversion}
\end{align}
Indeed, an overlap path $(\gamma_0,\ldots,\gamma_\ell)$ ending at $A$
can reach only vertices at graph distance at most $(\ell+1)k$ from $A$.
If $\gamma_0\not\subseteq A(r)$, it contains a vertex at distance at least
$r+1$, so
$\ell\ge\lceil(r+1)/k\rceil-1=\lfloor r/k\rfloor$.

\begin{corollary}[Finite-range geometric truncation]
\label{cor:finite-range-LR}
Assume the hypotheses of Theorem~\ref{thm:local-unital-LR} and
Eq.~\eqref{eq:LR-range}.  Let $X$ be supported on $A$.  For every
$\alpha>0$ and integral $r\ge k$, with
$v_\alpha=e^\alpha zJ$ and $\mu_\alpha=\alpha/k$,
\begin{align}
 \|\tau_{\mathcal L}(s,t)(X)-
       \tau_{\mathcal L_{A(r)}}(s,t)(X)\|_\infty
 &\le \frac{a_AJ}{v_\alpha}
 e^{-\mu_\alpha(r-k)}(e^{v_\alpha(t-s)}-1)\|X\|_\infty,
                                                                    \label{eq:finite-range-heis-truncation}\\
 \|[\mathcal L(t)-\mathcal L_{A(r)}(t)]
       \tau_{\mathcal L}(s,t)(X)\|_\infty
 &\le a_AJ e^{-\mu_\alpha(r-k)}
       (e^{v_\alpha(t-s)}-1)\|X\|_\infty .                     \label{eq:finite-range-heis-tail}
\end{align}
No volume-growth condition on $\mathsf G$ is used.
\end{corollary}

\begin{proof}
Combine Eq.~\eqref{eq:metric-conversion} with
Eqs.~\eqref{eq:general-heis-truncation-exp} and
\eqref{eq:general-heis-generator-tail-exp}, using
$e^{-\alpha\lfloor r/k\rfloor}\le e^{-\mu_\alpha(r-k)}$.
\end{proof}

We now impose the additional geometric hypothesis used by the learning
theorem:
\begin{align}
 |B_r(x)|\le C_1(1+r)^D\qquad(x\in V,\ r\ge0),\qquad D\ge1.      \label{eq:LR-polynomial-growth}
\end{align}
Define
\begin{align}
 \nu&:=\sup_{x\in V}|\{\gamma\in\Gamma:x\in\gamma\}|,\qquad
 b_k:=C_1(1+k)^D.
\end{align}
Then $z\le\nu b_k$.

\begin{corollary}[Polynomial-growth specialization]
\label{cor:polynomial-growth-LR}
Assume the hypotheses of Theorem~\ref{thm:local-unital-LR} and
Eqs.~\eqref{eq:LR-range} and \eqref{eq:LR-polynomial-growth}.  Let $X$ be
supported on $A\subseteq B_{r_A}(x_A)$ and define
\begin{align}
 c_*:=\nu C_1,\qquad \omega:=zJ,\qquad
 v_\alpha:=e^\alpha\omega,\qquad
 \mu_\alpha:=\alpha/k,\qquad \Theta:=t-s .
\end{align}
Then, for every integral $r\ge k$,
\begin{align}
 \|\tau_{\mathcal L}(s,t)(X)-
       \tau_{\mathcal L_{A(r)}}(s,t)(X)\|_\infty
 &\le \frac{c_*J}{v_\alpha}(1+r_A)^D
 e^{-\mu_\alpha(r-k)}(e^{v_\alpha\Theta}-1)\|X\|_\infty,
                                                                    \label{eq:heis-truncation}\\
 \|[\mathcal L(t)-\mathcal L_{A(r)}(t)]
       \tau_{\mathcal L}(s,t)(X)\|_\infty
 &\le c_*J(1+r_A)^D
 e^{-\mu_\alpha(r-k)}(e^{v_\alpha\Theta}-1)\|X\|_\infty .
                                                                    \label{eq:heis-generator-tail}
\end{align}
For two dynamics satisfying the comparison hypotheses of
Theorem~\ref{thm:local-unital-LR},
\begin{align}
 \|\tau_{\mathcal L_1}(s,t)(X)-
       \tau_{\mathcal L_2}(s,t)(X)\|_\infty
 \le c_*\eta(1+r_A)^Dg_\omega(\Theta)\|X\|_\infty .           \label{eq:heis-comparison}
\end{align}
\end{corollary}

\begin{proof}
We have $a_A\le\nu|A|\le c_*(1+r_A)^D$.  Moreover,
the first two claims follow from Corollary~\ref{cor:finite-range-LR}, while
the comparison claim follows from Eq.~\eqref{eq:general-heis-comparison}.
\end{proof}

\subsection{Application to the bistochastic learning dynamics}

\begin{corollary}[Bistochastic Hamiltonians and onsite Pauli dephasing]
\label{cor:route2-LR}
For the generator class in Eq.~\eqref{eq:thegenerator}, one may use
\begin{align}
 J=\max\{1,\tau\},\qquad
 b_k=C_1(1+k)^D,\qquad
 \nu\le4^{b_k}+3,\qquad z\le b_k\nu,                             \label{eq:route2-LR-constants}
\end{align}
uniformly in $n$.  The bounds
\eqref{eq:locality-assumption}--\eqref{eq:comparison-assumption} follow from
Corollary~\ref{cor:polynomial-growth-LR}; the same is true for the physical
Heisenberg adjoints.
\end{corollary}

\begin{proof}
Every physical subfamily sum is trace annihilating and unital and generates
a CPTP unital channel.  Moreover,
\begin{align}
 \left\|-\frac i2[P_\alpha,\,\cdot\,]\right\|_
 {\mathrm{cb},\infty\to\infty}\le1,\qquad
 \left\|\frac12(P_j(\,\cdot\,)P_j-\,\cdot\,)\right\|_
 {\mathrm{cb},\infty\to\infty}\le1.                             \label{eq:route2-local-norms}
\end{align}
Any diameter-$k$ Pauli string containing a fixed vertex is supported in its
$k$-ball.  Distinctness of the Pauli basis therefore gives at most
$4^{b_k}$ Hamiltonian labels through that vertex, in addition to the three
onsite dephasing labels.  Equations~\eqref{eq:route2-LR-constants} follow.

It remains to relate the left-generated learning channel to the
right-generated theorem without changing pictures in the LR argument.
Fix a physical interval $[s,t]$, put $\Theta=t-s$, and define
\begin{align}
 \widehat\tau(a,b)&:=T_{\mathcal S}(t-b,t-a),&
 \widehat{\mathcal L}(b)&:=\mathcal S(t-b),
 \qquad0\le a\le b\le\Theta.                                    \label{eq:route2-reversed-clock}
\end{align}
Differentiating the lower endpoint of the forward propagator gives
\begin{align}
 \partial_b\widehat\tau(a,b)
 =\widehat\tau(a,b)\circ\widehat{\mathcal L}(b),\qquad
 \widehat\tau(0,\Theta)=T_{\mathcal S}(s,t).                    \label{eq:route2-as-heisenberg}
\end{align}
Each $\mathcal S(t-b)$ is a valid local Heisenberg GKSL generator: the
Hamiltonian direction is the Heisenberg generator with the opposite
Hamiltonian sign, while Pauli dephasing is self-adjoint.  Thus
Eqs.~\eqref{eq:heis-truncation} and \eqref{eq:heis-comparison} give the
forward-channel truncation and comparison bounds.

For the instantaneous forward tail, apply
Lemma~\ref{lem:causal-path} to $\widehat\tau(0,\Theta)$ with the endpoint
family $\mathcal K_\gamma=\mathcal S_\gamma(t)$.  As noted in the proof of
Theorem~\ref{thm:local-unital-LR}, the endpoint maps need not coincide with
the generator terms at the terminal time.  This remains valid even though
$\widehat{\mathcal L}(\Theta)=\mathcal S(s)$, and gives
\begin{align}
 \|[\mathcal S(t)-\mathcal S_{A(r)}(t)]
       T_{\mathcal S}(s,t)(X)\|_\infty
 \le c_*J(1+r_A)^D e^{-\mu_\alpha(r-k)}
       (e^{v_\alpha(t-s)}-1)\|X\|_\infty.                       \label{eq:route2-forward-generator-tail}
\end{align}
Since $D\ge1$, this is stronger than the prefactor in
Eq.~\eqref{eq:generator-tail-assumption}.  Fixing $\alpha=1$ in the theorem,
take
\begin{align}
 r_0&=k,&
 \mu&=1/k,&
 v&=e zJ,&
 \mu'&=zJ,\nonumber\\
 C_{\rm loc}&=c_*J/v,&
 C_{\rm gen}&=c_*J,&
 C_{\rm cmp}&=c_*/(zJ)                                         \label{eq:main-LR-constant-choice}
\end{align}
proves Proposition~\ref{prop:locality_bounds}; the convention at $zJ=0$ is
continuous.

Finally,
\begin{align}
 \partial_tT_{\mathcal S}(s,t)^*
 =T_{\mathcal S}(s,t)^*\circ\mathcal S(t)^*
\end{align}
is already the physical Heisenberg equation.  Its local terms have the same
supports and completely bounded norms: Hamiltonian signs reverse and Pauli
dephasing is self-adjoint.  Theorem~\ref{thm:local-unital-LR}, followed by
Corollary~\ref{cor:polynomial-growth-LR}, therefore applies directly to
$T_{\mathcal S}^*$ with the same constants.
\end{proof}

\end{document}